\documentclass[aps,prx,twocolumn,superscriptaddress,floatfix]{revtex4-2}
\usepackage{amssymb}
\usepackage{amsmath}
\usepackage[dvipsnames]{xcolor}
\usepackage{color}         
\usepackage{graphics}      
\usepackage[pdftex]{graphicx}      
\usepackage{longtable}     
\usepackage{epsf}          
\usepackage{bm}            
\usepackage{asymptote}     
\usepackage{thumbpdf}
\usepackage{verbatim}
\usepackage{url}
\usepackage{MnSymbol}
\usepackage{physics}
\usepackage{graphicx}
\usepackage[colorlinks=true, hyperindex, breaklinks, linkcolor=blue, urlcolor=blue, citecolor=blue]{hyperref} 
\usepackage{enumitem}	

\usepackage{float}
\usepackage{tipa, extraipa}
\usepackage{color}
\usepackage{verbatim}
\usepackage{soul,xcolor}
\setstcolor{red}
\usepackage{colortbl}
\usepackage{amsthm}
\usepackage[linesnumbered, ruled, vlined]{algorithm2e}
\usepackage{makecell}
\usepackage{stmaryrd}
\usepackage{tikz-cd}

\usepackage{enumitem}

\newcommand{\outline}[1]{\noindent \textbf{#1}\\}

\newtheorem{theorem}{Theorem}
\newtheorem{lemma}[theorem]{Lemma}

\newtheorem{definition}[theorem]{Definition}
\newtheorem{proposition}[theorem]{Proposition}

\newcommand{\mc}{\mathcal}
\newcommand{\mb}{\mathbf}
\newcommand{\mbb}{\mathbb}
\newcommand{\mr}{\mathrm}
\renewcommand{\ker}[1]{\mathrm{ker}(#1)}

\usepackage{tabularx}

\usepackage[normalem]{ulem}

\SetKwInOut{Input}{Input}
\SetKwInOut{Output}{Output}
\SetKwFunction{FMain}{Rearrange}
\SetKwFunction{Fappend}{append}
\SetKwProg{Fn}{Function}{:}{}
\RestyleAlgo{ruled}

\begin{document}

\title{Fast and Parallelizable Logical Computation with Homological Product Codes}

\author{Qian Xu}
\email{qianxu@uchicago.edu}
\affiliation{Pritzker School of Molecular Engineering, The University of Chicago, Chicago 60637, USA}

\author{Hengyun Zhou}
\affiliation{QuEra Computing Inc., 1284 Soldiers Field Road, Boston, MA, 02135, US}
\affiliation{Department of Physics, Harvard University, Cambridge, Massachusetts 02138, USA}

\author{Guo Zheng}
\affiliation{Pritzker School of Molecular Engineering, The University of Chicago, Chicago 60637, USA}

\author{Dolev Bluvstein}
\affiliation{Department of Physics, Harvard University, Cambridge, Massachusetts 02138, USA}

\author{J. Pablo Bonilla Ataides}
\affiliation{Department of Physics, Harvard University, Cambridge, Massachusetts 02138, USA}

\author{Mikhail D. Lukin}
\affiliation{Department of Physics, Harvard University, Cambridge, Massachusetts 02138, USA}

\author{Liang Jiang}
\email{liang.jiang@uchicago.edu}
\affiliation{Pritzker School of Molecular Engineering, The University of Chicago, Chicago 60637, USA}

\begin{abstract}
Quantum error correction is necessary to perform large-scale quantum computation, but requires extremely large overheads in both space and time. High-rate quantum low-density-parity-check (qLDPC) codes promise a route to reduce qubit numbers, but performing computation while maintaining low space cost has required serialization of operations and extra time costs.
In this work, we design fast and parallelizable logical gates for qLDPC codes, and demonstrate their utility for key algorithmic subroutines such as the quantum adder.
Our gate gadgets utilize transversal logical CNOTs between a data qLDPC code and a suitably constructed ancilla code to perform parallel Pauli product measurements (PPMs) on the data logical qubits.
For hypergraph product codes, we show that the ancilla can be constructed by simply modifying the base classical codes of the data code, achieving parallel PPMs on a subgrid of the logical qubits with a lower space-time cost than existing schemes for an important class of circuits.
Generalizations to 3D and 4D homological product codes further feature fast PPMs in constant depth.
While prior work on qLDPC codes has focused on individual logical gates, we initiate the study of fault-tolerant compilation with our expanded set of native qLDPC code operations, constructing algorithmic primitives for preparing $k$-qubit GHZ states and distilling/teleporting $k$ magic states with $O(1)$ space overhead in $O(1)$ and $O(\sqrt{k} \log k)$ logical cycles, respectively.
We further generalize this to key algorithmic subroutines, demonstrating the efficient implementation of quantum adders using parallel operations.
Our constructions are naturally compatible with reconfigurable architectures such as neutral atom arrays, paving the way to large-scale quantum computation with low space and time overheads.
\end{abstract}
\maketitle
\tableofcontents

\section{Introduction \label{sec:introduction}}
Quantum error correction (QEC) is essential for realizing large-scale, fault-tolerant quantum computation. However, paradigmatic QEC schemes based on the surface code are very costly in terms of the space overhead, requiring millions of physical qubits for solving practical problems at useful scale~\cite{fowler2012surface,litinski2019game,beverland2022assessing,gidney2019how}. Recent breakthroughs in high-rate quantum low-density-parity-check (qLDPC) codes, in both asymptotic parameter scaling~\cite{panteleev2022quantum, breuckmann2020balanced, panteleev2022asymptotically, leverrier2022quantum, gu2022efficient,dinur2022good, lin2022good} and practical implementations~\cite{tremblay2022constant, xu2024constant, bravyi2024high, higgott2023constructions, viszlai2023matching, poole2024architecture}, promise a route to significantly reduce the qubit numbers.  In light of recent experimental implementations of various QEC schemes~\cite{google2023suppressing, sivak2022real, bluvstein2023logical, gupta2023encoding, levine2023demonstrating, ma2023high, da2024demonstration}, 
such developments hold promise for 
greatly accelerating the progress toward large-scale error-corrected quantum computers.

Although qLDPC codes constitute hardware-efficient quantum memories, processing the information stored in these codes is generally challenging 
due to the overlapping support of many logical qubits in the same code block. Consequently, logical computations based on these codes have so far involved additional time overhead. For instance, the leading approaches for implementing selective logical operations involve interfacing the qLDPC codes with rateless ancillae (codes with asymptotically vanishing rate), e.g. surface codes, via lattice surgery operations~\cite{cohen2021quantum, bravyi2024high, xu2024constant}. To maintain the low space overhead, only a few ancillae can be used, and consequently, logical computations have to be serialized.
In contrast, since each logical qubit can be independently operated on, logical computations using  more conventional QEC approaches (such as e.g.  surface codes) can be executed in a highly parallel fashion.
Hence, 
computations with qLDPC codes seem to incur a severe space-time tradeoff, with the increased time cost due to serialization possibly negating the space savings. 

An alternative approach for implementing logical operations in high-rate qLDPC codes while offering logical parallelism involves transversal gates~\cite{breuckmann2022fold, quintavalle2022partitioning}. For example, transversal inter-block CNOTs give logical CNOTs between every inter-block pair of encoded logical qubits for any two identical CSS codes~\cite{shor1996fault}. While parallel, these transversal gates are not selective,  acting on \emph{all logical qubits} homogeneously. 

Thus, it is natural to inquire if there exists any approach 
in between lattice surgery (with rateless ancillae) and transversal gates, that offers selectivity and parallelism simultaneously. 
A promising direction involves the so-called homomorphic CNOT~\cite{huang2022homomorphic}, which generalizes the transversal CNOTs between two identical codes to two distinct codes. In this approach, 
using a smaller code as an ancilla, it is possible to perform selective operations in parallel on a subset of logical qubits of a data block. As an example, Ref.~\cite{huang2022homomorphic} shows that it is possible 
to perform a measurement only on one of the logical qubits in a toric code using a surface-code ancilla. Such a homomorphic CNOT, however,  relies on identifying   a nontrivial homomorphism between two quantum codes, which is challenging for general codes. In particular, it remains unclear how to generalize the constructions from the topological codes in Ref.~\cite{huang2022homomorphic} to algebraically constructed qLDPC codes. 

In this work, we construct homomorphic inter-block CNOTs for a family of qLDPC codes -- the so-called homological product codes~\cite{bravyi2013homological, zeng2019higher, campbell2019theory, breuckmann2021quantum} -- that are considered leading candidates for practical fault tolerance~\cite{tremblay2022constant, xu2024constant}. By utilizing the structure of these codes as the tensor product of classical codes, we construct desired quantum-code homomorphisms by simply taking the tensor product of classical-code homomorphisms (Fig.~\ref{fig:homomorphic_measurement}(a)). 
By performing well-known structure-preserving modifications (such as puncturing and augmenting~\cite{huffman2010fundamentals}) of the base classical codes, we obtain structure-preserving modifications of the quantum codes that lead to inter-block homomorphic CNOTs between two \emph{distinct} homological product codes. Importantly, careful choices of such code modifications can preserve the code distances and only require transversal inter-block physical CNOTs, thereby ensuring that the constructed logical gadgets are fault-tolerant.

By applying these constructions to the hypergraph product (HGP) codes~\cite{tillich2014quantum}, which are 2D homological product codes, we obtain a new logical gadget that measures a selective pattern of Pauli product operators on any subgrid of the logical qubits in parallel  (Fig.~\ref{fig:homomorphic_measurement}(b)). In addition, we construct an automorphism gadget~\cite{breuckmann2022fold} that translates the logical grid with periodic boundary conditions for HGP codes with quasi-cyclic base codes. Combining these two gadgets with existing transversal/fold-transversal gates~\cite{quintavalle2022partitioning}, parallel logical computations with low space-time overhead can be realized directly. In particular, we show that  a layer of $\Theta(k)$ Clifford gates (consisting of Hadamards, $S$ gates, and intra-block CNOTs) on a HGP block with $k$ logical qubits can be applied with a constant space overhead in a sublinear ($< O(k)$) number of logical cycles. 

%
Using this broad set of efficient gate constructions, we study how they can be used to efficiently implement important algorithmic subroutines  (Fig.~\ref{fig:homomorphic_measurement}(c)), identifying additional structures that greatly reduce the required depth for specific algorithms. As examples, we show how to prepare a block of logical GHZ states, distill and consume magic states in parallel, and implement the quantum adder~\cite{gidney2018halving}---an important subroutine for many useful quantum algorithms---with  low space-time overhead. In addition, when applying our constructions to higher-dimensional homological product codes~\cite{campbell2019theory}, which support single-shot logical state preparation, we obtain even faster logical gadgets with a constant gate depth. Since the logical gadgets developed in this work are built upon transversal inter-block physical CNOTs, they are natural to implement in reconfigurable atom arrays~\cite{bluvstein2023logical} by overlapping two code blocks and applying global Rydberg laser pulses.

Before proceeding, we note several recently developed methods for realizing parallelizable logical computation based on concatenated codes~\cite{yamasaki2022time, yoshida2024concatenate, goto2024many} as well as color codes on hyperbolic manifold~\cite{zhu2023non}. Our work complements these studies by considering 
product qLDPC codes with concrete implementations~\cite{xu2024constant} and 
fault-tolerant protocols. The constructions in our work could potentially be also generalized to other product codes, such as lifted product codes~\cite{panteleev2022quantum}, generalized bicycle codes~\cite{panteleev2019degenerate, bravyi2024high}, fiber-bundle codes~\cite{Hastings2020}, and good qLDPC codes~\cite{panteleev2022quantum, breuckmann2020balanced, panteleev2022asymptotically, leverrier2022quantum, gu2022efficient,dinur2022good, lin2022good}. We also note that similar puncturing techniques were employed for implementing logical gates on qLDPC codes within the paradigm of code deformation~\cite{krishna2021fault}.

The manuscript is organized as follows: We begin by outlining the key insights and the main results in Sec.~\ref{sec:summary_of_results}. We introduce in Sec.~\ref{sec:preliminary} the background coding-theory techniques that form the basis of our technical constructions. Utilizing these techniques, we present a detailed construction of the homomorphic CNOTs and measurements for generic homological product codes in Sec.~\ref{sec:sketch_homomorphic_gadget}. With the new logical gadgets at hand, we study in Sec.~\ref{sec:HGP_codes} how common computational tasks and key algorithm subroutines can be compiled and implemented using the HGP codes with low space-time overhead. We further analyze in Sec.~\ref{sec:high_dimensional_codes} how the time overhead could be further reduced by using higher-dimensional homological product codes. Finally, we discuss the physical implementation of our proposed schemes in Sec.~\ref{sec:physical_implementation}. 
We note that readers primarily interested in the broad applications of our constructions may skip the technical Sec.~\ref{sec:preliminary} and Sec.~\ref{sec:sketch_homomorphic_gadget} first and focus on Secs.~\ref{sec:HGP_codes}-\ref{sec:physical_implementation}.
%

\section{Summary of key results \label{sec:summary_of_results}}

\begin{figure*}[!htbp]
    \centering
    \includegraphics[width=1\textwidth]{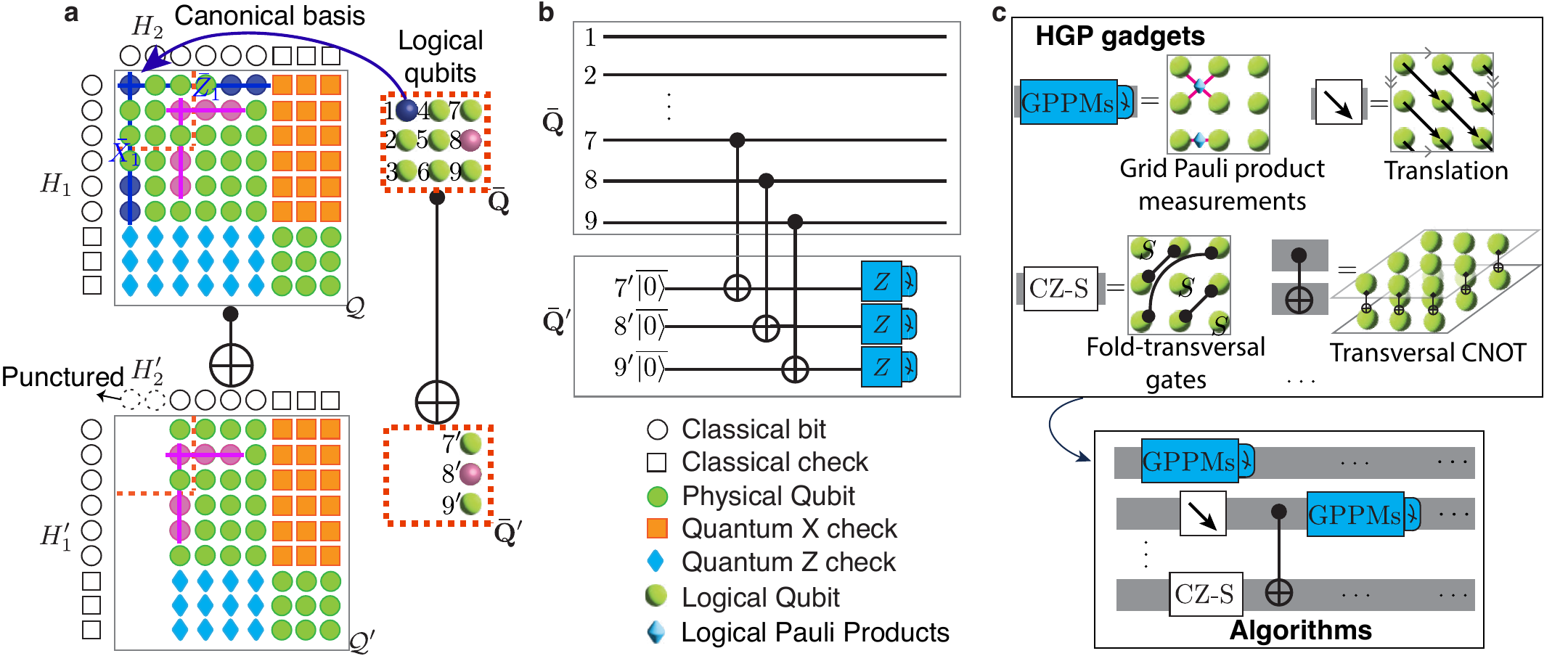}
    \caption{\textbf{Illustration of the homomorphic CNOT and other key gadgets for HGP codes.} \textbf{(a)} Each HGP code $\mc{Q}$, encoding a set of logical qubits $\mb{\bar{Q}}$, is constructed by taking the tensor product of two classical base codes $H_1$, $H_2$. 
    The qubits and quantum checks can be visualized on a square and inherit the structure of its base classical codes.
    The canonical basis provides a convenient way to make use of the logical structure (Sec.~\ref{sec:HGP_canonical_basis}, \cite{quintavalle2022reshape}).
    The logical $Z$ and $X$ operators of a logical qubit (3D balls) at the $(i,j)$ coordinate are supported on the $i$-th row and the $j$-th column of the physical qubits, respectively. 
    For instance, the blue physical qubits indicate the support of the logical $X$ and $Z$ operators associated with the blue logical qubit at coordinate $(1,1)$, and similarly for the pink logical qubit at $(2,3)$.
    We construct an ancilla HGP code $\mc{Q}^{\prime}$ by deleting a subset of the bits of the base codes of $\mc{Q}$, thereby deleting columns of data qubits and removing logical qubits supported on the deleted rows or columns.
    The structured code modification implies that applying physical transversal CNOTs (known as the homomorphic CNOT~\cite{huang2022homomorphic}) between the corresponding qubits of $\mc{Q}$ and $\mc{Q}^{\prime}$ gives rise to \emph{logical} transversal CNOTs between the corresponding logical qubits of $\mb{\bar{Q}}$ and $\mb{\bar{Q}^{\prime}}$ (third column in this example).
    \textbf{(b)} Non-destructive Pauli measurements on a subset of data logical qubits in $\mb{\bar{Q}}$ via homomorphic CNOTs and the generalized Steane measurement circuit~\cite{steane1997active}. 
    More specifically, we initialize logical qubits $\mb{\bar{Q^{\prime}}}$ in $\ket{0}$, apply a homomorphic CNOT between $\mb{\bar{Q}}$ and $\mb{\bar{Q}^{\prime}}$, and finally measure $\mb{\bar{Q}^{\prime}}$ in the $Z$ basis. Doing so measures single-qubit $Z$ operators on the $7$-th, $8$-th, and $9$-th logical qubits of $\mc{Q}$ in parallel without affecting the rest of the data logical qubits.
    This generalizes the standard Steane measurement circuit by allowing a different ancilla code $\mc{Q}^{\prime}$ from the data code $\mc{Q}$.
    \textbf{(c)} Using this approach, we  construct a suite of different gadgets with constant space overhead and $O(1)$ logical cycles: parallel measurements of a product pattern of Pauli product operators (Sec.~\ref{sec:GPPMs} and Fig.~\ref{fig:HGP_Grid_PPMs}), logical translation for HGP codes with quasi-cyclic base codes (Appendix~\ref{sec:translation_gadget}), fold-transversal gates~\cite{breuckmann2022fold, quintavalle2022partitioning} and transversal CNOTs. We use these building blocks to implement logical algorithms using only qLDPC blocks with low space-time overhead.
    }
    \label{fig:homomorphic_measurement}
\end{figure*}

\subsection{Parallel Pauli Product Measurements via Homomorphic CNOTs}
\label{subsec:gadgets}
In this work we first construct new fault-tolerant logical gadgets for a variety of qLDPC codes (Fig.~\ref{fig:homomorphic_measurement}, Fig.~\ref{fig:HGP_Grid_PPMs} and Table~\ref{tab:HGP_gadgets}), enabling selective, fast, and parallel logical operations.
Our construction builds on top of the homomorphic (i.e. structure preserving)  CNOT~\cite{huang2022homomorphic}, which generalizes the transversal CNOT between two identical CSS codes to two distinct codes, and the homomorphic measurement gadget, a generalization of the Steane measurement gadget~\cite{steane1997active}, which utilizes the homomorphic CNOT.
These methods apply physical CNOTs associated with some nontrivial homomorphisms (structure-preserving maps, reviewed in Sec.~\ref{sec:preliminary}) between two quantum codes, thereby guaranteeing that the stabilizer group is preserved and the process results in a valid logical operation.
%

To generalize these constructions from the restricted set of topological codes in Ref.~\cite{huang2022homomorphic} to a much broader range of high-rate qLDPC codes, we identify new homomorphisms for homological product codes (of any dimension) in Sec.~\ref{sec:sketch_homomorphic_gadget}, a widely-used family of product constructions.
We show that many well-known techniques for modifying classical codes, such as puncturing and augmenting~\cite{huffman2010fundamentals}, can be both structure- and distance-preserving, providing suitable homomorphisms between the modified code and the original classical code.
Applying this to the base classical codes involved in the product construction naturally induces a quantum-code homomorphism, giving rise to a homomorphic CNOT and an associated logical measurement gadget for the logical qubits.
Crucially, the modified codes (both classical and quantum) have preserved distances and the homomorphic logical CNOT only involves transversal physical CNOTs. As such, these gadgets are naturally fault-tolerant. 

We present an illustrative example of the above homomorphic gadget in Fig.~\ref{fig:homomorphic_measurement}(a).
Deleting (puncturing) a subset of bits of the base classical codes of the ancilla code $\mc{Q}^{\prime}$ removes a subset of the physical qubits. When choosing an appropriate logical operator basis (the canonical basis in Fig.~\ref{fig:homomorphic_measurement}(a)), this also removes a subset of the encoded logical qubits.
We thus obtain a smaller ancilla code $\mc{Q}^{\prime}$ that only encodes a subset of the logical qubits of $\mc{Q}$, but with a preserved distance.
The quantum-code homomorphism is then given by the natural inclusion map from $\mc{Q}^{\prime}$ to those of $\mc{Q}$.
Transversal CNOTs between the remaining physical qubits of $\mc{Q}$ and $\mc{Q}^{\prime}$ give rise to transversal logical CNOTs between the remaining logical qubits.
Measuring the ancilla logical block $\mc{Q}^{\prime}$ (Fig.~\ref{fig:homomorphic_measurement}(b)) then effectively measures a subset of logical qubits of $\mc{Q}$ \emph{in parallel}.
We can also measure a set of Pauli products of $\mc{Q}$ in parallel by adding (augmenting) checks to the base classical codes of the ancilla (see Fig.~\ref{fig:HGP_Grid_PPMs}).

This procedure works for general homological product codes, enabling selective, fast, and parallel Pauli product measurements of entire blocks of logical qubits.
Measurements of multiple disjoint Pauli products can be performed in parallel with our methods; the main constraint is that inheriting the product structure, the pattern of Pauli products needs to form a grid structure (Fig.~\ref{fig:HGP_Grid_PPMs}), which may naturally be present in many structured problems (see below).
The general procedure we employ here, in which selective and parallel logical operations on a data code are executed using a properly masked ancilla code patch, may also be useful for more general operations if suitable masks can be prepared.
When applied to higher-dimensional homological product codes (Sec.~\ref{sec:high_dimensional_codes}), this also results in logical gadgets with constant depth via check redundancies and the soundness property~\cite{campbell2019theory}.
It may be possible to further apply techniques of correlated decoding~\cite{cain2024correlated} and algorithmic fault tolerance~\cite{zhou2024algorithmic} to reduce the time cost in the case of the hypergraph product code, although this would require generalizing those results to the case of mixed types of QEC codes.
Additionally, for specific base classical codes with a quasi-cyclic structure, we can also form a quantum code automorphism~\cite{breuckmann2022fold} from the cyclic permutation of the underlying classical code, giving rise to a gadget that translates the logical qubits in a structured way. 
Building upon the parallel PPMs gadget and the translation gadget, we further construct a variety of useful gadgets for HGP codes in Appendix~\ref{sec:extra_gadgets}, including selective inter-block teleportation, logical cyclic shift, and parallel single-qubit gates, which enable complex logical computations.

\subsection{Fault-Tolerant Compilation with qLDPC Codes}
We next apply these new logical gadgets (Sec.~\ref{sec:HGP_codes}) to enable low space-time overhead implementations of a wide variety of large-scale logical quantum circuits on HGP codes (Fig.~\ref{fig:homomorphic_measurement}(c)), including random Clifford circuits (Theorem~\ref{theorem:parallel_Clifford_gates} and Table~\ref{tab:PPMs_cost_comparison}), GHZ state preparation (Fig.~\ref{fig:GHZ}), magic state distillation and consumption (Fig.~\ref{fig:MSD} and Fig.~\ref{fig:MSI}), and quantum addition (Fig.~\ref{fig:adder}).

The cost of QEC is usually quantified in terms of its space-time overhead, as various compilation methods can often trade between space and time~\cite{gidney2019efficient,fowler2012time,litinski2022active}.
As discussed in Sec.~\ref{sec:introduction}, the need for serialization in existing general qLDPC gate schemes may often negate the benefits of constant space overhead when considering the space-time cost.
Our methods, in contrast, enable the parallel implementation of many logical operations while maintaining constant space overhead, potentially improving the overall space-time overhead.

Using the homomorphic measurement gadgets described above and in Sec.~\ref{sec:sketch_homomorphic_gadget}, we form a ``Grid PPMs" (GPPMs, Fig.~\ref{fig:HGP_Grid_PPMs}) gadget for HGP codes that can measure a product pattern of PPMs on any subgrid of the logical qubits in parallel in one logical cycle (consisting of $d$ code cycles for a distance-$d$ code).
When combined with known transversal and fold-transversal gates~\cite{breuckmann2022fold, quintavalle2022partitioning}, we show that the GPPMs gadget generates the full Clifford group for an HGP code.
In addition, we construct a logical translation gadget that translates the logical qubit grid of a quasi-cyclic HGP code (with quasi-cyclic base classical codes) by simply permuting the physical qubits (see Table~\ref{tab:HGP_gadgets}, Appendix~\ref{sec:translation_gadget}, and Fig.~\ref{fig:translation_gadget}).

By selecting a family of quasi-cyclic HGP codes with competitive code parameters, we can execute parallel logical computations with not only a \emph{constant space overhead} but also \emph{lower space-time cost} compared to similar operations with conventional approaches based on the surface code. As shown in Table~\ref{tab:PPMs_cost_comparison} and Theorem~\ref{theorem:parallel_Clifford_gates} (see Appendix.~\ref{sec:proof_parallel_Clifford_gates} for details), a generic layer of $\Theta(k)$ Clifford gates (consisting of Hadamards, $S$ gates, and CNOTs) acting on $k$ logical qubits on a quasi-cyclic HGP block can be implemented in at most $O(k^{3/4})$ logical cycles, compared to $\Theta(k)$ for other existing schemes (Table~\ref{tab:PPMs_cost_comparison}). 
This low space-time overhead is fundamentally enabled by the parallelism of the GPPMs gadget, as existing selective gadgets involving lattice surgery with rateless codes~\cite{cohen2021quantum, bravyi2024high, xu2024constant} would have to execute computations sequentially to maintain constant space overhead and result in a linear logical depth.
We also note that the $O(k^{3/4})$ time overhead is based on a specific construction (which might not be optimal), and the practical cost of structured circuits may be even lower.

Turning our attention to fault-tolerant compilation for common algorithmic subroutines, we show that the logical parallelism can be further enhanced (and consequently, the time overhead further reduced) by compiling specific algorithms with more structured layers of gates.
We show that we can prepare a $k$-logical-qubit GHZ state (see Fig.~\ref{fig:GHZ}) in parallel with $O(1)$ space overhead in $O(1)$ logical cycles, using measurement-based preparation and the ability to measure multiple $ZZ$ Pauli products in parallel.
Noting that the bulk of operations in magic state factories can be done in parallel when using high-rate encodings, using transversal CNOTs and the parallel GPPMs, we show how to distill/consume $k$ magic states (see Fig.~\ref{fig:MSD} and Fig.~\ref{fig:MSI}) in parallel with $O(1)$ space overhead in $O(1)$ and $O(\sqrt{k}\log k)$ logical cycles, respectively, using HGP codes encoding $k$ logical qubits.
The latter, in particular, enables many practical algorithms involving parallel non-Clifford gates.
Finally, we present an efficient implementation of the quantum adder~\cite{gidney2018halving} with HGP blocks as such an example (see Fig.~\ref{fig:adder}), again utilizing the fact that through auxiliary Bell pairs commonly used for space-time trade-off~\cite{fowler2012time,gidney2019how,litinski2022active}, much of the adder structure can be executed in parallel via transversal CNOTs and the GPPMs gadget.

\section{Preliminaries \label{sec:preliminary}}
\subsection{Notation}
We use bold symbols to denote a set of objects. We denote $[n]$ as the set of integers $\{1,2,\cdots,n\}$ for some $n \in \mathbb{Z}^+$, $\{n_1 \rightarrow n_2\}$ the set of integers $\{n_1, n_1 + 1, \cdots, n_2\}$ with $n_2 \geq n_1$. 

Given a vector $v \in \mbb{F}^{n}$ over some field $\mbb{F}$ and a subset of indices $\mb{S} \subseteq [n]$, we denote $v|_{\mb{S}}$ as the restriction of $v$ on $\mb{S}$, i.e. a subvector of $v$ with only entries indexed by $\mb{S}$. Similarly, given a matrix $M \in \mbb{F}^{m\times n}$ over some field $\mbb{F}$ and a subset of column indices $\mb{S}$, we denote $H|_{\mb{S}}$ as the restriction of $H$ on $\mb{S}$, i.e. a submatrix of $H$ with only columns in $\mb{S}$. 

Let $\mb{Q}$ be some set of qubits and $\mb{O}$ be some set of coordinates. We say $\mb{Q} \simeq \mb{O}$ if $\mb{Q}$ are assigned to coordinates $\mb{O}$. Let $\mb{O}_0 \subseteq \mb{O}$. We refer to $\mb{Q}|_{\mb{O_0}}$ as the subset of qubits with coordinates $\mb{O}_0$. For two sets of qubits $\mb{Q}$ and $\mb{Q}^{\prime}$ assigned with the same set of coordinates $\mb{O}$, we refer to transversal CNOTs between them as pairs of CNOTs on each coordinate, i.e. $\bigotimes_{q \in \mb{O}}\mr{CNOT}(Q_q, Q^{\prime}_q)$.

We denote $\vec{e}_i$ as a unit column vector with the $i$-th entry being $1$. We do not specify the dimension of the vector in this notation, which should be clear from context.  

We denote $\mc{P}_n$ as the $n$-qubit Pauli group. 

\subsection{Classical codes, quantum codes, and chain complexes}
In this section, we review the basics of classical linear codes, quantum stabilizer codes, and their representation as chain complexes.

A $[n,k,d]$ classical linear code $\mc{C}$ over $\mbb{F}_2$ is a $k$-dimensional subspace of $\mbb{F}_2^n$. It can be specified as the row space of a generator matrix $G \in \mbb{F}_2^{k\times n}$, or the kernel of a check matrix $H \in \mbb{F}_2^{(n-k)\times n}$, with $H G^T = 0 \mod 2$. The distance $d$ of the code is the minimum Hamming weight of all codewords. 

A $[[n, k, d]]$ quantum stabilizer code $\mc{Q}$ is a $2^k$-dimensional subspace of the $2^n$-dimensional $n$-qubit Hilbert space. It is specified as the common $+1$ eigenspace of an Abelian subgroup $S$ of $\mc{P}_n$ that does not contain $-I$. The non-trivial logical operators of the code are given by $\mc{N}(S)\backslash S$,  where $\mc{N}$ denotes the normalizer with respect to $\mc{P}_n$. The distance of the code is the minimum Hamming weight of all nontrivial logical operators. 

For a Calderbank-Shor-Steane (CSS) stabilizer code~\cite{calderbank1996good,steane1996error}, the stabilizer generators can be divided into $X$-type operators and $Z$-type operators, represented by the $X$- and $Z$-check matrix $H_X \in \mbb{F}_2^{r_X\times n}$ and $H_Z \in \mbb{F}_2^{r_Z\times n}$, respectively. Each row $r$ of $H_X$ ($H_Z$) represents a $X$ ($Z$) type $n$-qubit Pauli operator $\bigotimes_{i = 1}^n X_i^{r_i}$ ($\bigotimes_{i = 1}^n Z_i^{r_i}$), where $X_i$ ($Z_i$) denotes the Pauli operator on the $i$-th qubit. The commutativity of the stabilizers requires that $H_X H_Z^T = 0 \mod 2$, which is also called the CSS condition. 

A length-$n$ chain complex (over $\mbb{F}_2$) $\mc{B}$
\begin{equation}
    B_n \xrightarrow{\partial_n} B_{n-1} \xrightarrow{\partial_{n-1}} \cdots \xrightarrow{\partial_{2}} B_1 \xrightarrow{\partial_1} B_0,
    \label{eq:chain_complex}
\end{equation}
is a collection of $\mbb{F}_2$ vector spaces $\{B_i\}_{i=0}^n$ and linear boundary maps between them $\{\partial_i: B_i \rightarrow B_{i-1}\}_{i = 1}^n$ that satisfy $\partial_{i-1}\partial_i = 0$. Let $\mb{\partial} := \{\partial_i\}_{i = 1}^n$. Informally, we can write $\partial^2 = 0$.

A classical code $\mc{C}$ with a check matrix $H$ can be represented by a length-$1$ chain complex:
\begin{equation}
    C_1 \xrightarrow{H} C_0,
\end{equation}
where the basis of $C_1$ and $C_0$ are associated with classical bits and checks, respectively. 

A CSS quantum stabilizer code $\mc{Q}$ with check matrices $H_X$ and $H_Z$ can be represented by a length-$2$ chain complex:
\begin{equation}
    Q_2 \xrightarrow{H_Z^T} Q_1 \xrightarrow{H_X} Q_0, 
    \label{eq:quantum_code_chain_complex}
\end{equation}
where the basis of $Q_2$, $Q_1$, and $Q_0$ are associated with $Z$ checks, qubits, and $X$ checks, respectively. Note that the maps are valid boundary maps due to the CSS condition: $H_X H_Z^T = 0 \mod 2$.

Because of the above relation between codes and chain complexes, we will refer to a classical code $\mc{C}$ or a quantum code $\mc{Q}$ interchangeably as their representing chain complex.

\subsection{Homological product codes from product complexes \label{sec:homological_product_codes}}
In this section, we review the homological product codes~\cite{bravyi2013homological, zeng2019higher, campbell2019theory, breuckmann2021quantum} and show how they are constructed from the tensor product of chain complexes. 
Technically, the terminology ``homological product codes" could refer to different constructions that cover the product of quantum codes~\cite{bravyi2013homological, campbell2019theory}. However, in this work, we only consider a subset of them -- the product of classical codes, or equivalently, length-$1$ chain complexes, (albeit referred to as high-dimensional hypergraph product codes~\cite{zeng2019higher}).
We closely follow the notation in Ref.~\cite{breuckmann2021quantum}.

Given $D$ base chain complexes, we can obtain a product complex by taking the tensor product of them:
\begin{definition}[Product complex]
    Let $\{\mc{B}^i\}_{i \in [D]}$ be $D$ chain complexes, where $\mc{B}^i = \{\{B^i_{x_i}\}_{x_i}, \{\partial^i_{x_i}\}_{x_i}\}$ (here, we do not specify the length of each chain complex, i.e. the range of $x_i$). We define a $D$-dimensional product complex $\mc{D} = \{\{ D_{\vec{x}} \}_{\vec{x} = (x_1, \cdots, x_n)^T \in \mbb{Z}^D}, \{ \partial^i_{\vec{x}}: D_{\vec{x}} \rightarrow D_{\vec{x} - \vec{e}_i}\}_{i \in [D], \vec{x} \in \mbb{Z}^D}\} :=
    \mr{Prod}(\{\mc{B}^i\}_{i \in [D]})$ as the tensor product of these chain complexes, where
    \begin{align}
        D_{\vec{x}} := \bigotimes_{i = 1}^D B^i_{x_i}, \quad \partial^i_{\vec{x}} := \bigotimes_{j = 1}^D (\partial^j_{x_j})^{\delta_{i,j}}, \label{eq:boundary_map_D_complex}
    \end{align}
    where $\delta$ is the Kronecker delta function and $(\partial^j_{x_j})^0$ is defined as the identity map.
    \label{def:product_complex}
\end{definition}
It is straightforward to check that the boundary maps of the product complex in Def.~\ref{def:product_complex} satisfy:
\begin{equation}
    \partial^i_{\vec{x} - \vec{e}_i} \partial^i_{\vec{x}} = 0 
    \quad \textrm{and} \quad 
    \partial^i_{\vec{x}} \partial^j_{\vec{y}} = \partial^j_{\vec{y}} \partial^i_{\vec{x}} (i \neq j).\\
    \label{eq:boundary_maps_HD_complex}
\end{equation}
Clearly, a product complex is a high-dimensional generalization of the chain complex (see Eq.~\eqref{eq:chain_complex}). Intuitively, it can be viewed as a $D$-dimensional hypercube, and its projection at a point $\vec{x}$ along the $i$-th direction resembles a chain complex. We can define the $i$-th type boundaries as $\mb{\partial}^i := \{\partial^i_{\vec{x}}\}_{\vec{x} \in \mbb{Z}^D}$. Informally, we can write Eq.~\eqref{eq:boundary_maps_HD_complex} as 
\begin{equation}
    (\partial^i)^2 = 0 \quad \mr{and} \quad [\partial^i, \partial^j] = 0 (i \neq j), 
\end{equation}
indicating that the linear maps along any direction form valid boundary maps and maps along different directions commute.

Note that in Def.~\ref{def:product_complex}, for simplicity, we allow the indices $x_i$ of the base chain complexes to take any integer values, making the projection of $\mc{D}$ along any direction infinite-length chain complexes. In practice, we typically apply some cutoff $T$ to the indices by, e.g. setting $D_{\vec{x}} = 0$ for any $|\vec{x}|_{\infty} > T$.

Since a quantum code is defined on a 1D chain complex (see Eq.~\eqref{eq:quantum_code_chain_complex}), we need to derive a 1D chain complex from the product complex, which is called the total complex:
\begin{definition}[Total complex of a product complex]
Let $\mc{D}$ be a product complex out of $D$ base chain complexes with vectors spaces $\{D_{\vec{x}}\}_{\vec{x}}$ and boundary maps $\{\partial^i_{\vec{x}}\}_{i \in [D], \vec{x}}$. We define its total chain complex $\mc{T} = \{\{T_k\}_k, \{\delta_k\}_k \}:= \mr{Tot}(\mc{D})$ as follows:
\begin{align}
    T_k := \bigoplus_{|\vec{x}| = k} D_{\vec{x}},
\end{align}
and the boundary maps:
\begin{equation}
    \delta_k (\bigoplus_{|\vec{x}| = k} a_{\vec{x}})=\sum_{|\vec{x}|=k}\left( \bigoplus_{|\vec{y}|=k-1} \partial_{\vec{y}, \vec{x}} a_{\vec{x}}\right),
    \label{eq:boundary_map_total_complex}
\end{equation}
for any $a_{\vec{x}} \in D_{\vec{x}}$ and 
\begin{equation}
    \partial_{\vec{y}, \vec{x}} := 
    \begin{cases}
    \partial^i_{\vec{x}} & \vec{x} - \vec{y} = \vec{e}_i\ \textrm{for some}\ i \in [D]\\
    0,              & \text{otherwise}
\end{cases},
\label{eq:partial_yx}
\end{equation}
\label{def:total_complex}
\end{definition}
Intuitively, the total complex is obtained by projecting the $D$-dimensional complex along the ``diagonal'' direction.

Once we obtain a total chain complex from a product complex (which can have a length longer than 2), we can define a quantum code from a length-$2$ subcomplex. In this work, we will focus on product complexes with length-$1$ base complexes $\{\mc{C}^i\}_{i \in [D]}$ (classical codes). In this case, a total complex 
$\mc{T} = \mr{Tot}(\mr{Prod}(\{\mc{C}^i\}_{i \in [D]}))$ will be of length $D$:
\begin{equation}
    T_D \xrightarrow{\delta_D} T_{D - 1} \xrightarrow{\delta_{D - 1}} \cdots \xrightarrow{\delta_1} T_0.
\end{equation}
Furthermore, we will primarily focus on $D = 2,3,4$. For $D = 2$, we obtain the standard hypergraph product code~\cite{tillich2014quantum}, with planar surface codes being a special instance. For $D = 3$ and $D = 4$, we obtain $3D$ and $4D$ homological product codes, with $3D$ and $4D$ surface/toric code being special instances, respectively.

\subsection{Homomorphic CNOT and homomorphic measurement gadget \label{sec:homomorphic_measurement}}
In this section, we briefly review the general framework of the homomorphic CNOT and the homomorphic measurement gadget introduced in Ref.~\cite{huang2022homomorphic}. 

\begin{definition}[Homomorphic CNOT]
    Let $\mc{Q}$ and $\mc{Q}^{\prime}$ be two quantum codes associated with two chain complexes $\{\{Q_i\}_{i = 0}^2, \{\partial_i\}_{i = 1}^2\}$, and $\{\{Q^{\prime}_i\}_{i = 0}^2, \{\partial^{\prime}_i\}_{i = 1}^2\}$, respectively. Let $\mb{\gamma} = \{\gamma_i: Q_i^{\prime} \rightarrow Q_i\}_{i = 0}^2$ be a homomorphism between the two chain complexes, i.e., the following diagram is commutative:
    \begin{equation}
    \label{eq:quantum_code_homo}
        \begin{tikzcd}
	{Q_2} & {Q_1} & {Q_0} \\
	{Q_2^{\prime}} & {Q_1^{\prime}} & {Q_0^{\prime}}
	\arrow["{\partial_2}", from=1-1, to=1-2]
	\arrow["{\partial_1}", from=1-2, to=1-3]
        \arrow["{\partial_2^{\prime}}", from=2-1, to=2-2]
	\arrow["{\partial_1^{\prime}}", from=2-2, to=2-3]
        \arrow["{\gamma_2}", from=2-1, to=1-1]
	\arrow["{\gamma_1}", from=2-2, to=1-2]
        \arrow["{\gamma_0}", from=2-3, to=1-3]
\end{tikzcd}
    \end{equation}
     Then physical $\mc{Q}$-controlled CNOTs specified by $\gamma_1$, i.e. a physical CNOT controlled by the $i$-th qubit of $\mc{Q}$ and targeted the $j$-th qubit of $\mc{Q}^{\prime}$ is applied if and only if $\gamma_1[i,j] = 1$, give some $\mc{Q}$-controlled \textbf{logical} CNOT gates between $\mc{Q}$ and $\mc{Q}^{\prime}$. We refer to such a logical gadget as a homomorphic CNOT associated with the homomorphism $\mb{\gamma}$.
    \label{def:homomorphic_CNOT}
\end{definition}

The above homomorphic CNOT gadget is a valid logical operation since the conditions that the stabilizers are preserved under such a gadget are equivalent to the diagram in Eq.~\eqref{eq:quantum_code_homo} being commutative~\cite{huang2022homomorphic}. In other words, finding a homomorphism between two quantum codes directly leads to a logical gadget that implements some inter-block logical CNOTs. 

Using the homomorphic CNOT in Def.~\ref{def:homomorphic_CNOT}, one can implement a homomorphic measurement gadget on a data quantum code $\mc{Q}$ by constructing a specific ancilla code $\mc{Q}^{\prime}$ and implementing the generalized Steane measurement (see Fig.~\ref{fig:homomorphic_measurement}(b)) that utilizes the inter-block homomorphic CNOTs. Specifically, by initializing the logical qubits of $\mc{Q}^{\prime}$ in the $Z$ ($X$) basis, applying the $\mc{Q}$- ($\mc{Q}^{\prime}$-) controlled homomorphic CNOTs, and measuring $\mc{Q}^{\prime}$ in the $Z$ ($X$) basis, we can measure products of Pauli $Z$ ($X$) logical operators of $\mc{Q}$. 

In general, the homomorphic measurement gadget implements $N$ Pauli product measurements (PPMs) on the data code non-destructively and \emph{in parallel}, where $N$ equals the number of logical qubits in $\mc{Q}^{\prime}$. In addition, it is easy to guarantee the fault tolerance of the gadget by using a large-distance ancilla, constant-depth homomorphic CNOTs, and fault-tolerant ancilla state preparation and measurement (which can be done by simply performing $d$ QEC cycles and transversally measuring the qubits, respectively, for any distance-$d$ CSS code).
Ref.~\cite{huang2022homomorphic} has constructed homomorphic measurement gadgets for performing one PPM on a toric code or a hyperbolic surface code using an ancilla code that encodes a single logical qubit. However, their construction relies on the notion of covering spaces and the topological properties of codes and it was not clear how to generalize their constructions to algebraically constructed qLDPC codes.


\subsection{Modifying classical codes \label{sec:puncturing_shortening}}
Here, we review some well-known techniques for modifying classical codes, which we will utilize later to induce structure-preserving modifications on the quantum homological product codes for implementing the homomorphic gadgets in Sec.~\ref{sec:homomorphic_measurement}.


\begin{definition}[Puncturing and shortening]
Let $\mc{C}$ be a $[n,k,d]$ classical code with a check matrix $H \in \mbb{F}_2^{(n-k)\times n}$ and a generator matrix $G \in \mbb{F}_2^{k\times n}$. Let $\mb{S} \subseteq [n]$ be a set of bit indices. Puncturing $H$ on $\mb{S}$ gives a new check matrix $H^{\mb{S}}$, which is defined as a submatrix of $H$ with columns in $\mb{S}$ deleted, i.e.
\begin{equation}
    H^{\mb{S}} := H|_{[n]\backslash \mb{S}}.
\end{equation}
Moreover, the generator matrix of $H^{\mb{S}}$ is $G_\mb{S}$, which is obtained by shortening $G$ on $\mb{S}$:
\begin{enumerate}
    \item Find the subcode of $\mc{C}$ with a generator matrix $M \in \mathbb{F}_2^{(k - m)\times n}$, for some $m \leq k$, such that:
    \begin{equation}
        \mr{rs}(M) = \mr{span}\{h \in \mr{rs}(G) \mid h|_{\mb{S}} = 0_{|\mb{S}|}\}. \label{eq:shorten_1}
    \end{equation}
    \item 
    \begin{equation}
        G_{\mb{S}} := M^{\mb{S}}. \label{eq:shortern_2}
    \end{equation}
\end{enumerate}
\label{def:puncturing_oper}
\end{definition}

Obviously, puncturing and shortening are dual to each other: puncturing the check matrix of a code corresponds to shortening its generator matrix (on the same set of bits). This gives a way of constructing a new code $\mc{C}^{\prime}$ from an old code $\mc{C}$ by puncturing its check matrix on a subset of bits. We refer to such a transformation as puncturing a code on some set of bits, for simplicity. In general, the new code $\mc{C}^{\prime}$ has $n^{\prime} \leq n$ and $k^{\prime} \leq k$. For generic puncturing, there is no guarantee on the new code distance, which could either increase, decrease, or remain the same.


\begin{definition}[Augmenting and expurgating]
Let $\mc{C}$ be a $[n,k,d]$ classical code with a check matrix $H \in \mbb{F}_2^{(n-k)\times n}$ and a generator matrix $G \in \mbb{F}_2^{k\times n}$. Let $H_0 \in \mbb{F}_2^{r_0\times n}$ be some new checks. Augmenting $H$ with $H_0$ gives a new check matrix $H^{+H_0}$, which is defined as appending the new checks in $H_0$ to $H$, i.e.
\begin{equation}
    H^{+H_0} := \left(\begin{array}{c}
         H \\
         H_0
    \end{array}
    \right).
\end{equation}
The generator matrix of $H^{+H_0}$ is $G_{-H_0}$, which is obtained by removing the codewords in $\mr{rs}(G)$ that do not satisfy the extra constraints imposed by $H_0$, a process called expurgating:
\begin{equation}
    \mr{rs}(G_{-H_0}) =  \mr{rs}(G) \cap \ker{H_0},
\end{equation}
where $\mr{rs}(\bullet)$ denotes the row space.
\label{def:augmenting_oper}
\end{definition}
Augmenting and expurgating are also dual to each other: augmenting the check matrix of a code corresponds to expurgating its generator matrix (with respect to the same set of extra checks). This also gives a way of constructing a new code $\mc{C}^{\prime}$ from an old code $\mc{C}$ by adding some extra checks. We refer to such a transformation as augmenting a code with some new checks, for simplicity. In general, the new code $\mc{C}^{\prime}$ has $n^{\prime} = n, k^{\prime} \leq k$ and $d^{\prime} \geq d$. 

For the purpose of this work, we will need to modify classical codes in a way that preserves the code distance. The augmenting-expurgating operation trivially satisfies this requirement since it only removes a subset of codewords. However, the same does not hold for the puncturing-shortening operation in general since some bits are removed during such an operation. Nevertheless, as we will show in the following, we can make sure that the distance is preserved if we only puncture on a specific subset of bits for any given code.

Given any $[n,k,d]$ classical code $\mc{C}$ with a check matrix $H \in \mbb{F}_2^{(n-k)\times n}$ and a set of bits labeled by the column indices of $H$, i.e. $\mb{B} \simeq [n]$, we can find a subset of $n - k$ bits $\mb{B_{\mr{NI}}}$ such that the columns of $H$ indexed by $\mb{B_{\mr{NI}}}$ are linearly independent. Without loss of generality, we can assume that $\mb{B_{\mr{NI}}}$ are the last $n - k$ bits since otherwise we can simply permute $\mb{B_{\mr{NI}}}$ to the last $n - k$ bits. Then, by performing row elementary operations, we can transform $H$ to a canonical form:
\begin{equation}
    H_c = \left( h_1, h_2,\cdots, h_k, I_{n-k} \right),
\end{equation}
where $h_i \in \mbb{F}_2^{(n-k)\times 1}$. With this canonical $H_c$, we can easily obtain the generator matrix, which is also in the canonical form:
\begin{equation}
    G_c = \left(I_k, \begin{array}{c}
        h_1^T \\
        h_2^T \\
        \cdots \\
        h_k^T\\
    \end{array} \right).
    \label{eq:classical_canonical_generator_mat}
\end{equation}

The complementary of $\mb{B}_{\mr{NI}}$, $\mb{B_I} := [n]\backslash \mb{B}_{\mr{NI}}$ (the first $k$ bits in the canonical form), are referred to as the information bits in the classical code literature~\cite{huffman2010fundamentals}.
Now, we show that puncturing on any set of bits $\mb{S} \subseteq \mb{B_I}$ that are information bits does not reduce the code distance. Again, without loss of generality, we assume $\mb{B_I} = [k]$ and $\mb{S} = [|\mb{S}|]$. Based on Def.~\ref{def:puncturing_oper}, we know that puncturing $H$ on $\mb{S}$ corresponds to shortening $G_c$ on $\mb{S}$, i.e. the new generator matrix can be written as
\begin{equation}
    G_{\mb{S}} = \left(I_{k - |\mb{S}|}, \begin{array}{c}
        h_{|\mb{S}| + 1}^T \\
        h_{|\mb{S}| + 2}^T \\
        \cdots \\
        h_k^T\\
    \end{array} \right).
\end{equation}
Since the $k - |\mb{S}|$ rows of $G_{\mb{S}}$ are the same as the last  $k - |\mb{S}|$ rows of $G_c$, up to some extra zeros, the distance of the new code is equal to or greater than $d$.

Moreover, the transformation from $G_c$ to $G_{\mb{S}}$ entails a concise transformation on the codewords by the puncturing operation. Let $\mb{\bar{C}} \simeq [k]$ denote the $k$ codewords of $\mc{C}$, which we refer to as logical bits, for simplicity. Puncturing on the bits indexed by $\mb{S}$ simply corresponds to removing the logical bits indexed by $\mb{S}$ (up to shortening other logical bits by some zero entries).

\section{Homomorphic CNOT and measurements for homological product codes \label{sec:sketch_homomorphic_gadget}}
In this section, we present our main technical results --- constructing nontrivial logical homomorphic CNOT and measurements (see Def.~\ref{def:homomorphic_CNOT}) for generic homological product codes (see Sec.~\ref{sec:homological_product_codes}). 

\subsection{Quantum-code homomorphisms induced by classical-code homomorphisms}
As introduced in Sec.~\ref{sec:homomorphic_measurement}, constructing a homomorphic CNOT between two quantum codes $\mc{Q}$ and $\mc{Q}^{\prime}$ comprises finding a homomorphism $\mb{\gamma}: \mc{Q}^{\prime} \rightarrow \mc{Q}$. However, finding a nontrivial homomorphism between two generic, non-topological quantum codes is challenging~\cite{huang2022homomorphic}. Fortunately, utilizing the product structure of homological product codes (Sec.~\ref{sec:homological_product_codes}), we can reduce the task of finding homomorphisms between two homological product codes to finding homomorphisms between their base classical codes, which, as we will show later, is a much easier task.

Let $\{\mc{C}^i\}_{i \in [D]}$ and $\{\mc{C}^{\prime i}\}_{i \in [D]}$ be two sets of base classical codes for constructing two $D$-dimensional homological product codes $\mc{Q}$ and $\mc{Q}^{\prime}$, respectively. Let $\{ \mb{\gamma}^i: \mc{C}^{\prime i} \rightarrow \mc{C}^i\}_{i \in [D]}$ be a set of homomorphisms between the classical codes, i.e. the following diagram is commutative:
\begin{equation}
    \begin{tikzcd}
    	{C^i_1} & {C^i_0} \\
    	{C_1^{\prime i}} & {C_0^{\prime i}}
    	\arrow["{\partial^i_1}", from=1-1, to=1-2]
            \arrow["{\partial^{\prime i}_1}", from=2-1, to=2-2]
            \arrow["{\gamma^i_{1}}", from=2-1, to=1-1]
    	\arrow["{\gamma^i_{0}}", from=2-2, to=1-2]
    \end{tikzcd}
    \end{equation}
In the following, we will show that a homomorphism $\mb{\gamma}: \mc{Q}^{\prime} \rightarrow \mc{Q}$ can be essentially constructed by taking tensor products of the classical-code homomorphisms $\{\gamma^i\}_{i \in [D]}$.

Recall that to construct a homological product $\mc{Q}$, we first construct a product complex $\mc{D} = \mr{Prod}(\{\mc{C}^i\}_{i = 1}^D)$ as the tensor product of all the base classical codes (see Def.~\ref{def:product_complex}). Then we project $\mc{D}$ onto a length-$D$ chain complex $\mc{T} = \mr{Tot}(\mc{D})$ (see Def.~\ref{def:total_complex}). Finally, for $D \geq 2$, we define the quantum code $\mc{Q} \subseteq \mc{T}$ as a length-$2$ subcomplex of  $\mc{T}$. The same construction is applied to $\mc{Q}^{\prime}$, i.e. $\mc{Q}^{\prime} \subseteq \mc{T}^{\prime} := \mr{Tot}(\mc{D}^{\prime})$, where $\mc{D}^{\prime} := \mr{Prod}(\{\mc{C}^{\prime i}\}_{i \in [D]})$.

To construct a homomorphism between $\mc{Q}$ and $\mc{Q}^{\prime}$, we first construct a homomorphism between $\mc{D}$ and $\mc{D}^{\prime}$ by simply taking the tensor product of the classical-code homomorphisms:

\begin{proposition}[Classically induced homomorphism for product complexes]
    The linear map $\gamma^D = \{\gamma^D_{\vec{x}}: D^{\prime}_{\vec{x}} \rightarrow D_{\vec{x}} \}_{\vec{x}}$, where
    \begin{equation}
        \gamma^D_{\vec{x}} := \bigotimes_{i=1}^D \gamma^i_{x_i},
        \label{eq:complex_homo_2}
    \end{equation}
    is a homomorphism between $\mc{D}^{\prime}$ and $\mc{D}$.
    \label{lemma:prod_complex_homo}
\end{proposition}
\begin{proof}
    We need to show that $\gamma^D$ preserves the boundary maps of the product complexes. More concretely, we need to show that the following diagram is commutative for any $\vec{x}$ and $i$:
    \begin{equation}
    \begin{tikzcd}
    	{D_{\vec{x}}} & {D_{\vec{x} - \vec{e}_i}} \\
    	{D^{\prime}_{\vec{x}}} & {D^{\prime}_{\vec{x} - \vec{e}_i}}
    	\arrow["{\partial^i_{\vec{x}}}", from=1-1, to=1-2]
            \arrow["{\partial^{\prime i}_{\vec{x}}}", from=2-1, to=2-2]
            \arrow["{\gamma^D_{\vec{x}}}", from=2-1, to=1-1]
    	\arrow["{\gamma^D_{\vec{x} - \vec{e}_i}}", from=2-2, to=1-2]
    \end{tikzcd}
    \end{equation}
Based on Eq.~\eqref{eq:boundary_map_D_complex} and Eq.~\eqref{eq:complex_homo_2}, we have 
\begin{equation}
    \partial^i_{\vec{x}}\gamma^D_{\vec{x}} = \bigotimes_{j=1}^D (\partial^j_{x_j})^{\delta_{i,j}} \gamma^j_{x_j} = \bigotimes_{j=1}^D \gamma^j_{x_j - \delta_{i,j}} (\partial^{\prime j}_{x_j})^{\delta_{i,j}} = \gamma^D_{\vec{x}-\vec{e}_i} \partial^{\prime i}_{\vec{x}},
    \label{eq:commutative_eq_complex}
\end{equation}
where we have utilized the assumption that $\gamma^j: \mc{C}^{j \prime} \rightarrow \mc{C}^j$ is a homomorphism, i.e.
\begin{equation}
    \partial^j_{x_j} \gamma^j_{x_j} = \gamma^j_{x_j - 1}\partial^{\prime j}_{x_j}.
\end{equation}
\end{proof}

With the classically induced homomorphism $\gamma^D: \mc{D}^{\prime} \rightarrow \mc{D}$, we can transform it to a homomorphism between the total complexes $\gamma: \mc{T}^{\prime} \rightarrow \mc{T}$, which also serves as the homomorphism between the subcomplexes $\mc{Q}^{\prime} \subseteq \mc{T}^{\prime}$ and $\mc{Q} \subseteq \mc{T}$:

\begin{proposition}[Classically induced homomorphism for homological product codes]
The linear map $\mb{\gamma} = \{\gamma_k: T^{\prime}_k \rightarrow T_k\}$ from $\mc{T}^{\prime}$ to $\mc{T}$, where 
\begin{equation}
    \gamma_k := \bigoplus_{|\vec{x}|=k} \gamma^D_{\vec{x}},
    \label{eq:total_complex_homo}
\end{equation}
with $\{\gamma^D_{\vec{x}}\}$ defined in Eq.~\eqref{eq:complex_homo_2}, is a homomorphism from $\mc{T}^{\prime}$ to $\mc{T}$. Note that the direct sum in Eq.~\eqref{eq:total_complex_homo} means that $\gamma_k$ is in a block-diagonal form, i.e. $\gamma_k T_k = \gamma_k (\bigoplus_{|\vec{x}| = k} D_{\vec{x}}) = \bigoplus_{|\vec{x}| = k}(\gamma^D_{\vec{x}} D_{\vec{x}})$. $\gamma$ also serves as a homomorphism from $\mc{Q}^{\prime} \subseteq \mc{T}^{\prime}$ to $\mc{Q} \subseteq \mc{T}$. 
\label{prop:homological_prod_code_homo_by_classical}
\end{proposition}
\begin{proof}
We need to prove that the following diagram is commutative 
    \begin{equation}
    \begin{tikzcd}
    	{T_{k}} & {T_{k - 1}} \\
    	{T^{\prime}_{k}} & {T^{\prime}_{k-1}}
    	\arrow["{\delta_k}", from=1-1, to=1-2]
            \arrow["{\delta^{\prime}_k}", from=2-1, to=2-2]
            \arrow["{\gamma_{k}}", from=2-1, to=1-1]
    	\arrow["{\gamma_{k-1}}", from=2-2, to=1-2]
    \end{tikzcd}
    \end{equation}
Let $a^{\prime} = \bigoplus_{|\vec{x}|=k}a^{\prime}_{\vec{x}} \in T_k^{\prime}$, where $a^{\prime}_{\vec{x}} \in \mc{D}^{\prime}_{\vec{x}}$. According to Eq.~\eqref{eq:boundary_map_total_complex} and Eq.~\eqref{eq:total_complex_homo}, we have
\begin{equation}
    \delta_{k} \gamma_k a^{\prime} = \delta_k (\bigoplus_{|\vec{x}|=k}\gamma^D_{\vec{x}}a^{\prime}_{\vec{x}}) 
    = \sum_{|\vec{x}|=k} ( \bigoplus_{|\vec{y}|=k-1} \partial_{\vec{y}, \vec{x}} \gamma^D_{\vec{x}} a^{\prime}_{\vec{x}}),
    \label{eq:LHS}
\end{equation}
and 
\begin{equation}
    \gamma_{k-1}\delta_k^{\prime} a^{\prime} 
    = \sum_{|\vec{x}|=k} \gamma_{k-1} (\bigoplus_{|\vec{y}|=k-1} \partial^{\prime}_{\vec{y}, \vec{x}} a^{\prime}_{\vec{x}} ) 
    =  \sum_{|\vec{x}|=k} (\bigoplus_{|\vec{y}|=k-1} \gamma^D_{\vec{y}} \partial^{\prime}_{\vec{y}, \vec{x}} a^{\prime}_{\vec{x}} ).
    \label{eq:RHS}
\end{equation}
To prove Eq.~\eqref{eq:LHS} equals Eq.~\eqref{eq:RHS}, we only need to show $\partial_{\vec{y}, \vec{x}} \gamma^D_{\vec{x}} a^{\prime}_{\vec{x}} = \gamma^D_{\vec{y}} \partial^{\prime}_{\vec{y}, \vec{x}} a^{\prime}_{\vec{x}}$. If $\vec{x} - \vec{y} = \vec{e}_i$ for some $i \in [D]$, $\partial_{\vec{y}, \vec{x}} = \partial^i_{\vec{x}}$ (see Eq.~\eqref{eq:partial_yx}), and 
\begin{equation}
    \partial_{\vec{y}, \vec{x}} \gamma^D_{\vec{x}} = \partial^i_{\vec{x}}\gamma^D_{\vec{x}} = \gamma^D_{\vec{x} - \vec{e}_i} \partial^{\prime i}_{\vec{x}} = \gamma^D_{\vec{y}} \partial^{\prime}_{\vec{y}, \vec{x}},
\end{equation}
according to Proposition~\ref{lemma:prod_complex_homo} (see Eq.~\eqref{eq:commutative_eq_complex}); Otherwise, $\partial_{\vec{y}, \vec{x}} \gamma^D_{\vec{x}} = \gamma^D_{\vec{y}} \partial^{\prime}_{\vec{y}, \vec{x}} = 0$ since $\partial_{\vec{y}, \vec{x}} = \partial^{\prime}_{\vec{y}, \vec{x}} = 0$.
\end{proof}

\subsection{Classical-code homomorphisms}
According to Proposition~\ref{prop:homological_prod_code_homo_by_classical}, we can find a homomorphism between two homological product codes by simply finding a set of homomorphisms between their base classical codes. 
In this section, we construct a few useful classical-code homomorphisms based on the puncturing and the augmenting operations on classical codes (see Sec.~\ref{sec:puncturing_shortening}).


We first show that the puncturing (see Def.~\ref{def:puncturing_oper}) and the augmenting (see Def.~\ref{def:augmenting_oper}) operations on a classical code are structure-preserving and they naturally induce a homomorphism between the old code and the modified code, which we call the puncturing-augmenting homomorphism:
\begin{proposition}[Puncturing-Augmenting homomorphism for classical codes]
Let $\mc{C}: C_1 \xrightarrow{H} C_0$ be a $[n,k,d]$ classical code, $\mb{S} \subseteq [n]$ a subset of bit indices, and $H_0 \in \mbb{F}_2^{m\times (n - |\mb{S}|)}$ a set of extra checks. Without loss of generality, assume that $\mb{S} = [|\mb{S}|]$. Let $\mc{C}^{\prime}: C_1^{\prime} \xrightarrow{H^{\mb{S, +H_0}}} C_0^{\prime}$ be another classical code obtained by puncturing $\mc{C}$ on $\mb{S}$ and then augmented with $H_0$. Then the linear map $\gamma = \{\gamma_1: C_1^{\prime} \rightarrow C_1, \gamma_0: C_0^{\prime} \rightarrow C_0\}$, where
\begin{equation}
    \gamma_1 = \left(\begin{array}{c}
         0_{|\mb{S}|\times (n - |\mb{S}|)}\\
         I_{n - |\mb{S}|}
    \end{array}\right),
    \quad \gamma_0 = \left(I_{n - k}, 0_{(n-k)\times m}\right),
    \label{eq:puncturing_maps}
\end{equation}
is a homomorphism from $\mc{C}^{\prime}$ to $\mc{C}$, i.e.
\begin{equation}
    \begin{tikzcd}
    	{C_1} & {C_0} \\
    	{C_1^{\prime}} & {C_0^{\prime}}
    	\arrow["{H}", from=1-1, to=1-2]
            \arrow["{H^{\mb{S}, +H_0}}", from=2-1, to=2-2]
            \arrow["{\gamma_{1}}", from=2-1, to=1-1]
    	\arrow["{\gamma_{0}}", from=2-2, to=1-2]
    \end{tikzcd}
    \end{equation}
is commutative.
\label{prop:puncturing_augmenting_homo}
\end{proposition}
\begin{proof}
    Based on the definition of $H^{\mb{S}, +H_0}$, we have
    \begin{equation}
        H = (\alpha, H^{\mb{S}}), \quad H^{\mb{S}, +H_0} =
        \left(\begin{array}{c}
         H^{\mb{S}} \\
         H_0
        \end{array}\right).
    \end{equation}
    for some $\alpha \in \mbb{F}_2^{(n-k)\times |\mb{S}|}$. Then, according to Eq.~\eqref{eq:puncturing_maps}, 
    \begin{equation}
        H \gamma_1 = H^{\mb{S}} = \gamma_0 H^{\mb{S}, +H_0}.
    \end{equation}
\end{proof}

    

Note that in Proposition~\ref{prop:puncturing_augmenting_homo}, we have assumed that $\mb{S}$ are the first $|\mb{S}|$ bits for simplicity. In general, for any $\mb{S} \subseteq [n]$,  the bits $[n-|\mb{S}|]$ of $\mc{C}^{\prime}$ are identified with the bits $[n]\backslash \mb{S}$ of $\mc{C}$. Then, $\gamma_1$ can be defined as the inclusion map from $[n - |\mb{S}|]$ to $[n]$ with such an identification. 

If a homomorphism in Proposition~\ref{prop:puncturing_augmenting_homo} only involves puncturing (augmenting) of the classical code, we refer to it as a puncturing (augmenting) homomorphism.

    

\subsection{Homomorphic measurement gadget for generic homological product codes \label{sec:general_homo_gadget}}
In this section, we sketch the homomorphic measurement gadget that can perform selected and parallel PPMs on a generic homological product code.

To perform a selected set of PPMs on a data code $\mc{Q}$, we first prepare a carefully designed ancilla code $\mc{Q}^{\prime}$, which is another homological product code of the same dimension. According to Def.~\ref{def:homomorphic_CNOT}, we can construct a homomorphic CNOT gadget from $\mc{Q}$ to $\mc{Q}^{\prime}$ if we can find a homomorphism between $\mc{Q}$ and $\mc{Q}^{\prime}$. According to Proposition~\ref{prop:homological_prod_code_homo_by_classical}, such a homomorphism can be constructed by essentially taking the tensor product of the homomorphisms between the base classical codes. 
We construct such classical-code homomorphisms by using the puncturing-augmenting homomorphisms presented in Proposition~\ref{prop:puncturing_augmenting_homo}. Under such a construction, the base codes of $\mc{Q}^{\prime}$ are constructed by puncturing and/or augmenting the base codes of $\mc{Q}$, which can remove some classical codewords. As we will show in Sec.~\ref{sec:HGP_codes} and Sec.~\ref{sec:high_dimensional_codes}, the logical operators of $\mc{Q}$ and $\mc{Q}^{\prime}$ are essentially given by distributing their base classical codewords to different rows/columns. As such, we can obtain a smaller ancilla code $\mc{Q}^{\prime}$ that encodes fewer logical qubits. Then, using the generalized Steane measurement gadget (see Fig.~\ref{fig:homomorphic_measurement}(b)), we can achieve selected measurements on a subset of logical qubits of $\mc{Q}$ using $\mc{Q}^{\prime}$.

For example, as shown in Fig.~\ref{fig:homomorphic_measurement}(a,b), by puncturing the classical codes on selected subsets of bits, we remove a subset $\mb{Q_0}$ of the qubits $\mb{Q}$ of $\mc{Q}$ when constructing $\mc{Q}^{\prime}$. Let $\mb{\gamma} = \{\gamma_2, \gamma_1, \gamma_0\}: \mc{Q}^{\prime} \rightarrow \mc{Q}$ be the homomorphism induced by the classical puncturing-augmenting homomorphisms (see Proposition~\ref{prop:puncturing_augmenting_homo}). The qubits $\mb{Q}^{\prime}$ are identified with $\mb{Q}\backslash \mb{Q_0}$ under $\gamma_1$, i.e. $\gamma_1(Q^{\prime}_i) = Q_{i}$ for $i \in [|\mb{Q}\backslash \mb{Q_0}|]$, where we have labeled the qubits such the $\mb{Q_0}$ are the last $|\mb{Q_0}|$ qubits of $\mc{Q}$. Deleting $\mb{Q_0}$ also deletes a subset of logical qubits in $\mc{Q}$ when constructing $\mc{Q}^{\prime}$. Specifically, let $\mb{\bar{L}}$ be the set of nontrivial logical operators of $\mc{Q}$. We can choose a representation of the logical qubits $\mb{\bar{Q}} = \{\bar{Q}_i\}_{i=1}^{k}$, such that the logical $X$ and $Z$ operators $(\bar{X}_i, \bar{Z}_i)$ ($\bar{X}_i, \bar{Z}_i \in \mb{\bar{L}}$) of $\bar{Q}_i$ form conjugating pairs, i.e. $[\bar{X}_i, \bar{Z}_j] = 0$ for $i \neq j$ and $\{\bar{X}_i, \bar{Z}_i\} = 0$. We find a similar basis for the logical qubits $\mb{\bar{Q}^{\prime}}$ of $\mc{Q}^{\prime}$ with a logical operator basis $\{(\bar{X}^{\prime}_i, \bar{Z}^{\prime}_i)\}_{i = 1}^{k - |\mb{\bar{Q}_0}|}$, where $\mb{\bar{Q}_0} \subseteq \mb{\bar{Q}}$ denotes the set of logical qubits that are removed. The logical operators of $\mb{\bar{Q}^{\prime}}$ are identified with those of $\mb{\bar{Q}}\backslash \mb{\bar{Q}_0}$ under $\gamma_1$, i.e. 
\begin{equation}
    \gamma_1(\bar{X}^{\prime}_i) = \bar{X}_i, \quad \gamma_1(\bar{Z}^{\prime}_i) = \bar{Z}_i, 
\end{equation}
for $i \in [k - |\mb{\bar{Q}_0}|]$. Note that, again, we have labeled the logical qubits of $\mc{Q}$ such that $\mb{\bar{Q}_0}$ are the last $|\mb{\bar{Q}_0}|$ logical qubits. With a slight abuse of notation, we write $\gamma_1(\bar{Q}^{\prime}_i) = \bar{Q}_{i}$.

The homomorphic CNOT gate (specified by $\gamma_1$) between $\mc{Q}$ and $\mc{Q}^{\prime}$ is implemented by physical transversal $\mc{Q}$-controlled CNOTs between $\mb{Q}\backslash \mb{Q}_0$ and $\mb{Q}^{\prime}$. This amounts to applying logical transversal $\mc{Q}$-controlled CNOTs between $\mb{\bar{Q}}\backslash \mb{\bar{Q}_0}$ and $\mb{\bar{Q}^{\prime}}$. Specifically, the physical CNOTs transform the logical operators of the two codes as follows:
\begin{equation}
\begin{aligned}
    (\bar{X}_{i}, \bar{Z}_{i}) & \rightarrow (\bar{X}_{i} \bar{X}^{\prime}_i, \bar{Z}_{i}), \\
    (\bar{X}^{\prime}_i, \bar{Z}^{\prime}_i) & \rightarrow (\bar{X}^{\prime}_i, \bar{Z}_{i} \bar{Z}^{\prime}_i), \\
\end{aligned}
\end{equation}
for any $i \in [k - |\mb{\bar{Q}_0}|]$, and 
\begin{equation}
    (\bar{X}_i, \bar{Z}_i) \rightarrow (\bar{X}_i, \bar{Z}_i),
\end{equation}
for any $i \in \{k - |\mb{\bar{Q}_0}| + 1 \rightarrow k\}$. Then, using the generalized Steane measurement circuit shown in Fig.~\ref{fig:homomorphic_measurement}(b), we can perform a parallel, non-destructive measurement of the logical $Z$ operators of $\mb{\bar{Q}}\backslash \mb{\bar{Q}_0}$ using the ancilla $\mb{\bar{Q}^{\prime}}$. 

As we will show later, we can also measure certain products of logical $Z$ operators of $\mb{\bar{Q}}$ by also augmenting new checks to the base codes of $\mc{Q}$ when constructing $\mc{Q}^{\prime}$. 

Note that the above homomorphic measurement gadget can be both parallel and selective since, depending on the puncturing and augmenting pattern we perform on the base codes of $\mc{Q}$, we can obtain different ancilla codes with different logical qubits, and thereby different patterns of PPMs on $\mc{Q}$. The number of parallel PPMs being performed equals the number of logical qubits in $\mc{Q}^{\prime}$.

\subsection{Homomorphic measurement gadget for HGP codes \label{sec:homomorphic_measurements_HGP_technical}}
Following the general description in Sec.~\ref{sec:general_homo_gadget} for generic homological product codes, we concretely construct, as an example, homomorphic measurement gadgets for the HGP codes.

We first identify a canonical logical operator basis that defines a canonical set of logical qubits for a generic HGP code. In such a basis, the logical qubits form a 2D grid and, as we will show later, such a basis facilitates the construction of our homomorphic measurement gadget. 

\subsubsection{Hypergraph product codes and their canonical logical operator basis \label{sec:HGP_canonical_basis}}

Let $\mc{C}^1: C^1_1 \xrightarrow{\partial^1_1} C^1_0$ and $\mc{C}^2: C^2_1 \xrightarrow{\partial^2_1} C^2_0$ be two length-$1$ chain complexes. 
As shown in Eq~(\ref{eq:HGP_complex}),
\begin{equation}
    \begin{tikzcd}
    	{Q_0} & & {C^1_0\otimes C^2_0} & \\
    	Q_1 & C^1_0\otimes C^2_1 & & C^1_1\otimes C^2_0 \\
     Q_2 & & C^1_1\otimes C^2_1 & \\
    	\arrow["{H_X}", from=2-1, to=1-1]
            \arrow["{H_Z^T}", from=3-1, to=2-1]
            \arrow["{I\otimes \partial^2_1}", from=2-2, to=1-3]
    	\arrow["{\partial^1_1\otimes I}", from=2-4, to=1-3]
     \arrow["{\partial^1_1\otimes I}", from=3-3, to=2-2]
     \arrow["{I \otimes \partial^2_1}", from=3-3, to=2-4]
    \end{tikzcd}
    \label{eq:HGP_complex}
\end{equation}
we can construct a 2D homological product code $\mc{Q}: Q_2\xrightarrow{H_Z^T} Q_1 \xrightarrow{H_X} Q_0$, also called a HGP code, as the total complex of the tensor product of $\mc{C}^1$ and $\mc{C}^2$.

Specifically, we have: $Q_2 = C^1_1\otimes C^2_1$, whose basis are associated with the $Z$ checks $\mb{S_Z}$; $Q_1 = (C^1_0\otimes C^2_1) \oplus (C^1_1\otimes C^2_0)$, whose basis are associated with the qubits $\mb{Q}$; $Q_0 = C^1_0\otimes C^2_0$, whose basis are associated with the $X$ checks $\mb{S_X}$. 

We assign $\mc{C}^2$ with a $[n_2, k_2, d_2]$ classical code with a check matrix $H_2 \in \mbb{F}_2^{(n_2 - k_2)\times n_2}$ in the standard way. The basis of $C^2_1$ and $C^2_0$ are associated with the bits $\mb{B_2}$ and the checks $\mb{C_2}$ of the code, respectively, and $\partial^2_1 = H_2$. However, to be consistent with the literature, we assign $\mc{C}^1$ with the \emph{transpose} of another $[n_1, k_1, d_1]$ classical code with a check matrix $H_1 \in \mbb{F}_2^{(n_1 - k_1)\times n_1}$. In this case, the basis of $C^1_1$ and $C^1_0$ are associated with the checks $\mb{C_1}$ and the bits $\mb{B}_1$, respectively, and $\partial^1_1 = H_1^T$.

Under such an assignment, we have a $[[n=n_1 n_2 + (n_1 - k_1)(n_2 - k_2), k = k_1 k_2, d = \mr{min}\{d_1, d_2\}]]$ HGP code with check matrices
\begin{equation}
\begin{aligned}
    H_X & = (I_{n_1}\otimes H_2, H^T_1 \otimes I_{n_2 - k_2}), \\
    H_Z & = (H_1\otimes I_{n_2}, I_{n_1 - k_1} \otimes H^T_2).
\end{aligned} 
\end{equation}
Note that we have assumed that both $H_1$ and $H_2$ are full rank. 

Now, we arrange the quantum bits and checks of $\mc{Q}$ on a 2D grid and label them with their coordinates. As shown in Fig.~\ref{fig:homomorphic_measurement}(a), by laying the bits/checks of the first and the second classical code vertically and horizontally, respectively, we have the following identification of the quantum bits and checks with the Cartesian product of the classical bits and checks:
\begin{equation}
\begin{split}
    \mb{Q} & = \mb{B_1}\times \mb{B_2}\cup \mb{C_1}\times \mb{C_2} \\
    & \simeq [n_1]\times[n_2] \cup \{n_1 + 1 \rightarrow n_1 + r_1\}\times \{n_2 + 1 \rightarrow n_2 + r_2\}, \\
    \mb{S_Z} &= \mb{C_1}\times \mb{B_2} \simeq \{n_1 + 1 \rightarrow n_1 + r_1\}\times [n_2], \\
    \mb{S_X} &= \mb{B_1}\times \mb{C_2} \simeq  [n_1]\times \{n_2 + 1 \rightarrow n_2 + r_2\}, \\
\end{split}
\label{eq:coordinate_system}
\end{equation}
where we denote $r_1 := n_1 - k_1$ and $r_2 := n_2 - k_2$ as the number of (independent) checks of $H_1$ and $H_2$, respectively. With these coordinates, each basis element in $Q_2$, $Q_1$, or $Q_0$ is associated with the corresponding object ($Z$ check, qubit, or $X$ check) placed at their assigned coordinate. For instance, the basis element $\vec{e}_i\otimes \vec{e}_j$ in $Q_1$ is associated with a qubit at the coordinate $(i,j)$, which we denote as $Q_{i,j}$.

A complete set of logical $X$ and $Z$ operators of $\mc{Q}$ are given by~\cite{quintavalle2022reshape, quintavalle2022partitioning}
\begin{equation}
\begin{aligned}
    \mb{\bar{X}}&  = \left\{
    \left(\begin{array}{c}
         f\otimes h\\
         0 \\
    \end{array} \right) 
    \mid 
    f \in \ker{H_1}, h \in \mr{rs}({H_2})^{\bullet}
    \right \}, \\
    \mb{\bar{Z}} &  = \left\{
    \left(\begin{array}{c}
         h^{\prime} \otimes f^{\prime}\\
         0 \\
    \end{array} \right) 
    \mid 
    h^{\prime} \in \mr{rs}({H_1})^{\bullet}, f^{\prime} \in \ker{H_2} 
    \right \},
\end{aligned}
\label{eq:HGP_logicals}
\end{equation}
where $\mr{rs}({H_{\alpha}})^{\bullet}$ denotes the complementary space~\footnote{the space $V$ that satisfies $V\cap \rm{rs}(H_{\alpha}) = 0$ and $\mbb{F}_2^{n_{\alpha}} = V \oplus \rm{rs}(H_{\alpha})$} of $\mr{rs}({H_{\alpha}})$ in $\mathbb{F}_2^{n_{\alpha}}$ for $\alpha = 1, 2$. 

We now find a canonical basis for the logical operators in Eq.~\eqref{eq:HGP_logicals} such that they form conjugate pairs and they admit a canonical form analogous to that of a classical code in Eq.~\eqref{eq:classical_canonical_generator_mat}. In addition, the logical qubits in such a canonical basis will be arranged on a 2D grid.
Without loss of generality, we assume that each $H_{\alpha}$ is row-reduced to their canonical form,
\begin{equation}
    H_{\alpha,c} = (h^{\alpha}_1, h^{\alpha}_2, \cdots, h^{\alpha}_{k_{\alpha}}, I_{n_{\alpha} - k_{\alpha}}),
\end{equation}
via elementary row operations, where $h^{\alpha}_{j} \in \mathbb{F}_2^{n_{\alpha} - k_{\alpha}}$. Note that this is always the case if we permute the information bits of each classical code to the first few bits (see Sec.~\ref{sec:puncturing_shortening}). Using the canonical form of $H_{\alpha}$, we can derive their generator matrices whose rows span $\ker{H_{\alpha}}$:
\begin{equation}
        G_{\alpha} = \left( I_{k_{\alpha}}, 
        \begin{array}{c}
             h_1^{\alpha T}\\
             h_2^{\alpha T}\\
             \cdots\\
             h_{k_{\alpha}}^{\alpha T}
        \end{array}
        \right).
\end{equation}
Let $e_j^n$ denote a unit $n$-dimensional column vector with the $j$-th entry being $1$, and let $b^{\alpha}_j := \left( \begin{array}{c}
e_j^{k_{\alpha}} \\
h_j^{\alpha}
\end{array} \right)$. Then we have 
\begin{equation}
\begin{aligned}
     \ker{H_{\alpha}} & = \mr{rs}(G_{\alpha}) = \mr{span} \{b^{\alpha}_j\}_{j \in [k_{\alpha}]}, \\
     \mr{rs}({H_{\alpha}})^{\bullet} & = \mr{rs}({H_{\alpha, c}})^{\bullet} = \mr{span} \{e^{n_{\alpha}}_j\}_{j \in [k_{\alpha}]}
\end{aligned}
\label{eq:ker_im}
\end{equation}

Substituing Eq.~\eqref{eq:ker_im} into Eq.~\eqref{eq:HGP_logicals}, we can find a canonical basis for $\mb{\bar{X}}$ and $\mb{\bar{Z}}$ forming conjugating pairs $\{(\bar{X}_{i,j}, \bar{Z}_{i,j})\}_{i \in [k_1], j \in [k_2]}$, where:
\begin{equation}
    \bar{X}_{i,j} = \left(\begin{array}{c}
         b^1_i \otimes e^{n_2}_j \\
         0 \\
    \end{array} \right), \quad 
    \bar{Z}_{i,j} = \left(\begin{array}{c}
         e^{n_1}_i \otimes b^2_j \\
         0 \\
    \end{array} \right).
    \label{eq:HGP_canonical_logicals}
\end{equation}
It is easy to verify the commutation relation for the paired logical operators in Eq.~\eqref{eq:HGP_canonical_logicals} since
\begin{equation}
    \bar{X}_{i,j}^T \bar{Z}_{i,j} = (b^{1 T}_i e^{n_1}_{i^{\prime}})(e^{n_2 T}_{j} b^2_{j^{\prime}}) = \delta_{i,i^{\prime}} \delta_{j,j^{\prime}} \mod 2.
\end{equation}

A key feature for such a canonical basis in Eq.~\eqref{eq:HGP_canonical_logicals} is that each pair of logical operators $(\bar{X}_{i,j}, \bar{Z}_{i,j})$ are supported on the $j$-th column and $i$-th row of the qubit grid, respectively, and they overlap at exactly one physical qubit with coordinate $(i,j)$ on the $[k_1]\times [k_2]$ subgrid. See Eq.~\eqref{eq:coordinate_system} for the coordinate system and Fig.~\ref{fig:homomorphic_measurement}(a) for an example of $\bar{Q}_{1,1}$ and $\bar{Q}_{2,3}$. As such, we can think of the set of logical qubits $\mb{\bar{Q}} = \{\bar{Q}_{i,j}\}_{i \in [k_1], j \in [k_2]}$ as arranged on a $[k_1] \times [k_2]$ 2D grid, which are in one-to-one correspondence with the $[k_1]\times  [k_2]$ subgrid of physical qubits (see the upper left physical block in Fig.~\ref{fig:homomorphic_measurement}(a)).

\subsubsection{A basic homomorphic measurement gadget \label{sec:horizontal_PPMs}}

Let $\mc{Q} = \mr{HGP}(H_1, H_2)$ be a HGP code with a $[n_1, k_1, d_1]$ base vertical code and a $[n_2, k_2, d_2]$ base horizontal code. Recall that in the canonical logical operator basis (see Eq.~\eqref{eq:HGP_canonical_logicals}), the logical qubits of $\mc{Q}$ are arranged on a $[k_1]\times [k_2]$ grid, i.e. $\mb{\bar{Q}} = \{\bar{Q}_{i,j}\}_{i \in [k_1], j \in [k_2]}$. Let $\mb{\mc{E}_h} = \{\mb{e_i}\}$ be a collection of disjoint horizontal hyperedges, i.e., $\mb{e_i} \subseteq [k_2]$ and $\mb{e_i}\cap \mb{e_j} = \emptyset$ for $i \neq j$. We present a gadget for measuring $k_1\times |\mb{\mc{E}_h}|$ logical Pauli products, which are horizontal Pauli products specified by $\mb{\mc{E}_h}$ across all the $k_1$ rows: $\bigcup_{i \in [k_1]} \{\bigotimes_{j \in \mb{e}}\bar{Z}_{i,j}\}_{\mb{e} \in \mb{\mc{E}_h}}$, in parallel. 


To simplify the notation, let $H_{\mr{rep}}^{\mb{\mc{E}_h}}$ denote the check matrix of $|\mb{\mc{E}}_h|$ disjoint repetition codes, each supported on $\mb{e}_i \in \mb{\mc{E}_h}$, i.e., $H_{\mr{rep}}^{\mb{\mc{E}_h}}$ is the vertical concatenation of $\{H_i \in \mbb{F}_2^{(|\mb{e_i}| - 1)\times n_2}\}$, where the rows of $H_i$ are simply parities of any two bits of $\mb{e_i}$. For example, for $n_2 = 6$ and $\mb{\mc{E}_h} = \{\{1,2,3\}, \{4,5\}\}$, we have 
\begin{equation}
    H_{\mr{rep}}^{\mb{\mc{E}_h}} = \left(\begin{array}{cccccc}
     1 & 1 & 0 & 0 & 0 & 0 \\
     0 & 1 & 1 & 0 & 0 & 0 \\
     0 & 0 & 0 & 1 & 1 & 0 \\
\end{array}\right).
\end{equation}

We present the horizontal PPMs gadget in Alg.~\ref{alg:HGP_horizontal_PPMs}. The key is to construct an ancilla code by puncturing and augmenting the original horizontal classical code according to $\mc{E}_h$. Specifically, we puncture on all the bits not in $\mc{E}_h$ and augment repetition codes supported on each of the hyperedges in  $\mc{E}_h$. This removes all the columns of logical qubits with indices out of $\mc{E}_h$ and merges columns of logical qubits indexed by each of the hyperedges in $\mc{E}_h$, which guarantees that only Pauli products associated with $\mc{E}_h$ are measured using the merged logical qubits in the ancilla code. In Theorem~\ref{theorem:horizontal_PPMs_FT}, we prove that the gadget implements the desired logical measurements and is fault-tolerant. 

\begin{algorithm}[h!]
\caption{Horizontal PPMs for HGP code}\label{alg:HGP_horizontal_PPMs}
\Input{An HGP code $\mc{Q} = \mr{HGP}(H_1, H_2)$ with a $[n_1, k_1, d_1]$ base vertical code and a $[n_2, k_2, d_2]$ base horizontal code; A collection of disjoint horizontal hyperedges $\mb{\mc{E}_h} = \{\mb{e_i}\}$, where $\mb{e_i} \subseteq [k_2]$. }
\Output{A horizontal PPMs gadget.}
Let $\mb{S} = [k_2] \backslash (\bigcup_{\mb{e} \in \mb{\mc{E}_h}} \mb{e})$, and $H_{\mr{Rep}}^{\mb{\mc{E}_h}} \in \mbb{F}_2^{r\times n_2}$ the check matrix of the $|\mb{\mc{E}_h}|$ disjoint repetition codes, each supported on $\mb{e}_i \in \mb{\mc{E}_h}$. 
Let $H_2^{\prime} = H_2^{+H_{\mr{Rep}}^{\mb{\mc{E}_h}}, \mb{S}}$ be the check matrix obtained by first augmenting the checks in $H_{\mr{Rep}}^{\mb{\mc{E}_h}}$ and then puncturing on $\mb{S}$. 
Construct an ancilla code $\mc{Q}^{\prime} = \mr{HGP}(H_1, H_2^{\prime})$ with qubits $\mb{Q^{\prime}}$. \\
\tcp{Construct an ancilla code by performing appropriate puncturing and augmenting on the horizontal classical code}

Let $\mb{O}$ denote the coordinates of the qubits $\mb{Q}$ of $\mc{Q}$, $\mb{O}_0 = [n_1]\times \mb{S}$, and $\mb{O}_1 = \{n_1 + 1 \rightarrow n_1 + (n_1-k_1)\}\times \{n_2 + (n_2 - k_2) + 1 \rightarrow n_2 + (n_2 - k_2) + r\}$. We have $\mb{Q} \simeq \mb{O}$ and $\mb{Q^{\prime}} \simeq \mb{O}\backslash \mb{O}_0 \cup \mb{O}_1$. \\
\tcp{See Fig.~\ref{fig:HGP_Grid_PPMs}(b) for an illustration of the coordinates and the extra qubits outside the solid box in the middle panel for an illustration of $\mb{O}_1$.}

Prepare $\mb{Q}^{\prime}$ in $\ket{0}^{\otimes|\mb{Q}^{\prime}|}$ and measure the stabilizers of $\mc{Q}^{\prime}$ for $d = \min\{d_1, d_2\}$ rounds. \\

Apply transversal $\mc{Q}$-controlled CNOTs between $\mb{Q}|_{\mb{O}\backslash \mb{O_0}}$ and $\mb{Q^{\prime}}|_{\mb{O}\backslash \mb{O_0}}$. \\

Transversally measure $\mb{Q}^{\prime}$ in the $Z$ basis.\\
\tcp{Implement the generalized Steane measurement circuit in Fig.~\ref{fig:homomorphic_measurement}(b)}
\end{algorithm}

\begin{theorem}
The horizontal-PPMs gadget in Alg.~\ref{alg:HGP_horizontal_PPMs} performs logical measurements of $k_1\times |\mb{\mc{E}_h}|$ PPMs, $\bigcup_{i \in [k_1]} \{\bigotimes_{j \in \mb{e}}\bar{Z}_{i,j}\}_{\mb{e} \in \mb{\mc{E}_h}}$, of a HGP code in parallel and fault-tolerantly.
\label{theorem:horizontal_PPMs_FT}
\end{theorem}
\begin{proof}
Let $H_{2,c} = \left(h_1, h_2, \cdots, h_{k_2}, I_{n_2-k_2}\right)$ be the canonical form of $H_2$. The corresponding generator matrix is then $G_2 = \left( I_{k_{2}}, 
        \begin{array}{c}
             h_1^{T}\\
             h_2^{T}\\
             \cdots\\
             h_{k_{2}}^{T}
        \end{array}
        \right).$
Without loss of generality, assume that the hyperedges in $\mb{\mc{E}_h}$ are ordered, i.e. 
$\mb{\mc{E}_h} = \{\mb{e_i}\}_{i = 1}^t$, where $\mb{e_i} = \{\sum_{j = 1}^{i-1}|\mb{e_j}| + 1 \rightarrow \sum_{j = 1}^{i}|\mb{e_j}|\}$.
According to Def.~\ref{def:augmenting_oper}, the generator matrix of $H_2^{+H_{\mr{rep}}^{\mb{\mc{E}_h}}}$, $G_{2, -H_{\mr{rep}}^{\mb{\mc{E}_h}}}$, is given by removing the codewords in $G_2$ that do not satisfy the extra constraints imposed by the augmented checks $H_{\mr{rep}}^{\mb{\mc{E}_h}}$. It is straightforward to find $G_{2, -H_{\mr{rep}}^{\mb{\mc{E}_h}}} \in \mbb{F}_2^{(t + |\mb{S}|)\times n_2}$ and
\begin{equation}
    G_{2, -H_{\mr{rep}}^{\mb{\mc{E}_h}}}[i,] = \sum_{j \in \mb{e_i}} G_2[j,],
\end{equation}
for $i \in [t]$, and
\begin{equation}
    G_{2, -H_{\mr{rep}}^{\mb{\mc{E}_h}}}[i,] = G_2[\sum_{j=1}^t |\mb{e}_j| + i,]
\end{equation}
for $i \in \{t + 1 \rightarrow t + |\mb{S}|\}$. Note that we denote $M[j,]$ as the $j$-th row of a matrix $M$. Namely, each of the first $t$ codewords of $G_{2, -H_{\mr{rep}}^{\mb{\mc{E}_h}}}$ are given by merging the codewords of $G_2$ indexed by $\mb{e}_i$ and the last $|\mb{S}|$ codewords of $G_2$ are preserved. Finally, according to Def.~\ref{def:puncturing_oper}, the generator matrix of $H_2^{\prime} := H_2^{+H_{\mr{rep}}^{\mb{\mc{E}_h}}, \mb{S}}$, $G_2^{\prime}$, is given by shortening $G_{2, -H_{\mr{rep}}^{\mb{\mc{E}_h}}}$ on $\mb{S}$, which simply corresponds to removing the last $|\mb{S}|$ codewords of $G_{2, -H_{\mr{rep}}^{\mb{\mc{E}_h}}}$ and deleting the columns indexed by $\mb{S}$, i.e., $G_2^{\prime} \in \mbb{F}_2^{t\times (n_2 - |\mb{S}|)}$ and
\begin{equation}
     G_2^{\prime}[i,] = \sum_{j \in \mb{e_i}} G_2[j,]|_{[n_2]\backslash \mb{S}}. \label{eq:generator_mat_new_code}
\end{equation}
Let $\ker{H_2} = \mr{span}\{b^2_i\}_{i \in [k_2]}$, where $b^2_i = G_2[i,]^T$, and $\mr{rs}(H_2)^{\bullet} = \mr{span}\{e_i^{n_2}\}_{i \in [k_2]}$. According to Eq.~\eqref{eq:generator_mat_new_code}, we have
\begin{equation}
    \ker{H_2^{\prime}} = \mr{rs}(G_2^{\prime}) = \mr{span}\{(\sum_{j \in \mb{e}_i}b^2_j)|_{[n_2]\backslash \mb{S}} \}_{i \in [t]}.
\end{equation}
In addition, one can derive that 
\begin{equation}
    \mr{rs}(H_2^{\prime})^{\bullet} = \rm{span}\{(e^{n_2}_{\sum_{j=1}^{i-1}|\mb{e}_j| + 1})|_{[n_2]\backslash \mb{S}}\}_{i \in [t]}.
\end{equation}
Then, according to the canonical logical operator basis in Eq.~\eqref{eq:HGP_canonical_logicals}, the logical operators of $\mc{Q}'$ are identified with those of $\mc{Q}$ via their shared coordinates:
\begin{equation}
\begin{split}
    \bar{X}^{\prime}_{i,j} & \simeq \bar{X}_{i,\sum_{k=1}^{j-1}|\mb{e}_j| + 1},\\
    \bar{Z}^{\prime}_{i,j} & \simeq \bigotimes_{k \in \mb{e_j}} \bar{Z}_{i, k}
\end{split}
\end{equation}

According to Proposition~\ref{prop:homological_prod_code_homo_by_classical} and~\ref{prop:puncturing_augmenting_homo}, we can construct a homomorphism $\gamma = \{\gamma_2, \gamma_1, \gamma_0\}: \mc{Q}^{\prime} \rightarrow \mc{Q}$ induced by the classical puncturing-augmenting homomorphism from $H^{\prime}$ to $H$, where $\gamma_1$ identifies qubits with the same coordinates, i.e.  $\gamma_1: \mb{Q}^{\prime}|_{\mb{O}\backslash \mb{O}_0} \rightarrow \mb{Q}|_{\mb{O}\backslash \mb{O}_0}$. Finally, the logical action associated with the transversal CNOTs specified by $\gamma_1$ can be calculated:
\begin{equation}
    \bar{Z}^{\prime}_{i,j} \rightarrow \bar{Z}^{\prime}_{i,j}\bigotimes_{k \in \mb{e_j}} \bar{Z}_{i, k},
    \label{eq:logical_action_Z}
\end{equation}
for any $(i,j) \in [k_1]\times [t]$, and 
\begin{equation}
    \bar{X}_{i,k} \rightarrow \bar{X}_{i,k} \bar{X}^{\prime}_{i,j},
    \label{eq:logical_action_X}
\end{equation}
up to some stabilizers, for any $(i,j) \in [k_1]\times [t]$ and $k \in \mb{e_j}$. Other logical operators of $\mc{Q}^{\prime}$ and $\mc{Q}$ remain unchanged.

Eq.~\eqref{eq:logical_action_Z} and Eq.~\eqref{eq:logical_action_X} indicate that the constructed logical gate implements the logical CNOTs between each logical qubit $\bar{Q}^{\prime}_{i,j}$  and $\{\bar{Q}_{i,k}\}_{k \in \mb{e}_j}$. Therefore, the Steane measurement circuit utilizing such logical CNOTs implements the logical PPMs $\bigcup_{i \in [k_1]} \{\bigotimes_{j \in \mb{e}}\bar{Z}_{i,j}\}_{\mb{e} \in \mb{\mc{E}_h}}$ in parallel.

Finally, we show that the entire protocol is fault-tolerant. Since we only puncture $H_2$ on the information bits, the distance of $H_2^{\prime}$ is preserved, i.e. $d_2^{\prime} \geq d_2$ (see Sec.~\ref{sec:puncturing_shortening}). Since $H_2^{\prime}$ is still full-rank, the distance of $\mc{Q}^{\prime}$ satisfies $d^{\prime} = \min\{d_1, d_2^{\prime}\} \geq d = \mr{min}\{d_1, d_2\}$. Now, the logical Steane measurement circuit is clearly fault-tolerant as it uses a large-distance ancilla code and only involves fault-tolerant gadgets: computational basis state preparation with $d$ code cycles, transversal logical gates, and transversal readout. 
\end{proof}

\section{Parallelizable logical computation with hypergraph product codes \label{sec:HGP_codes}}
In this section, we present new schemes for realizing parallel logical computations using only HGP codes. The new schemes leverage the transversal (homomorphic) CNOTs between two \emph{distinct} HGP codes, and the related selective homomorphic measurements, that we constructed in Sec.~\ref{sec:sketch_homomorphic_gadget}. Although this section aims to be self-contained and only quote the results in Sec.~\ref{sec:sketch_homomorphic_gadget} in a non-technical way, readers are encouraged to quickly go through Sec.~\ref{sec:sketch_homomorphic_gadget}, particularly Sec.~\ref{sec:homomorphic_measurements_HGP_technical}, to get familiar with the notation and high-level ideas.

In Sec.~\ref{sec:GPPMs}, we first put the coding-theoretical and physical-level results in Sec.~\ref{sec:sketch_homomorphic_gadget} together and obtain a new fault-tolerant logical gadget for a generic HGP code, that performs selective PPMs on any subgrid of the logical qubits in parallel. Then, in Sec.~\ref{sec:HGP_computation}, we consider logical-level computations with HGP codes by utilizing the new parallel PPMs gadget, combined with known logical transversal/fold transversal gates.

\begin{figure*}[hbt!]
    \centering
    \includegraphics[width=1\textwidth]{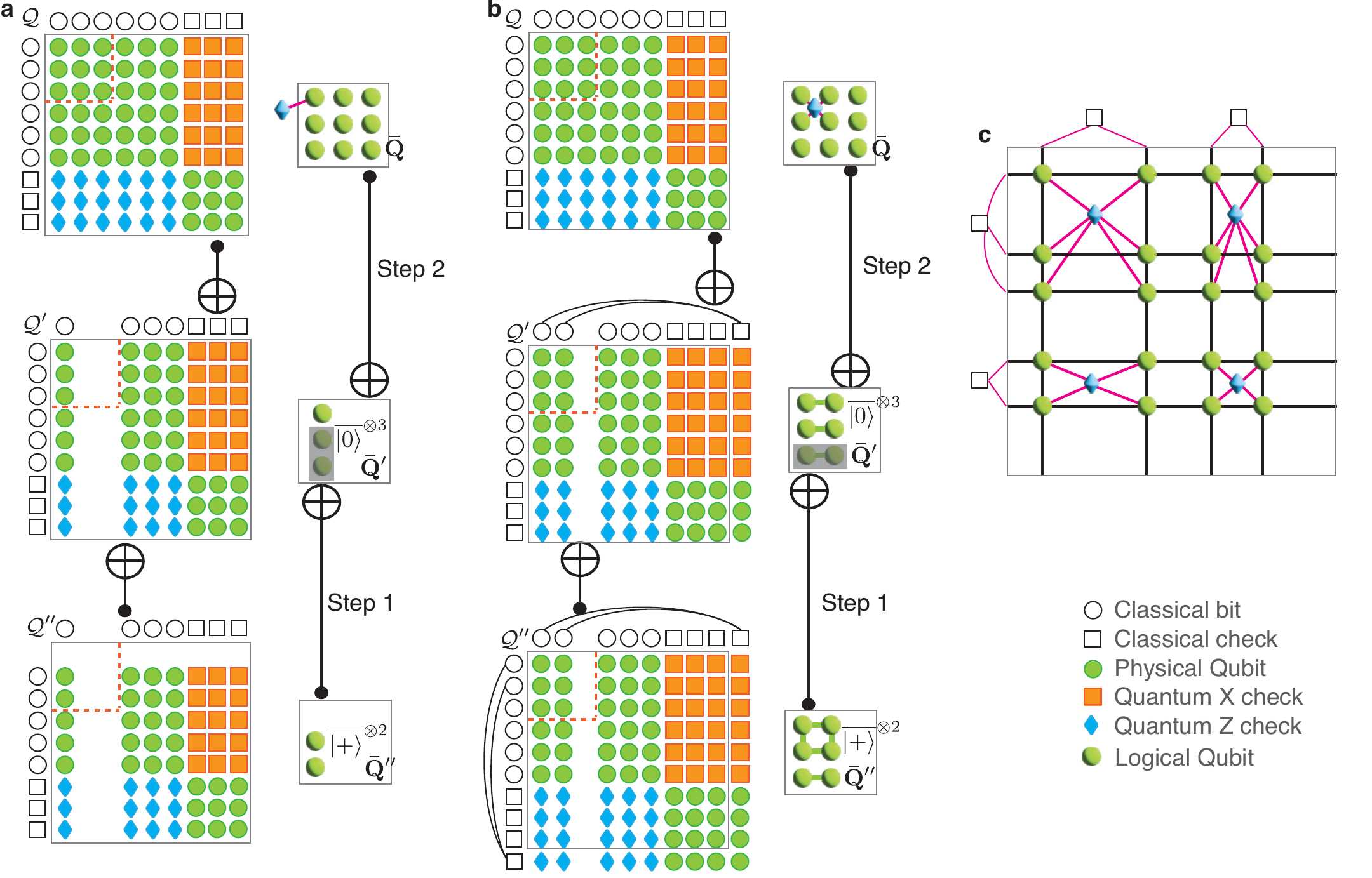}
    \caption{\textbf{Illustration of parallel Grid Pauli product measurements (GPPMs) for HGP codes.} (a) An example for measuring the single-qubit $Z$ operator of the top-left logical qubit in $\mb{\bar{Q}}$. 
    An ancilla code $\mc{Q}^{\prime}$ is constructed by puncturing on the second and the third bits of the horizontal base code of $\mc{Q}$; another mask code $\mc{Q}^{\prime \prime}$ is constructed by puncturing on the first bit of the vertical base code of $\mc{Q}^{\prime}$.
    In the first step, the bottom two logical qubits of $\mc{Q}^{\prime}$ are reset to $\overline{\ket{+}}$ using the mask code $\mc{Q}^{\prime \prime}$ and the $\mc{Q}^{\prime \prime}$-controlled homomorphic CNOT; in the second step, the desired $Z$ measurement on the top-left logical qubit of $\mc{Q}$ is obtained using $\mc{Q}^{\prime}$ and the $\mc{Q}$-controlled homomorphic CNOT. (b) Another example for measuring the weight-four $Z$ operator supported on the top-left four logical qubits in $\mb{\bar{Q}}$. An ancilla code $\mc{Q}^{\prime}$ is constructed by first puncturing on the third bit of the horizontal base code of $\mc{Q}$ and then augmenting a check that connects the first two bits; another mask code $\mc{Q}^{\prime \prime}$ is constructed by augmenting a check that connects the first two bits of the vertical base code of $\mc{Q}^{\prime}$. In the first step, the bottom logical qubit of $\mc{Q}^{\prime}$ is reset to $\overline{\ket{+}}$ while the top two logical qubits are prepared to a Bell state using $\mc{Q}^{\prime \prime}$ and the $\mc{Q}^{\prime \prime}$-controlled homomorphic CNOT; in the second step, the desired weight-four $Z$ operator on the top-left four logical qubits of $\mc{Q}$ is measured using the $\mc{Q}^{\prime}$ and the $\mc{Q}$-controlled homomorphic CNOT.  (c) Illustration of a general GPPMs gadget that measures a grid pattern of PPMs on a subgrid of the logical qubits. The PPMs (the 3D diamonds) are specified by the product of two collections of hyperedges, where each hyperedge (the empty squares) is a set of rows or columns. }
    \label{fig:HGP_Grid_PPMs}
\end{figure*}

\subsection{Grid Pauli product measurements for HGP codes \label{sec:GPPMs}}
Here, we present the construction of a homomorphic measurement gadget for HGP codes that measures Pauli products on any subgrid of the logical qubits selectively and in parallel. We refer to such a gadget as a ``Grid PPMs" (GPPMs) gadget (see Def.~\ref{def:GPPMs}) and will use this as an elementary gadget for executing logical computations with HGP codes later in this work. 



The GPPMs gadget is essentially built upon the homomorphic measurement gadget we introduced in Sec.~\ref{sec:horizontal_PPMs} (see Alg.~\ref{alg:HGP_horizontal_PPMs}).
Let $\mc{Q}$ be an HGP code constructed by taking the hypergraph product of a vertical $[n_1, k_1, d_1]$-classical code $H_1$ and a $[n_2, k_2, d_2]$-horizontal classical code $H_2$.
The horizontal $Z$-PPMs gadget measures horizontal $Z$-type PPMs, $\bigcup_{i \in [k_1]} \{\bigotimes_{j \in \mb{e}}\bar{Z}_{i,j}\}_{\mb{e} \in \mb{\mc{E}_h}}$, specified by a set of horizontal hyperedges $\mb{\mc{E}_h}$ across all the rows (see the horizontal empty squares in Fig.~\ref{fig:HGP_Grid_PPMs}(c) for an illustration). 
The key of the gadget is to first construct an ancilla code $\mc{Q}^{\prime}$ by puncturing some bits of $H_2$, thereby removing any action on the corresponding qubits, and then augmenting new checks to it, thereby measuring Pauli products given by the checks.
We then implement the homomorphic $\mc{Q}$-controlled CNOT between $\mc{Q}$ and $\mc{Q}^{\prime}$ and measure $\mc{Q}^{\prime}$ to complete the PPM.

We can similarly construct a vertical PPMs gadget that measures vertical $X$-PPMs specified by a set of vertical hyperedges $\mb{\mc{E}_v}$ across all the columns. We can simply achieve this by performing the same puncturing-augmenting operation on the vertical code $H_1$ of $\mc{Q}$, instead of the horizontal code $H_2$. Note, however, that the quantum code homomorphism $\gamma = \{\gamma_2, \gamma_1, \gamma_0\}: \mc{Q} \rightarrow \mc{Q}^{\prime}$ (see Def.~\ref{def:homomorphic_CNOT} and Sec.~\ref{sec:horizontal_PPMs}) now reverses direction, since the base vertical complex $\mc{C}_1$ and $\mc{C}_1^{\prime}$ are associated with the \emph{transpose} of the corresponding vertical codes, and the puncturing-augmenting operation of $H_1$ actually induces a puncturing-augmenting homomorphism from $\mc{C}_1$ to $\mc{C}_1^{\prime}$. 
This is because if $H_1^{\prime}$ is obtained by puncturing (augmenting) $H_1$, then $H_1^T$ is obtained by augmenting (puncturing) $H_1^{\prime T}$. Consequently, the homomorphic CNOT between $\mc{Q}$ and $\mc{Q}^{\prime}$ now reverses direction and becomes $\mc{Q}^{\prime}$-controlled both physically and logically. As such, we should prepare the logical qubits of $\mc{Q}^{\prime}$ in the $X$ basis, apply $\mc{Q}^{\prime}$-controlled homomorphic CNOTs, and finally measure $\mc{Q}^{\prime}$ transversely in the $X$ basis. Then, this gadget implements the logical PPMs $\bigcup_{j \in [k_2]} \{\bigotimes_{i \in \mb{e}}\bar{X}_{i,j}\}_{\mb{e} \in \mb{\mc{E}_v}}$ in parallel.

In the following, we show how to construct a GPPMs gadget that measures PPMs on any subgrid of the logical qubits $[k_1]\times [k_2]$, which could contain PPMs across different columns and rows, by combining the horizontal $Z$-PPMs gadget and the vertical $X$-PPMs gadget.

As the simplest example, we show how to perform a $Z$ measurement for only the top-left $(1,1)$ logical qubit. As illustrated in Fig.~\ref{fig:HGP_Grid_PPMs}(a), by puncturing the horizontal code on all but the very first bit, we can simultaneously remove all but the first column of the logical qubits of $\mc{Q}$ when constructing the ancilla code $\mc{Q}^{\prime}$. Applying the horizontal PPMs gadget with horizontal hyperedges $\mb{\mc{E}_h} = \{\{1\}\}$ would then perform single-qubit $Z$ measurements on the first-column logical qubits of $\mc{Q}$. 

To only measure $\bar{Q}_{1,1}$,  we introduce another code $\mc{Q}^{\prime \prime}$, which we call the mask code, that is obtained by puncturing the vertical code of $\mc{Q}^{\prime}$ on the first bit. As shown in Fig.~\ref{fig:HGP_Grid_PPMs}(a), doing so removes the first logical qubit of $\mc{Q}^{\prime}$. We can then perform single-qubit $X$ measurements on all but the first logical qubit of $\mc{Q}^{\prime}$ with the mask code $\mc{Q}^{\prime \prime}$ by utilizing the vertical PPMs gadget with hyperedges $\mb{\mc{E}_v} = \{\{i\}\}_{i \in [k_1]\backslash [1]}$. 
Combining these elements together, we obtain a two-step protocol for measuring the $(1,1)$ logical qubit of $\mc{Q}$:
\begin{enumerate}
    \item Prepare all the logical qubits of $\mc{Q}^{\prime}$ and $\mc{Q}^{\prime \prime}$ in $\overline{\ket{0}}$ and $\overline{\ket{+}}$ states, respectively. Then perform vertical $X$-PPMs on $\mc{Q}^{\prime}$ with $\mb{\mc{E}_v} = \{\{i\}\}_{i \in [k_1]\backslash [1]}$ using $\mc{Q}^{\prime \prime}$, which resets all but the first logical qubit of $\mc{Q}^{\prime}$ into $\overline{\ket{+}}$ states.
    \item Perform a horizontal $Z$-PPMs on $\mc{Q}$ with $\mb{\mc{E}_h} = \{\{1\}\}$ using $\mc{Q}^{\prime}$. Note that only $\bar{Q}_{1,1}$ is meausred since the rest of the first-column logical qubits of $\mc{Q}$ are not entangled with the corresponding logical qubits in $\mc{Q}^{\prime}$, which are reset to $\overline{\ket{+}}$ by the first homomorphic measurement, via the homomorphic CNOT gate. 
\end{enumerate}

Obviously, the above scheme can be readily generalized to performing parallel single-qubit $Z$ measurements on a subgrid of the logical grid, i.e. logical qubits supported on the intersection of any subset of rows $\mb{\mr{Rows}} = \{r_i\}$ and columns $\mb{\mr{Cols}} = \{c_j\}$. This can be done by puncturing $H_2$ on $[k_2]\backslash \mb{\mr{Rows}}$ and puncturing $H_1^{\prime}$ on $\mb{\mr{Cols}}$ when constructing $\mc{Q}^{\prime}$ and $\mc{Q}^{\prime \prime}$, respectively. 

Moreover, we can generalize the above scheme to measure some pattern of Pauli \textit{product} measurements (PPMs) on a subgrid of the logical grid. As an example, we show how to measure a weight-$4$ Pauli product $\bar{Z}_{1,1}\bar{Z}_{1,2}\bar{Z}_{2,1}\bar{Z}_{2,2}$ supported on the intersection of the first two columns and rows in Fig.~\ref{fig:HGP_Grid_PPMs}(b). We first construct an ancilla patch $\mc{Q}^{\prime}$ by puncturing $H_2$ on all but the first two bits and then augmenting an extra check that checks the first two bits. Doing so removes all but the first two columns of logical qubits of $\mc{Q}$. In addition, the remaining two columns of logical qubits merge into one column by the augmented checks, i.e. the logical operators of each pair of logical qubits in each row are now equivalent up to some stabilizers in $\mc{Q}^{\prime}$. Now, applying the horizontal $Z$-PPMs gadget with $\mb{\mc{E}_h} = \{\{1,2\}\}$ between $\mc{Q}$ and $\mc{Q}^{\prime}$ would perform weight-$2$ $\bar{Z}_{i,1}\bar{Z}_{i,2}$ measurements across all the rows on $\mc{Q}$.
To obtain the desired weight-$4$ measurement, we again introduce another mask code $\mc{Q}^{\prime \prime}$. We construct $\mc{Q}^{\prime \prime}$ by augmenting an extra check to $H_1^{\prime}$ that checks the first two bits. This merges the first two logical qubits in $\mc{Q}^{\prime \prime}$. Then, applying the vertical $X$-PPMs gadget with $\mb{\mc{E}_v} = \{\{1,2\}, \{3\}, \{4\}, \cdots, \{k_1\}\}$ between $\mc{Q}^{\prime}$ and $\mc{Q}^{\prime \prime}$ performs a $\bar{X}^{\prime}_{1,1}\bar{X}^{\prime}_{2,1}$ measurement on the first two logical qubits of $\mc{Q}^{\prime}$ and single-qubit $X$ measurements on remaining logical qubits. This creates a Bell pair for the first two while masking (resetting to $\overline{\ket{+}}$ states) the rest of the logical qubits. Finally, the horizontal $Z$-PPMs gadget between $\mc{Q}$ and $\mc{Q}^{\prime}$ performs the desired $\bar{Z}_{1,1}\bar{Z}_{1,2}\bar{Z}_{2,1}\bar{Z}_{2,2}$ measurement. 

Finally, we can generalize the above protocol and perform the following pattern of PPMs in parallel on a HGP code:

\begin{definition}[Grid PPMs]
Given a canonical representation of logical qubits in an HGP code on a 2D grid $[k_1]\times[k_2]$, and two sets of disjoint hyperedges $\mb{\mc{E}_r}$ and $\mb{\mc{E}_c}$, where each hyperedge is a collection of rows or columns, respectively, then we define a pattern of Grid PPMs of a type $P \in \{X, Z\}$ as:
\begin{equation}
    \mr{GPPMs}(\mr{P}, \mb{\mc{E}_r}, \mb{\mc{E}_c}) := \{\bigotimes_{q \in \mb{e_r} \times \mb{e_c}} \bar{P}_q \}_{\mb{e_r} \in  \mb{\mc{E}_r}, \mb{e_c} \in  \mb{\mc{E}_c}},
\end{equation}
where $\bar{P}_q$ denotes the logical Pauli operator $\bar{X}$ ($\bar{Z}$) of the logical qubits $\bar{Q}_q$ with coordinates $q$ for $P = X$ ($Z$).
\label{def:GPPMs}
\end{definition}

For completeness, we provide the concrete protocol for implementing a $Z$-type GPPMs in Alg.~\ref{alg:HGP_grid_PPMs}. Note that step 2 of Alg.~\ref{alg:HGP_grid_PPMs} merges columns of logical qubits of $\mc{Q}^{\prime}$ according to $\mb{\mc{E}_c}$ and step 4 further prepares these merged logical qubits into vertical disjoint GHZ states according to $\mb{\mc{E}_r}$ across all the columns. Together, they make sure that the final step measures the desired PPMs on disjoint grids according to $\mb{\mc{E}_r}\times \mb{\mc{E}_c}$. Note that we can construct a GPPMs gadget measuring $X$-type PPMs using a two-step protocol similar to Alg.~\ref{alg:HGP_grid_PPMs}, except that we now construct the ancilla code and the mask code by puncturing-augmenting the vertical classical code and the horizontal classical code, respectively, and consequently, the homomorphic CNOTs, e.g. in Fig.~\ref{fig:HGP_Grid_PPMs}(a), all reverse direction.

\begin{algorithm}[h!]
\caption{Grid PPMs for HGP code}\label{alg:HGP_grid_PPMs}
\Input{An HGP code $\mc{Q} = \mr{HGP}(H_1, H_2)$; Two sets of hyperedges $\mb{\mc{E}_r}$ and $\mb{\mc{E}_c}$}
\Output{A Grid PPMs gadget $\mr{GPPMs}(Z, \mb{\mc{E}_r}, \mb{\mc{E}_c)}$.}
\tcp{Note that classical Pauli frame updates based on the measurement outcomes are neglected here for simplicity.}

Let $\mb{\mc{E}_h} = \mb{\mc{E}_c}$ and $\mb{S}_h = [k_2]\backslash (\bigcup_{\mb{e_h} \in \mb{\mc{E}_h}} \mb{e_h})$. Construct $\mb{\mc{E}_v}$ and $\mb{S}_v$ as follows:
For each $\mb{e_v} \in \mb{\mc{E}_r}$, if $|\mb{e_v}| = 1$, append $\mb{e_v}$ to $\mb{S}_v$; Otherwise ($|\mb{e_v}| \geq 2$), append the parities of their components, i.e. $\{\{\mb{e_v}[i], \mb{e_v}[i + 1]\}\}_{i \in [|\mb{e}_v| - 1]}$, to $\mb{\mc{E}_v}$. \\
\tcp{Set the checks ($\mb{\mc{E}_h}$ and $\mb{\mc{E}_v}$) to be augmented and the bits ($\mb{S}_h$ and $\mb{S}_v$) to be punctured for the horizontal and vertical codes, respectively.}

Construct an ancilla code $\mc{Q}^{\prime} = \mr{HGP}(H_1^{\prime}, H_2^{\prime})$ by augmenting checks in $H_{\mr{rep}}^{\mb{\mc{E}_h}}$ (with length $n_2$) to $H_2$ and then puncturing on $\mb{S}_h$. Construct another mask code $\mc{Q}^{\prime \prime} = \mr{HGP}(H_1^{\prime \prime}, H_2^{\prime \prime})$ by augmenting checks in  $H_{\mr{rep}}^{\mb{\mc{E}_v}}$ (with length $n_1$) to $H_1^{\prime}$ and then puncturing on $\mb{S}_v$.\\
\tcp{See Fig.~\ref{fig:HGP_Grid_PPMs} for illustrations.}

Prepare the logical qubits of $\mc{Q}^{\prime}$ and $\mc{Q}^{\prime \prime}$ transversely in $\overline{\ket{0}}$ and $\overline{\ket{+}}$ by preparing their physical qubits transversely in $\ket{0}$ and $\ket{+}$ states, respectively, and measuring $d$ rounds of stabilizers.\\

Apply physical transversal CNOTs from $\mc{Q}^{\prime \prime}$ to $\mc{Q}^{\prime}$ (applying pairs of CNOTs on qubits with the same coordinate, see Fig.~\ref{fig:HGP_Grid_PPMs} for examples) and measure $\mc{Q}^{\prime \prime}$ transversely in the $X$ basis. \\
\tcp{This resets a subset of the logical qubits of $\mc{Q}^{\prime}$ to $\overline{\ket{+}}$ and prepares a subset to GHZ states.}

Apply physical transversal CNOTs from $\mc{Q}$ to $\mc{Q}^{\prime}$ and measure $\mc{Q}^{\prime}$ transversely in the $Z$ basis.
\tcp{This finally measures the desired $Z$-PPMs on $\mc{Q}$ using $\mc{Q}^{\prime}$.}
\end{algorithm}

For a constant-rate HGP code encoding $k$ logical qubits, each pattern of GPPMs is implemented in parallel with a constant space overhead since the data, the ancilla, and the mask code are all of size $O(k)$. Meanwhile, it only takes $d + O(1)$ code cycles (only step 3 of Alg.~\ref{alg:HGP_grid_PPMs} takes $d$ cycles while other steps are of constant depth). Let a logical cycle refer to $d$ code cycles, then each GPPM can be implemented efficiently with $O(1)$ space overhead in $O(1)$ logical cycles.

Note that it is also possible to implement PPMs with other patterns, e.g. PPMs with a mixture of $X$ and $Z$ Paulis, if combining the homomorphic CNOT with other Clifford operations such as the fold-transversal H-SWAP gate in Table~\ref{tab:HGP_gadgets} during the Steane measurement. In this work, we focus only on these CSS-type GPPMs as elementary gadgets and execute logical computation by combining them with other transversal/fold-transversal gates, as will be discussed in the next section. 

\subsection{Logical computation using GPPMs with low space-time-overhead \label{sec:HGP_computation}}

In Theorem~\ref{theorem:parallel_Clifford_gates}, we show that the GPPMs gadget developed in this work (see Def.~\ref{def:GPPMs}), when combined with standard transversal operations as well as the fold-transversal gates developed in Ref.~\cite{quintavalle2022partitioning, breuckmann2022fold}, generate the full Clifford group on $k$ logical qubits encoded in a HGP code. All the HGP gadgets that we use in this work for logical computation are listed in Table~\ref{tab:HGP_gadgets}.

\setlength{\tabcolsep}{4pt} 
\renewcommand{\arraystretch}{1.5} 
\begin{table*}
    \centering
    \begin{tabular}{c|c|c|c}
    \hline
    \hline
         Logical gadgets &  Physical operation & Logical operation & Time \\
         \hline
        $\overline{\ket{0}}/\overline{\ket{+}}$ state preparation & $\ket{0}^{\otimes n}/\ket{+}^{\otimes n}$ $\rightarrow$ $d$ QEC cycles & $\overline{\ket{0}}^{\otimes k}/\overline{\ket{+}}^{\otimes k}$ & $O(d)$ \\
         \hline
       $Z/X$ basis measurements & $M_Z^{\otimes n}/M_X^{\otimes n}$ & $\overline{M}_Z^{\otimes k}/\overline{M}_X^{\otimes k}$ & $O(1)$\\
         \hline
         \textbf{Inter-block CNOTs} & $\bigotimes_{(i,j) \in \mb{O}\backslash \mb{O_0}} \mr{CNOT}(Q_{i,j}, Q^{\prime}_{i,j})$ & $\bigotimes_{(i,j) \in \mb{\bar{O}}\backslash \mb{\bar{O}_0}} \overline{\mr{CNOT}}(\bar{Q}_{i,j}, \bar{Q}^{\prime}_{i,j})$ & $O(1)$ \\
         \hline
          \textbf{GPPMs} & 
          \makecell{State preparation + Inter-block CNOTs \\ + transversal measurements} & $\{\bigotimes_{q \in \mb{e_r} \times \mb{e_c}} \overline{Z/X}_q \}_{\mb{e_r} \in  \mb{\mc{E}_r}, \mb{e_c} \in  \mb{\mc{E}_c}}$ & $O(d)$\\
         \hline
         H-SWAP & $\left(H^{\otimes n}\right)\left(\bigotimes_{(i,j) \in \mb{O}^+}\mr{SWAP}(Q_{i,j}, Q_{j,i})\right)$ & $(\bar{H}^{\otimes k})(\bigotimes_{(i,j) \in \mb{\bar{O}}^+}\mr{SWAP}(\bar{Q}_{i,j}, \bar{Q}_{j,i}))$ & $O(1)$ \\
         \hline
         
         CZ-S & 
         \makecell{$\left(\bigotimes_{i \in [n_1]} S(Q_{i,i}\right)\left(\bigotimes_{i \in \{n_1+1 \rightarrow 2n_1-k_1\}} S^{\dagger}(Q_{i,i}\right)$
         \\
         $\left(\bigotimes_{(i,j) \in \mb{O}^+} \mr{CZ}(Q_{i,j}, Q_{j,i})\right)$}
         & $\left(\bigotimes_{i \in [k_1]} \bar{S}(\bar{Q}_{i,i}\right)\left(\bigotimes_{(i,j) \in \mb{\bar{O}}^+}\overline{\mr{CZ}}(\bar{Q}_{i,j}, \bar{Q}_{j,i})\right)$ & $O(1)$ \\
         \hline
        \makecell{\textbf{Translation}} & Translations of blocks of qubits & $\bar{T}_{\alpha, \beta}, \quad (\alpha, \beta) \in \mb{\bar{O}}$ & $O(1)$ \\
    \hline
    \hline
    \end{tabular}
    \caption{\textbf{Summary of the HGP logical gadgets we utilize for logical computation.} All the gadgets, unless specially noted, are applied on a generic $[[n, k, d]]$ HGP code $\mc{Q}$ out of a vertical $[n_1, k_1, d_1]$-classical code $H_1$ and a horizontal $[n_2, k_2, d_2]$-classical code $H_2$ (see Fig.~\ref{fig:homomorphic_measurement}(a) for an illustration). The physical qubits are denoted $\{Q_{i,j}\}_{(i,j) \in \mb{O}}$ with coordinates $\mb{O} := [n_1]\times [n_2]\cup \{n_1 + 1, \cdots, 2n_1 - k_1\}\times \{n_2 + 1, \cdots, 2n_2 - k_2\}$, and the logical qubits are denoted $\{\bar{Q}_{i,j}\}_{(i,j) \in \mb{\bar{O}}}$ with coordinates $\mb{\bar{O}} := [k_1]\times [k_2]$ (see Fig.~\ref{fig:homomorphic_measurement}(a) and Eq.~(\ref{eq:coordinate_system}) for the coordinates of the physical and logical qubits). We denote $\mb{O}^+ := \{(i,j) \in \mb{O} \mid j > i\}$ as the upper blocks of $\mb{O}$, and similarly for $\mb{\bar{O}}^+$. The fold-transversal H-SWAP and CZ-S gates require the HGP to be symmetric, i.e. $H_1 = H_2$. The translation gadget, which translates the logical qubit along any direction under periodic boundary conditions, i.e. $\bar{T}_{\alpha, \beta}: \bar{Q}_{i,j} \rightarrow \bar{Q}_{(i+\alpha)\mod k_1, (j + \beta)\mod k_1}$, further requires that the base code is quasi-cyclic (see Appendix~\ref{sec:translation_gadget}). 
    The inter-block CNOTs are applied between $\mc{Q}$ and another HGP code $\mc{Q}^{\prime}$ that is obtained by puncturing/augmenting the base codes of $\mc{Q}$, during which a subset of physical qubits $\mb{O_0}$ and logical qubits $\mb{\bar{O}_0}$ of $\mc{Q}$ are removed (puncturing) and other checks added (augmenting, see Fig.~\ref{fig:homomorphic_measurement}(a)). Transversal physical CNOTs between pairs of qubits identified by $\mb{O}\backslash \mb{O_0}$ give transversal logical CNOTs between pairs of logical qubits identified by $\mb{\bar{O}}\backslash \mb{\bar{O}_0}$. In the case of $\mc{Q}^{\prime} \simeq \mc{Q}$ and $\mb{O_0} = \mb{\bar{O}_0} = \emptyset$, we recover the standard transversal logical CNOTs between two blocks of CSS codes. 
    All the gadgets can be implemented with a constant space overhead, i.e. using $O(k)$ physical qubits and the gadget times are listed in units of code cycles (number of QEC syndrome extraction rounds).
    The inter-block CNOTs (with $\mc{Q}^{\prime} \nsimeq \mc{Q}$) and the GPPMs (see Def.~\ref{def:GPPMs}) are introduced in this work, under the general framework of homomorphic gates and measurements~\cite{huang2022homomorphic}; The translation gadget is explicitly constructed in this work, under the general framework of autonomorphism gates~\cite{breuckmann2022fold}; The state preparation and transversal measurements are standard logical operations for any CSS codes; The H-SWAP and CZ-S gates are introduced in Ref.~\cite{quintavalle2022partitioning}, under the general framework of fold-transversal gates~\cite{breuckmann2022fold}. See also Fig.~\ref{fig:homomorphic_measurement}(c) for an illustration of some of the gadgets.}
    \label{tab:HGP_gadgets}
\end{table*}

To realize logical computation with low space-time overhead, we further require that Clifford operations on different qubits can be implemented in parallel. Here, we consider the implementation of a layer of $\Theta(k)$ Clifford gates, consisting of Hadamards, $S$ gates, and intra-block CNOTs, acting on $k$ logical qubits of a $[[n, k, d]]$ HGP code and compare its space-time cost to that using surface codes with the same distances. 
As shown in Table~\ref{tab:PPMs_cost_comparison}, although implementing such a sequence of PPMs using HGP data codes and lattice surgery with rateless ancillae~\cite{cohen2021quantum, xu2024constant, bravyi2024high} (ancillae with vanishing encoding asymptotically, e.g. surface codes) has a lower space overhead compared to using only surface codes, their space-time cost is essentially the same, since the $\Theta(k)$ gates would have to be executed sequentially in $\Theta(k)$ logical cycles to maintain the constant space overhead. As such, the former scheme is essentially trading time for space. In contrast, introducing the GPPMs gadget for the HGP code enables the parallel implementation of these gates in less than $O(k)$ logical cycles and outperforms surface codes with lattice surgery also in terms of the total space-time cost.

The ability to implement generic Clifford gates with a sublinear depth for HGP codes utilizing the GPPMs is summarized in Theorem~\ref{theorem:parallel_Clifford_gates}, for which we provide the proof in Appendix.~\ref{sec:proof_parallel_Clifford_gates}.
\begin{theorem}[Parallizable Clifford gates for HGP codes]
The gadgets in Table~\ref{tab:HGP_gadgets}, excluding the translation gadget, generate the full Clifford group on a generic HGP code. Furthermore, a layer of $\Theta(k)$ Clifford gates, consisting of Hadamards, $S$ gates, and CNOTs, on $k$ logical qubits of a quasi-cyclic HGP code can be implemented with $O(1)$ space overhead and in $O(k^{3/4})$ logical cycles using the gadgets in Table~\ref{tab:HGP_gadgets}.
\label{theorem:parallel_Clifford_gates}
\end{theorem}

Note that the costs in Theorem~\ref{theorem:parallel_Clifford_gates} is not necessarily optimal, and further reductions in the space-time cost may be possible.
As we will show in Sec.~\ref{sec:MSD_MSI}, we can also distill and consume magic states for implementing parallel non-Clifford gates utilizing these parallel Clifford operations. In combination, we can realize logical computation with an asymptotically lower space-time cost using HGP codes when compared to surface codes with lattice surgery.

\subsection{Fault-tolerant compilation with qLDPC codes}
Armed with a broad new set of native gate operations, we now consider how various key algorithmic subroutines can be efficiently implemented using qLDPC codes.
The circuit-depth upper bound $O(k)^{3/4}$ in Theorem~\ref{theorem:parallel_Clifford_gates} is for implementing a generic (or worst-case) layer of Clifford operations. For many practical computational tasks/subroutines, we can compile them into layers of structured Clifford operations such that they can be implemented in much fewer, even a constant number of, logical cycles. 
We present some task-specific and algorithm-tailored examples in the following sections.

\begin{table}[h!]
    \centering
    \begin{tabular}{c|c|c|c}
    \hline
    \hline
         &  Space & Time & Space-Time\\
         \hline
         \makecell{Surface \\ (lattice surgery)} & $\Theta(k d^2)$ & $\Theta(d)$\footnote{This analysis only applies to conventional schemes based on lattice surgery; transversal gates may modify this cost} & $\Theta(k d^3)=\Theta(k^{5/2})$ \\
         \hline
         \makecell{HGP \\ (lattice surgery, \\ rateless ancillae)} & $\Theta(k)$ & $\Theta(k)\times d$ & $\Theta(k^2 d)=\Theta(k^{5/2})$ \\
         \hline
         \makecell{HGP \\
         (GPPMs)} & $\Theta(k)$ & $ O(k^{3/4})\times d$ & $O(k^{7/4} d)=O(k^{9/4})$ \\
    \hline
    \hline
    \end{tabular}
    \caption{Comparison of the space-time cost of $\Theta(k)$ Clifford gates, consisting of Hadamards, $S$ gates, and CNOTs, on $k$ logical qubits using (1) surface codes with lattice surgery gates (2) HGP codes using lattice surgery with rateless ancillae (3) HGP codes using GPPMs, assuming long-range connectivity. We assume individual HGP blocks satisfy the scaling $d=\Theta(\sqrt{k})$, where $k$ denotes the number of logical qubits per block, and compare to surface codes with the same distances. In practice, the code distance will be chosen based on the target error suppression, and the computation can be divided into independent blocks that still satisfy the scaling $d=\Theta(\sqrt{k})$. 
    }
    \label{tab:PPMs_cost_comparison}
\end{table}

\subsubsection{GHZ state generation}
\begin{figure}
    \centering
    \includegraphics[width=0.5\textwidth]{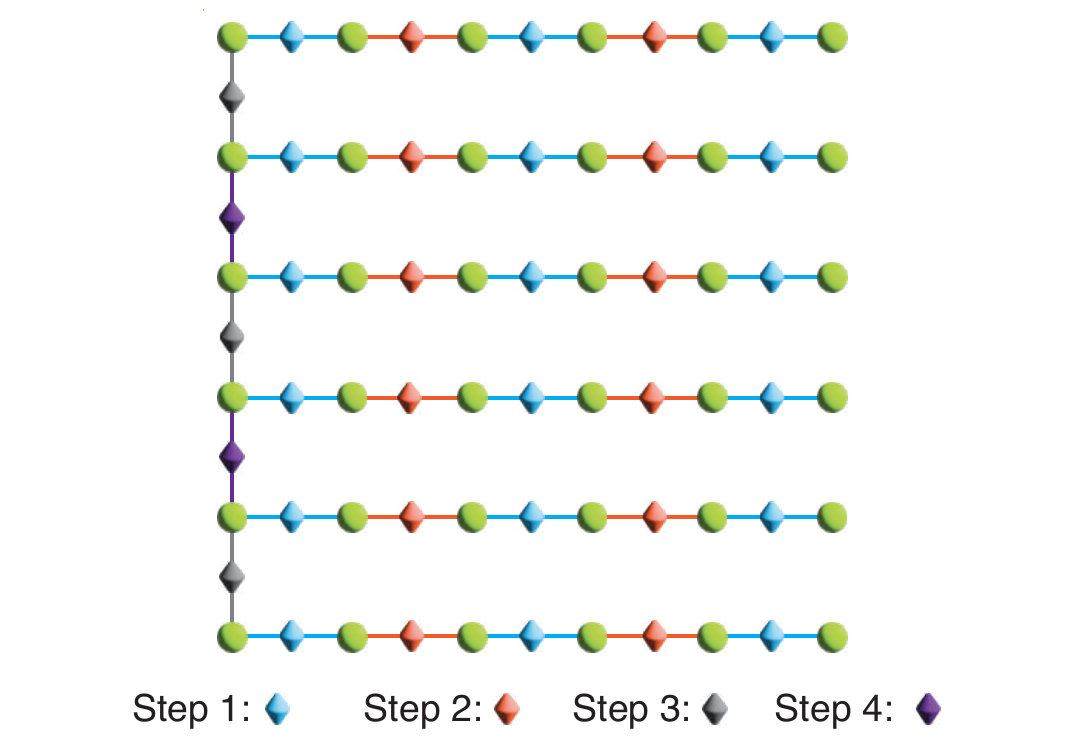}
    \caption{\textbf{A protocol for generating a GHZ state on all logical qubits of an HGP code using four GPPMs gadgets.} The Pauli products to be measured by each GPPMs gadget are indicated by the 3D diamonds of a particular color.}
    \label{fig:GHZ}
\end{figure}
As shown in Fig.~\ref{fig:GHZ}, a GHZ state across all logical qubits on a HGP code block with a $M\times N$ grid of logical qubits can be generated with $O(1)$ space overhead and in $O(1)$ logical cycles by transversally preparing all the logical qubits in the $\ket{+}$ state and performing the following sequence of GPPMs:
\begin{equation}
\begin{aligned}
    & \mr{GPPMs}(Z, \{\{r\}\}_{r\in [M]}, \{\{2j-1, 2j\}\}_{j \in \lfloor N/2\rfloor}) \\
    \rightarrow \ & \mr{GPPMs}(Z, \{\{r\}\}_{r\in [M]}, \{\{2j, 2j + 1\}\}_{j \in \lfloor (N-1)/2 \rfloor}) \\
    \rightarrow \ & \mr{GPPMs}(Z, \{\{2i-1, 2i\}\}_{i \in \lfloor M/2 \rfloor}, \{\{1\}\}) \\
    \rightarrow \ & \mr{GPPMs}(Z, \{\{2i, 2i + 1\}\}_{i \in \lfloor (M-1)/2 \rfloor}, \{\{1\}\}),
\end{aligned}
\end{equation}
where the four GPPMs measure the PPMs in Fig.~\ref{fig:GHZ} represented by the blue, red, grey, and purple diamonds, respectively.

\subsubsection{Magic state distillation and consumption\label{sec:MSD_MSI}}
\outline{Magic state distillation}
As another example, we consider distilling magic states in parallel using only qLDPC codes. At a high level, we can distill block(s) of $k$ magic states in parallel encoded in $[[n, k, d]]$ code patch(es) using $M$ identical patches. As each qLDPC patch has a constant encoding rate, such a parallel distillation scheme is still space efficient as long as $M$ is not too large. Moreover, doing so mostly requires inter-block Clifford operations, which are generally easier than intra-block operations. Most of the operations are transversal, i.e. the same operation is applied to the same set of qubits across different qLDPC patches. However, the conditional Clifford fixups (or equivalently, the reactive measurements in Ref.~\cite{litinski2019game, litinski2022active}) can break the transversal structure of the circuit. As such, we will need our GPPMs gadget to perform selected operations on certain code patches. Therefore, the entire distillation scheme is also time-efficient as long as these selected operations can be implemented in parallel using the GPPMs.

\begin{figure*}
    \centering
    \includegraphics[width=1\textwidth]{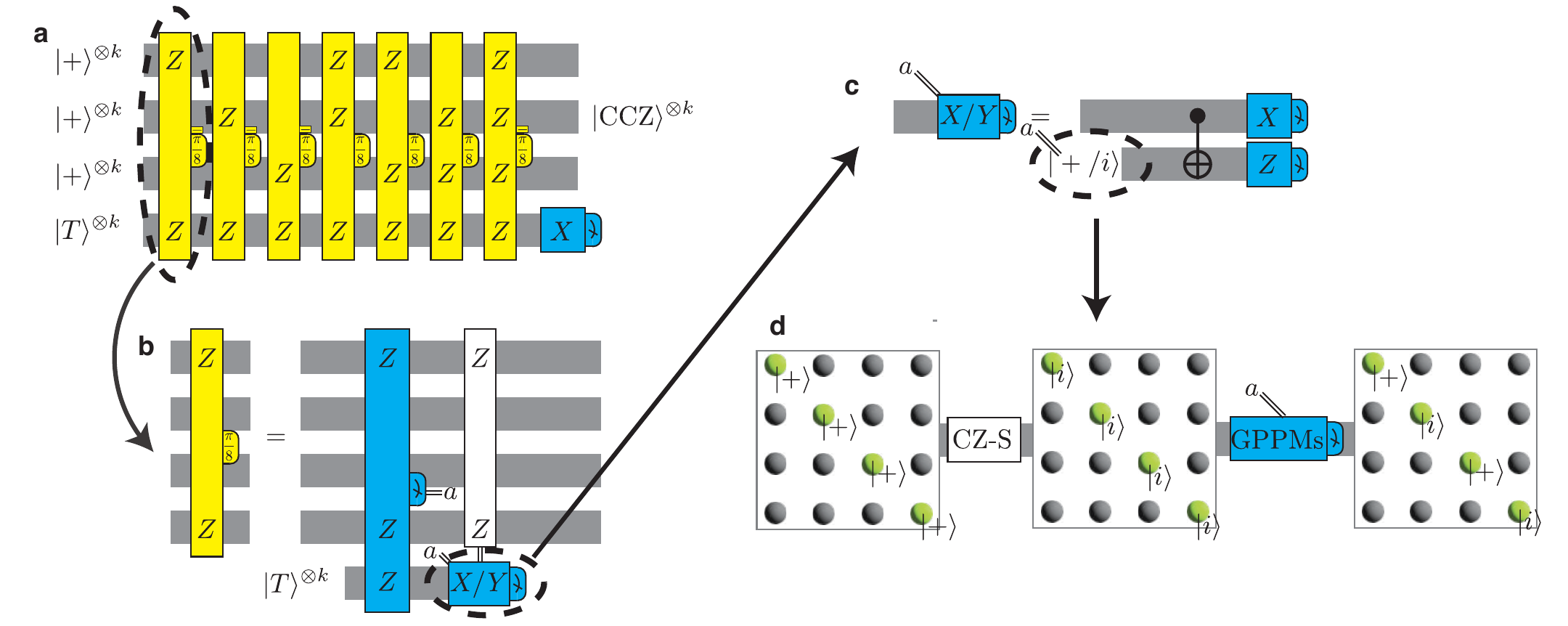}
    \caption{\textbf{One round of parallel magic state distillation on diagonal logical qubits of HGP codes with $O(1)$ space overhead and in $O(1)$ logical cycles.} Each thick grey line indicates a block HGP code encoding $k$ logical qubits. An operation involving different thick lines refers to transversal operations on all the logical qubits of the corresponding HGP codes, unless specially noted. Each yellow operation denotes a Pauli product rotation $\exp(- i P \phi)$ with a multi-qubit Pauli operator $P$ and an rotation angle $\phi$; Each blue operation denotes a Pauli product measurement. The right-going and left-coming double lines denote classical outputs and inputs, respectively. (a) The $8$-to-CCZ distillation \emph{logical} circuit that converts $8$ blocks of noisy $\ket{T}$ states into one block of less noisy $\ket{\mr{CCZ}}$ states~\cite{litinski2022active}. (b) Each of the Pauli product rotations in (a) can be implemented by supplementing an extra block of $\ket{T}$ states, performing joint transversal PPMs, and finally measuring the $\ket{T}$ block reactively in $X$ or $Y$ basis, depending on the previous PPMs. (c) The reactive measurements of the $\ket{T}$ block in (b) can be implemented by introducing another ancilla block, whose logical qubits are initialized reactively in $\ket{+}$ or $\ket{i}$ states, and then performing transversal Bell measurements. (d) The reactive state preparation of the ancilla block in (c) can be implemented by (1) initializing the logical qubits transversely in $\ket{+}$ states, (2) applying the fold-transversal CZ-S gate to convert the diagonal qubits to $\ket{i}$ states, (3) performing a pattern of GPPMs to reset some of the diagonal qubits to $\ket{+}$ states. }
    \label{fig:MSD}
\end{figure*}

\begin{figure*}
    \centering
    \includegraphics[width=1\textwidth]{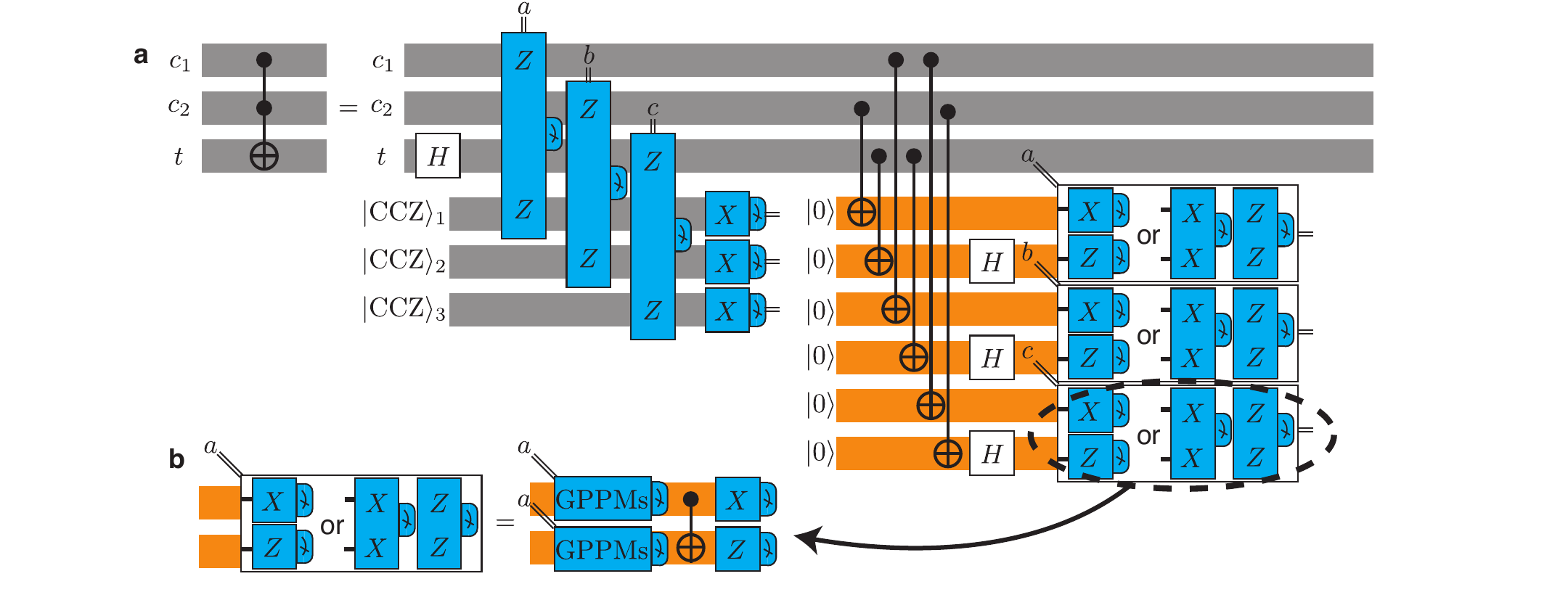}
    \caption{\textbf{Parallel inter-block Toffoli gates on HGP codes by consuming $\ket{\mr{CCZ}}$ states in parallel on diagonal logical qubits with $O(1)$ space overhead and in $O(1)$ logical cycles.} (a) A $\ket{\mr{CCZ}}$ consumption circuit by performing joint PPMs between the data blocks and the $\ket{\mr{CCZ}}$ blocks and transversally measuring the $\ket{\mr{CCZ}}$ blocks, followed by reactive CZs on the data blocks~\cite{litinski2022active}. The reactive CZs are implemented by coupling the data blocks via transversal CNOTs to extra ancilla blocks (indicated in orange), which are then reactively measured in a single-qubit or Bell basis. (b) The reactive measurements on each pair of the ancilla blocks in (a) can be implemented by two GPPMs on each block followed by transversal Bell measurements. Note that we have neglected the Pauli corrections in the circuit for simplicity, which can be implemented simply by updating the Pauli frame. 
    }
    \label{fig:MSI}
\end{figure*}

As shown in Fig.\ref{fig:MSD}(a), we consider the ``$8$-to-CCZ" distillation circuit that converts $8$ noisy $\ket{T}$ states into one less noisy $\ket{\mr{CCZ}}$ state~\cite{jones2013low,gidney2019efficient,litinski2022active}. We set the input to be four HGP code patches, three of which encode $\ket{+}$ states while the remaining encodes $\ket{T}$ states. Then, the distillation circuit is implemented by a sequence of transversal $Z$-type Pauli product rotations across the four patches and transversal $X$ measurements on the $\ket{T}$ patch. Each of the transversal Pauli product rotations can be implemented by introducing another code patch that encodes $\ket{T}$ transversally, performing a joint transversal PPM, and finally measuring the introduced $\ket{T}$ patch reactively depending on the PPM outcomes (see Fig.~\ref{fig:MSD}(b)). Since each logical qubit on the $\ket{T}$ patch is measured in $X$ or $Y$ basis reactively, such measurements are not transversal, and thus are the most expensive components of the distillation task.

As shown in Fig.~\ref{fig:MSD}(c), we can further reduce the task of reactive measurements in the $X/Y$ basis in Fig.~\ref{fig:MSD}(b) to a reactive state preparation in the $X/Y$ basis. By introducing another ancilla patch with qubits reactively initialized in the $X/Y$ basis, and performing a transversal Bell measurement between the two patches, we can obtain the $X/Y$ reactive measurement outcomes of the data patch from the Bell measurement outcomes. 

We can implement the reactive state preparation subroutine efficiently using the GPPMs gadget. As shown in Fig.~\ref{fig:MSD}(d), we first consider using only the diagonal logical qubits in a symmetric HGP code. We can prepare the diagonal qubits in an arbitrary pattern of $\ket{+}/\ket{i}$ states in three steps:
\begin{enumerate}
    \item Prepare all the diagonal qubits in $\ket{+}$ states.
    \item Apply the fold-transversal CZ-S gate (see Table~\ref{tab:HGP_gadgets}) to convert all the diagonal $\ket{+}$ states to $\ket{i}$ states. 
    \item Apply a $\mr{GPPMs}($X$, \mb{\mc{E}}_0, \mb{\mc{E}}_0)$, where $\mb{\mc{E}}_0 = \{\{i\}\}_{i \in \mb{S}}$ and $\mb{S}$ contains the indices of the diagonal qubits to be prepared in the $X$ basis.
\end{enumerate}
The above procedure implements the reactive state preparation subroutine on the diagonal $\sqrt{k}$ qubits with $O(k)$ physical qubits and in $O(1)$ logical cycles. Therefore, each round of the distillation circuit in Fig.~\ref{fig:MSD} can also be run on the diagonal $\sqrt{k}$ logical qubits with $O(k)$ physical qubits and $O(1)$ code cycles. 

Note that the main obstacle of performing the above distillation task on all logical qubits instead of the diagonal ones boils down to preparing an entire block of logical qubits reactively in the $X/Y$ basis (see Fig.~\ref{fig:MSD}(d)). 
In Appendix~\ref{sec:extra_gadgets}, we show how to realize it in $O(\sqrt{k}\log k)$ logical cycles by imposing an additional translational symmetry of the HGP codes. At a high level, by using a family of quasi-cyclic classical codes, we can obtain HGP codes with a translational automorphism, i.e., the logical qubit grid can be translationally shifted with periodic boundary conditions by simply permuting the physical qubits. Then, we can repeatedly generate $\sqrt{k}$ $\ket{i}$ states using the diagonal logical qubits (utilizing the CZ-S gate) and distribute them to other qubits via the translational automorphism. Finally, we can selectively reset a subset of logical qubits to $\ket{+}$ states efficiently using the $X$-type GPPMs. Therefore, by using HGP codes with additional translational symmetry, we can perform magic state distillation in Fig.~\ref{fig:MSD} on $k$ logical qubits with $O(1)$ space overhead in $O(\sqrt{k}\log k)$ logical cycles per distillation round. 

We note that the same protocol can be applied to distilling $\ket{T}$ states in parallel using, e.g. a $15$-to-$1$ distillation protocol~\cite{litinski2022active}, with the same scaling of the space and the time overhead.  

\outline{Parallel non-Clifford gates by consuming magic states}
With $\Theta(k)$ magic states distilled on qLDPC patches in parallel, we can also consume them and perform parallel non-Clifford gates. 
For instance, as shown in Fig.~\ref{fig:MSI}(a), we can implement parallel Toffoli gates on three $[[n, k, d]]$ HGP patches by consuming three patches of distilled $\ket{\mr{CCZ}}$ states. The $\ket{\mr{CCZ}}$ states are consumed via the transversal $Z$-type PPMs together with the data patches, and the transversal $X$ measurements that follow. Then, reactive CZ gates on the data patches are applied depending on the PPMs outcomes. We follow Ref.~\cite{litinski2022active} and convert these reactive CZs into reactive PPMs by introducing six ancilla patches (indicated by the orange blocks in Fig.~\ref{fig:MSI}), interacting them with the data patches via the transversal CNOTs, and reactively measuring them in either $X/Z$ basis or a Bell basis in a pairwise fashion. These reactive measurements are, again, not transversal operations since qubits on the same patch may need to be measured in different bases. Thus, we need to implement them using the selective GPPMs gadget.

As shown in Fig.~\ref{fig:MSI}(b), we can realize the reactive measurement on a pair of the ancilla patches in Fig.~\ref{fig:MSI}(a) by first performing non-destructive $X/Z$-type GPPMs selectively on the subset of qubits to be measured in the single-qubit basis, and then performing the transversal destructive Bell measurements. Both the GPPMs in Fig.~\ref{fig:MSI}(b) and the transversal Hadamard gates in Fig.~\ref{fig:MSI}(a) can be implemented in parallel in $O(1)$ logical cycles if the task is restricted again only on the diagonal logical qubits of symmetrical HGP codes (the fold-transversal H-SWAP gate implements transversal Hadamards on the diagonal qubits). Therefore, we obtain a protocol that implements $\sqrt{k}$ Toffoli gates in parallel on diagonal logical qubits by consuming $\sqrt{k}$ $\ket{\mr{CCZ}}$ states with $O(k)$ physical qubits in $O(1)$ logical cycles. 

Similar to the magic state distillation task, we can also extend the magic state consumption task to all $k$ logical qubits (beyond the diagonal ones) with $O(k)$ physical qubits in $O(\sqrt{k}\log k)$ logical cycles by using HGP codes with additional translation symmetries. The trick is to implement the transversal Hadamards in Fig.~\ref{fig:MSI}(a) on a HGP patch by first implementing the fold-transversal H-SWAP gate, and then using the translational automorphism, combined with GPPMs, to cancel the extra swaps on the non-diagonal qubits. See Appendix.~\ref{sec:transversal_H} for more details.

In conclusion, by using HGP codes with additional translational symmetries, we can not only distill $k$ magic states but also consume them and implement non-Clifford gates in parallel with $O(1)$ space overhead in $O(\sqrt{k}\log k)$ logical cycles.

\subsubsection{Quantum adder}
With the ability to distill and consume magic states in parallel, we can explore computation subroutines that require many parallel non-Clifford gates. As an example, we consider the quantum adder~\cite{gidney2018halving}, which is an important subroutine of many useful quantum algorithms such as the factoring algorithm~\cite{shor1994algorithms}. 

The adder inputs two $k$-qubit registers $\ket{a} = \otimes_{i=1}^k \ket{a_i}$ and $\ket{b} = \otimes_{i=1}^k \ket{b_i}$ representing two integers $a$ and $b$, respectively, and outputs two registers $\ket{a}$ and $\ket{a + b}$. Note that the adder is a unitary quantum circuit, so the input can also be a superposition of integers and the corresponding output will be a superposition of the added integers.

\begin{figure*}
    \centering
    \includegraphics[width=1\textwidth]{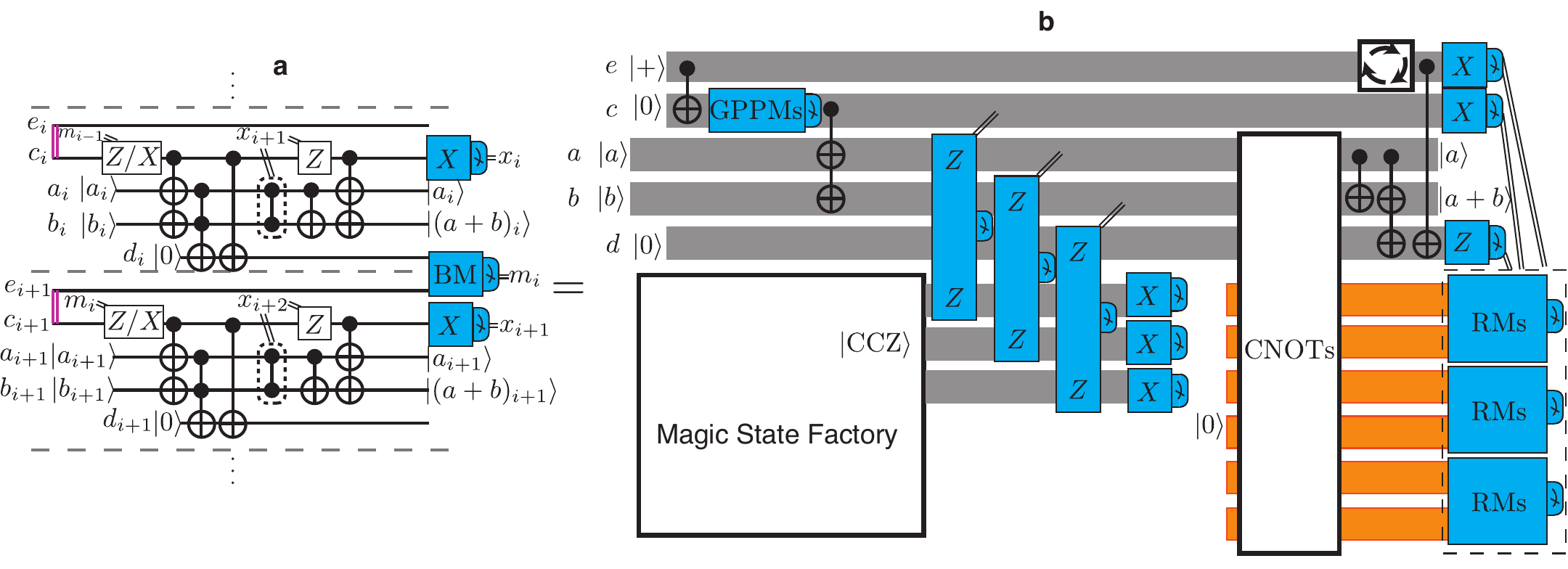}
    \caption{\textbf{An efficient parallel implementation of a quantum adder using quasi-cyclic HGP codes.} (a) The Gidney adder circuit~\cite{gidney2018halving, litinski2022active} with temporary-AND Toffolis~\footnote{Toffolis with a $\ket{0}$-initialized target qubit, which can have a reduced space-time cost compared to regular Toffolis~\cite{gidney2018halving}} has repeated sectors, each sandwiched by two dashed lines. Different sectors are connected by the bridge qubits that is teleported across different blocks back in time~\cite{litinski2022active}. Specifically, each output carry qubit $d_i$ is teleported back in time as the input carry qubit $c_{i + 1}$ (of the next bit), utilizing the Bell measurements (BM) and the Bell state preparation (indicated by the leftmost verticle double lines). The right-going and left-coming double lines next to an operation denote classical outputs and inputs, respectively. (b) We can implement all the sectors of (a) approximately in parallel by encoding different qubits of the same type, e.g. $\{a_i\}$, into an HGP code and performing mostly inter-block operations. In particular, the temporary-AND Toffolis in different sectors can be implemented in parallel using transversal $\ket{\mr{CCZ}}$ distillation and consumption (see Fig.~\ref{fig:MSD} and Fig.~\ref{fig:MSI}, as well as Sec.~\ref{sec:MSD_MSI}). Similar to the magic state consumption circuit in Fig.~\ref{fig:MSI}, the ancilla blocks (the orange lines) first interact with the $a$, $b$, and $c$ blocks via the transversal inter-block CNOTs (with details omitted here) and are finally measured reactively in a pairwise fashion. Different from Fig.~\ref{fig:MSI}, here, the ancilla reactive measurements need to be executed sequentially and qubits in the same block need to be measured in order depending on the previous outcomes (the details of the Pauli feedforwards that can be tracked in software are omitted here). The gate with three arrows at the top right corner denotes a cyclic permutation of the logical qubits in an HGP code. 
    }
    \label{fig:adder}
\end{figure*}

We follow the circuit in Ref.~\cite{litinski2022active}, which is a variant of the Gidney ripple-carry adder~\cite{gidney2018halving, cuccaro2004new}, that performs a $k$-qubit addition using $k-1$ temporary-AND Toffolis. We refer readers to Fig. 16 of Ref.~\cite{litinski2022active} (and related texts) for more details of the circuit.
As shown in Fig.~\ref{fig:adder}(a), the circuit consists of $k-2$ repeated segments (the first and the last segments are different), each computing the addition of the $i$-th bit of $a$ and $b$. The $i$-th and the $(i + 1)$-th segments are connected by a shared carry bit that is simultaneously the output of the $i$-th segment and the input of the $(i+1)$-th segment. 
Because of these shared carry bits, the computation is generically sequential. Nevertheless, as shown in Fig.~\ref{fig:adder}, one can parallelize the computation on different segments by introducing bridge qubits for the carry bits~\cite{litinski2022active}. 
More specifically, we introduce an output carry qubit $d_i$ for the $i$-th segment and a pair of input carry qubits $e_{i + 1}$ and $c_{i + 1}$ in a Bell state for the $(i + 1)$-th segment. 
As shown in Fig.~\ref{fig:adder}, we can now execute all the segments in parallel and finally, perform a Bell measurement between $d_i$ and $e_{i + 1}$ to effectively teleport $d_i$ back in time as the input carry bit $c_{i + 1}$ of the $(i + 1)$-th segment. Such a circuit can also be easily verified using ZX calculus~\cite{litinski2022active}. 

The circuit in Fig.~\ref{fig:adder}(a) is now almost parallel, except that there is a reactive CZ on each pair of the $a$ and $b$ bits depending on the measurement of the next $c$ bit. These reactive CZs come from the uncompute step of the temporary-AND Toffolis, and can be combined with the reactive CZs that come after consuming the $\ket{\mr{CCZ}}$ states at the computing step of the temporary-AND Toffolis. See Ref.~\cite{gidney2018halving} and Ref.~\cite{litinski2022active} for more details. Thus, to perform these reactive CZs, we can introduce some ancilla qubits, let them interact with the $a$, $b$, and $d$ qubits via transversal CNOTs, and finally remove them via reactive measurements (same as Fig.~\ref{fig:MSI}(a)). Now, the entire circuit can be implemented exactly in parallel, except for the final reactive measurement of the ancilla qubits. 

Given this repeated and parallel structure, we can implement the adder using qLDPC patches. As shown in Fig.~\ref{fig:adder}(b), we can encode the five types of qubits into five identical $[[n,k,d]]$ HGP patches, which we call $a$-, $b$-, $c$-, $d$-, and $e$-patch, respectively. 
Then the circuit in Fig.~\ref{fig:adder}(a) can be implemented using HGP patches in Fig.~\ref{fig:adder}(b) that involves mainly \emph{inter-block and transversal} operations, except for the final reactive measurements of the six ancilla patches (the orange blocks). These reactive pairwise measurements are the same as those in Fig.~\ref{fig:MSI}(a), except that they now also depend on the $X$ measurements of the $c$ patch as well as the Bell measurements of the $d$ and the $e$ patches. Note that these reactive measurements need to be executed sequentially, i.e. the basis of the next RMs on the same block will depend on the outcome of the previous RMs, fundamentally because the adder task is generically sequential. 
Nevertheless, these are simply ancilla patches to be consumed and the data patches ($a$- and $b$- patches) are free to perform other computation tasks in parallel. 
To this end, the entire qLDPC-based adder scheme in Fig.~\ref{fig:adder}(b) can be implemented with $c_1 \sqrt{k} \log k + c_2 k$ logical cycles, where $c_1$ and $c_2$ are constants associated mainly with the parallel magic state distillation and the sequential ancilla RMs, respectively. We expect that $c_2 \ll c_1$ and the adder can be implemented approximately in parallel for practical integer sizes.

Finally, as shown in Fig.~\ref{fig:adder}(b), in addition to the RMs, there are two extra operations that are not transversal. The first operation is the $Z$-measurement on the first qubit of the $c$-block after the first transversal CNOT. This is due to the fact that the addition of the first bits $a_0$ and $b_0$ does not have any input carry bit, which requires us to reset the carry qubit $c_0$ to $\ket{0}$ via a $Z$ measurement. This can be easily implemented using a $Z$-type GPPMs selecting only the first logical qubit of the $c$-block. Finally, as shown in Fig.~\ref{fig:adder}(a), the final Bell measurements between the $d$ and $e$ blocks are performed between the $i$-th $d$ qubit and the $(i + 1)$-th $e$ qubit, which also breaks the transversal structure. 
Fortunately, we can realize such a mismatched Bell measurement by simply cyclically shifting the $e$ block and then performing a transversal Bell measurement. 
As presented in Appendix~\ref{sec:cyclic_shift}, we can implement such a cyclic shift of the $e$ block in $O(1)$ logical cycles by combining the translational automorphism of a translational-symmetrical HGP code and the GPPMs gadget. 

To conclude, we can implement a $k$-bit quantum adder in parallel using qLDPC codes in $O(\sqrt{k} \log k)$ logical cycles by leveraging the parallel magic state distillation and consumption protocols that we developed in Sec.~\ref{sec:MSD_MSI}. 

We note that there are many other algorithms/subroutines that demand the parallel implementation of non-Clifford gates, such as the quantum random-access memory~\cite{giovannetti2008quantum} and various quantum state preparation gadgets~\cite{zhang2022quantum}. We thus expect similar techniques that leverage the parallel qLDPC-encoded non-Clifford gates can be exploited for a broad class of computational tasks with low space-time overhead. 
 
\section{Single-shot logical gates for 3D/4D homological product codes \label{sec:high_dimensional_codes}}
As described in Sec.~\ref{sec:sketch_homomorphic_gadget}, the homomorphic measurement gadget we developed for the homological product codes work generally for any dimension $D \geq 2$ by simply puncturing/augmenting their base classical codes. In particular, we can generalize the GPPMs gadget for the $2D$ HGP code measuring a grid pattern of PPMs to a $D$-dimensional gadget for a $D$-dimensional code measuring a $D$-dimensional-hypercube pattern of PPMs selectively and in parallel.

Although higher-dimensional ($3D$ or $4D$) homological product codes tend to have larger block sizes, they have redundant check matrices inherently due to their higher-dimensional product construction, i.e. the syndromes satisfy extra linear constraints referred to as metachecks~\cite{campbell2019theory}. These metachecks could help reduce the number of repeated syndrome measurements in the presence of measurement errors. For example, as we show in Appendix~\ref{sec:single_shot_state_prep}, 4D homological product codes, whose checks satisfy the soundness property~\cite{campbell2019theory}, support single-shot preparation of computational-basis states. This enables homomorphic measurements with \emph{constant depth}.


Here, in Def.~\ref{def:CPPMs}, we explicitly present a parallel PPMs gadget for 3D homological product codes by generalizing the GPPMs gadget for the HGP codes straightforwardly, and the construction for 4D codes follows. 
As we show in Appendix~\ref{appendix:3D_PPMs}, the logical qubits of a 3D homological product codes can be arranged on a cube, and we can construct a ``Cube'' PPMs (CPPMs) gadget that measures a pattern of PPMs in parallel on any subcube of the logical qubits in Def.~\ref{def:CPPMs}. See Appendix~\ref{appendix:3D_PPMs} for details.

\begin{definition}[Cube PPMs]
Given a canonical representation of logical qubits in a 3D homological product code on a 3D cube $[k_1]\times[k_2]\times [k_3]$, and three sets of hyperedges $\mb{\mc{E}_x}$, $\mb{\mc{E}_y}$ and $\mb{\mc{E}_z}$, where each hyperedge is a collection of indices in the $X$, $Y$, and $Z$ direction, respectively, we define a pattern of Cube PPMs of a type $P \in \{X, Z\}$ as:
\begin{equation}
    \mr{CPPMs}(\mr{P}, \mb{\mc{E}_x}, \mb{\mc{E}_y}, \mb{\mc{E}_z}) := \{\bigotimes_{q \in \mb{e_x} \times \mb{e_y} \times \mb{e_z}} \bar{P}_q \}_{\mb{e_x} \in  \mb{\mc{E}_x}, \mb{e_y} \in  \mb{\mc{E}_y}, \mb{e_z} \in  \mb{\mc{E}_z}},
\end{equation}
where $\bar{P}_q$ denotes the logical Pauli operator $\bar{X}$ ($\bar{Z}$) of the logical qubits $\bar{Q}_q$ with coordinates $q$ for $P = X$ ($Z$).
\label{def:CPPMs}
\end{definition}

\section{Physical implementation \label{sec:physical_implementation}}

Our proposal is particularly natural to implement in dynamically-reconfigurable qubit architectures~\cite{bluvstein2022quantum,bluvstein2023logical,pino2021demonstration}, such as neutral atom arrays, where the qubits can be dynamically rearranged and parallel two-qubit gates can be applied with global controls.
The experiments of Ref.~\cite{bluvstein2023logical} showed that parallel control over logical qubits can dramatically simplify the implementation of error correction, due to the fact that all the physical qubits of the block need to realize the same operation to realize a targeted logical operation.
Since the new homomorphic CNOTs and measurements are built upon transversal physical CNOTs between two alike code patches, we expect that they can also be efficiently implemented by overlapping the two patches and then applying global pairwise CNOTs.
Combined with recent proposals for the efficient implementation of various qLDPC codes in neutral atom arrays~\cite{xu2024constant,viszlai2023matching}, utilizing again the natural control parallelism afforded by optical tools, we expect that all required operations can be efficient implemented.

In this work, we further observe that due to the large-scale parallelism on the \emph{algorithmic level}, many key subroutines such as the quantum adder again involve the same, repeated operations at the \emph{logical qubit} level. This parallelism is then particularly well-suited to LDPC codes with transversal operations, and indicates that sophisticated operations on dense block encodings can be controlled efficiently at both the logical and physical level, offering unique opportunities for dramatically reducing the costs of large-scale computation.

\section{Discussion and outlook}
In this paper, we presented a suite of new methods for performing logical gates with homological product codes, enabling fast and parallelizable logical operations.
Crucially, our methods eliminate the need for surface- or repetition-code-like structures in the construction of the logical gate, removing a key bottleneck on parallelism of previous approaches~\cite{cohen2022low,xu2024constant,bravyi2024high}.
Moreover, the construction is remarkably simple, making use of well-known modifications to the base classical codes of the homological product codes.
This enables a high degree of parallelism and hardware-efficient implementation, which are crucial in practice.

Building on top of this set of basic operations, we focus on the parallel, efficient implementation of key algorithmic subroutines, initiating the study of fault-tolerant compilation with native qLDPC code operations.
We report efficient implementations of large GHZ state preparation, magic state distillation and consumption, culminating in the efficient implementation of a quantum adder.

Our results open up many exciting areas of future research.
While we have proposed a variety of useful logical gadgets, there is ample room to further expand the range of accessible operations, and achieve additional reductions in their space-time cost.
In particular, the idea of masking an ancilla qLDPC block and preparing inhomogeneous logical states for executing selective and parallel logical operations (see Fig.~\ref{fig:HGP_Grid_PPMs}) could be exploited to design new logical gadgets with even lower time overhead.
Additionally, in order to make full use of the protocols for efficiently distilling and consuming magic states in Sec.~\ref{sec:MSD_MSI}, we also need to prepare/inject initial noisy magic states with infidelities below some distillation threshold~\cite{bravyi2005universal} as the input to the magic state factory. It will therefore be interesting to generalize existing protocols for high-fidelity magic state injection~\cite{li2015magic,lao2022magic, gidney2023cleaner} to various qLDPC codes. 
In Sec.~\ref{sec:high_dimensional_codes}, we show how to apply our constructions to higher-dimensional homological product codes and measure PPMs in parallel and in constant depth. 
Some of these higher-dimensional codes, e.g. the 3D surface code~\cite{vasmer2019three,kubica2015unfolding,bombin2013gauge}, also support transversal non-Clifford gates. As such, it would be interesting to explore a richer set of constant-depth logical operations with these codes, ideally with high encoding rates~\cite{zhu2023non}, in the future.
Similarly, it will also be interesting to generalize our constructions to other qLDPC codes with product constructions, such as the lifted product code~\cite{panteleev2019degenerate, panteleev2022quantum, bravyi2024high}, fibre bundle code~\cite{Hastings2020}, and balanced product code~\cite{breuckmann2020balanced}. 

Another important direction is efficient decoding and numerical simulations of our protocol.
We expect general purpose decoding algorithms such as BP-OSD~\cite{panteleev2019degenerate,roffe2020decoding,Roffe_LDPC_Python_tools_2022} and hypergraph union-find~\cite{delfosse2022toward,cain2024correlated} to achieve good performance, and one may be able to further exploit the product structure and expansion properties of these codes to improve performance~\cite{stambler2023addressing,krishna2023viderman,grospellier2021combining}.
In addition, the use of correlated decoding techniques~\cite{cain2024correlated} and principles of algorithmic fault tolerance~\cite{zhou2024algorithmic} may further reduce the time overhead of our construction, allowing only $O(1)$ rounds of syndrome extraction in certain cases. 

Finally, with these new techniques and vastly-expanded set of operations, a frontier of future research will be to perform end-to-end compilations of large scale algorithms, and demonstrate a concrete space-time saving over current schemes in the full algorithmic setting.
In addition to the quantum adder described here, quantum simulation may be another key area of interest, where in many contexts all logical qubits realize the same structured evolution~\cite{campbell2021early}, and thus are well-suited to parallelism in LDPC blocks. This could be explored both for efficient trotterized Hamiltonian evolution~\cite{childs2021theory}, as well as parallelized operations using qubitization for studying electronic spectra in materials~\cite{babbush2018encoding,low2016hamiltonian}. 

\begin{acknowledgments}
We acknowledge helpful discussions with John Preskill, Christopher Pattison, Shouzhen Gu, Han Zheng, Nithin Raveendran, Asit Pradhan, and Daniel Litinski, Nishad Maskara, Madelyn Cain, Christian Kokail. We especially thank Shilin Huang for initial discussions and insightful comments. We acknowledge support from the ARO (W911NF-23-1-0077), ARO MURI (W911NF-21-1-0325, W911NF-20-1-0082), AFOSR MURI (FA9550-19-1-0399, FA9550-21-1-0209, FA9550-23-1-0338), DARPA (HR0011-24-9-0359, HR0011-24-9-0361, ONISQ W911NF2010021, IMPAQT HR0011-23-3-0012), IARPA Entangled Logical Qubits program (ELQ, W911NF-23-2-0219) NSF (OMA-1936118, ERC-1941583, OMA-2137642, OSI-2326767, CCF-2312755, PHY-2012023, CCF-2313084), DOE/LBNL (DE-AC02-05CH11231), NTT Research, Samsung GRO, the Center for Ultracold Atoms (a NSF Physics Frontiers Center, PHY-1734011), and the Packard Foundation (2020-71479).
D.B. acknowledges support from the NSF Graduate Research Fellowship Program (grant DGE1745303) and The Fannie and John Hertz Foundation.
J.P.B.A. acknowledges support from the Generation Q G2 fellowship and the Ramsay Centre for Western Civilisation.
After the completion of this project, we became aware of related work studying the implementation of logical operations on qLDPC codes using improved lattice surgery~\cite{cowtan2024ssip,cross2024ibmgates}.
\end{acknowledgments}

\newpage
\begin{appendix}
\newpage
\newpage
\appendix 

\section{Additional gadgets for logical operations on HGP codes
\label{sec:extra_gadgets}}
In this section, we present additional logical gadgets used for the HGP codes in the main text. We first present the logical translation gadget listed in Table~\ref{tab:HGP_gadgets} for HGP codes with quasi-cyclic base classical codes. Then, we show how to implement parallel logical Hadamard gates (without extra swaps) and parallel $\ket{i}$-state preparation utilizing the translation gadget. 
Such gadgets enable us to implement parallel magic state distillation and injection (see Sec.~\ref{sec:MSD_MSI}) on all logical qubits of HGP codes, extending from the diagonal qubits. 
We also present a gadget for performing selective inter-block teleportation as well as a gadget for cyclically shifting all logical qubits. 

\subsection{Logical translation gadget \label{sec:translation_gadget}}
The logical translation gadget is based on an automorphism of a quantum code $\gamma = \{\gamma_2, \gamma_1, \gamma_0\}: \mc{Q} \rightarrow \mc{Q}$ such that the following diagram commutes:
\begin{equation}
    \label{eq:quantum_code_auto}
        \begin{tikzcd}
	{Q_2} & {Q_1} & {Q_0} \\
	{Q_2} & {Q_1} & {Q_0}
	\arrow["{H_Z^T}", from=1-1, to=1-2]
	\arrow["{H_X}", from=1-2, to=1-3]
        \arrow["{H_Z^T}", from=2-1, to=2-2]
	\arrow["{H_X}", from=2-2, to=2-3]
        \arrow["{\gamma_2}", from=2-1, to=1-1]
	\arrow["{\gamma_1}", from=2-2, to=1-2]
        \arrow["{\gamma_0}", from=2-3, to=1-3]
\end{tikzcd}
\end{equation}
Clearly, this is a special case of the homomorphism between two quantum codes $\gamma: \mc{Q}^{\prime} \rightarrow \mc{Q}$ with $\mc{Q}^{\prime} = \mc{Q}$. We again work in the standard basis where each basis vector of $Q_2$, $Q_1$, and $Q_0$ represents a $Z$ check, a qubit, and a $X$ check, respectively. If $\gamma_1$ is a permutation matrix, the physical permutation of the qubits according to $\gamma_1$ preserves the stabilizer group and thus implements a logical gate. Such a gate is called an automorphism gate~\cite{breuckmann2022fold, bravyi2024high}. In the following, we will construct a special type of automorphism gate for HGP codes that implements a translation of all the logical qubits in the canonical basis. 

Based on Proposition~\ref{prop:homological_prod_code_homo_by_classical}, we can also construct an automorphism for a homological product code by taking the tensor product of automorphisms of its base classical codes. In the particular case of a HGP code $\mc{Q}$, which is constructed from two base length-$1$ chain complexes $\mc{C}^1$ and $\mc{C}^2$ (see Eq.~\eqref{eq:HGP_complex}), we can construct an automorphism $\gamma = \{\gamma_2, \gamma_1, \gamma_0\}: \mc{Q} \rightarrow \mc{Q}$, where
\begin{equation}
\begin{aligned}
    \gamma_2 & = \gamma^1_1 \otimes \gamma^2_1, \\
    \gamma_1& = (\gamma^1_0\otimes \gamma^2_1)\oplus (\gamma^1_1 \otimes \gamma^2_0), \\
    \gamma_0 & = \gamma^1_0 \otimes \gamma^2_0,
\end{aligned}
\label{eq:automorphism}
\end{equation}
where $\{\gamma^i_1, \gamma^i_0\}$ is an automorphism for $\mc{C}^i$ ($i = 1, 2$) such that the following diagram commutes:
\begin{equation}
    \begin{tikzcd}
    	{C^i_1} & {C^i_0} \\
    	{C^i_1} & {C^i_0}
    	\arrow["{\partial^i_1}", from=1-1, to=1-2]
            \arrow["{\partial^i_1}", from=2-1, to=2-2]
            \arrow["{\gamma^i_1}", from=2-1, to=1-1]
    	\arrow["{\gamma^i_0}", from=2-2, to=1-2]
    \end{tikzcd}
    \end{equation}
Note that the direct sum in Eq.~(\ref{eq:automorphism}) indicates that $\gamma_1$ takes a block-diagonal form with respect to the two blocks of qubits in Eq.~\eqref{eq:HGP_complex}.

\begin{figure}
    \centering
    \includegraphics[width=0.5\textwidth]{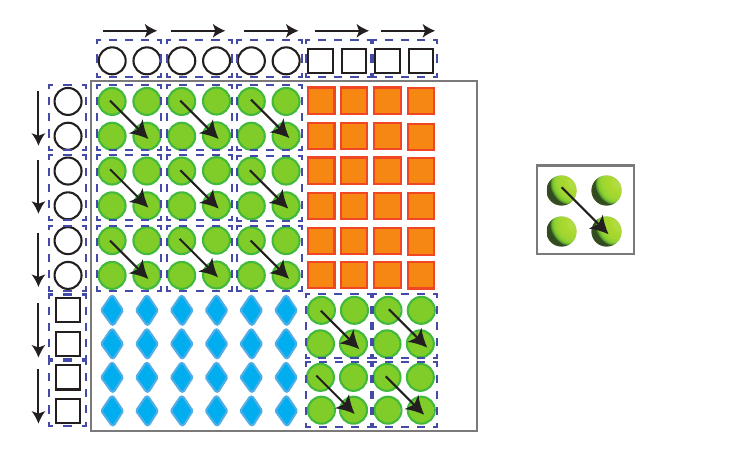}
    \caption{\textbf{Illustration of the logical translation gadget for HGP codes with quasi-cyclic base codes.} Physical block-translations corresponding to products of horizontal and vertical block-cyclic-shifts give rise to translation of the logical block under periodical boundary conditions.}
    \label{fig:translation_gadget}
\end{figure}

For a symmetric HGP code with a base check matrix $H_0$, we assign $\partial_1^{1 T} = \partial_1^2 = H_0$. Furthermore, we consider a quasi-cyclic base code~\cite{huffman2010fundamentals}, whose check and generator matrix are in the following form:
\begin{equation}
    H_0 \in {\mbb{R}_l}^{r_b\times n_b}, \quad G_0 \in {\mbb{R}_l}^{1\times n_b}\ \mr{with}\ G_0[1,1] = \mb{1}, 
    \label{eq:OGSC_form}
\end{equation}
where $\mbb{R}_l := \mbb{F}_2[x]/(x^l - 1)$ denotes the quotient polynomial ring, $x$ represents a circulant matrix of size $l$ that shifts the entries by $1$, and $\mb{1}$ represents the $l\times l$ identity matrix. $l$ is also denoted as the lift size in the literature. Based on the form of $G_0$, this code encodes $l$ codewords, which can be cyclically shifted by cyclically shifting each of the $n_b$ blocks of bits. Moreover, the same block shifts of the bits correspond to cyclically shifting the $r_b$ blocks of checks. Thus, it admits an automorphism $\{\sigma_{c, i}, \sigma_{b, i}\}$, where $\sigma_{c, i}$ and $\sigma_{b, i}$ denote the block-cyclic-shift matrix by $i$ for the checks and the bits, respectively. For example, for $n_b = 2$, $\sigma_{b,1} = \left( \begin{array}{cc}
    x & 0 \\
    0 & x
\end{array}\right)$. Finally, we can construct an automorphism for the HGP code as
\begin{equation}
    \gamma_1 = \left( \begin{array}{cc}
       \sigma_{b,i}\otimes \sigma_{b,j} & 0 \\
       0  & \sigma_{c,i}\otimes \sigma_{c,j} \\
    \end{array}\right),
    \label{eq:HGP_shifts}
\end{equation}
where the two blocks of shifts in Eq.~(\ref{eq:HGP_shifts}) are applied to the two blocks of qubits given by the product of the classical bits and checks, respectively (see Eq.~(\ref{eq:coordinate_system}) and Fig.~\ref{fig:translation_gadget}).  Moreover, the qubits can be further divided into $l\times l$ blocks, each undergoing the same translation by $(i,j)$ under periodic boundary conditions (see Fig.~\ref{fig:translation_gadget}). It is easy to verify that the canonical logical qubits $\{\bar{Q}_{x,y}\}_{x \in [l], y \in [l]}$ (there are $l\times l$ logical qubits in total since each base code encodes $l$ logical bits) are also translated by $(i,j)$ under periodic boundary conditions, i.e.
\begin{equation}
    \bar{P}_{x,y} \rightarrow \bar{P}_{(x + i)\mod l,\  (y + j)\mod l},
\end{equation}
for $P = X, Z$. We denote such a gadget as $\bar{T}_{i,j}$. We an example for $\bar{T}_{1,1}$ with $ l = 2$, $n_b = 3$, and $r_b = 2$ in Fig.~\ref{fig:translation_gadget}. 
 
Finally, we note that not all quasi-cyclic classical codes are in the required form of Eq.~(\ref{eq:OGSC_form}). Generic quasi-cyclic codes might have multiple rows in their block generator matrix $G_0$, corresponding to disjoint blocks of codewords, where only cyclic permutation within each block is permitted; their block generator matrix might not be in the canonical form, i.e. starting with an identity matrix, which is required to define the canonical logical operator basis used in our work. Generator matrices that satisfy Eq.~\eqref{eq:OGSC_form} are referred to as one-generator systematic-circulant (OGSC), and algebraic conditions for obtaining codes with OGSC generator matrices have been explored in the literature~\cite{li2006efficient}. Although we are not aware of any asymptotically good family of OGSC codes that achieve cosntant encoding rates and linear distances, we present several finite-size examples in Table~\ref{tab:QC_code_params} through numerical search, which feature even better parameters than those based on random expander graphs~\cite{grospellier2021combining, tremblay2022constant, xu2024constant}.
\begin{table*}
    \centering
    \begin{tabular}{c|c|c|c|c}
    \hline
    \hline
          Classical-Code Parameter &  Quantum-Code Parameter & Classical Check Matrix & Classical Generator Matrix & Lift Size\\
         \hline
         $[9,3,4]$ & [[117, 9, 4]] & $\left(\begin{array}{ccc}
             x^2 & x^2 & x^2 \\
             x & x^2 & 0 \\
         \end{array}\right)$ & $(1, x, 1 + x)$ & 3\\
         \hline
         $[12,3,6]$ & [[225, 9, 6]] & $\left(\begin{array}{cccc}
             x^2 & x^2 & x^2 & 0\\
             x^2 & 0 & x^2 & x^2 \\
             x^2 & x^2 & x & x^2 \\
         \end{array}\right)$ & $(1, 1+x^2, x^2, 1+x^2)$ & 3\\
         \hline
         $[16, 4, 8]$ & [[400, 16, 8]] & $\left(\begin{array}{cccc}
             x^3 & x^3 & 0 & x^3\\
             x^3 & x^2 & x^3 & x^2 \\
             x^3 & x^3 & x^2 & 0 \\
         \end{array}\right)$ & $(1, 1+x+x^2, 1+x, x+x^2)$ & 4\\
         \hline
          $[20, 5, 9]$ & [[625, 25, 9]] & $\left(\begin{array}{cccc}
             x^4 & 0 & x^4 & x^3\\
             0 & x^3 & x^3 & x^4 \\
             x^3 & x^4 & 0 & x^3 \\
         \end{array}\right)$ & $(1, x+x^2+x^3, 1+x^2+x^3, x+x^2)$ & 5\\
    \hline
    \hline
    \end{tabular}
    \caption{\textbf{Parameters and code matrices of finite-size OGSC classical codes and the resulting HGP codes.}}
    \label{tab:QC_code_params}
\end{table*}

\subsection{Selective inter-block teleportation}
Here, we present a gadget that teleports any subset $\mb{\bar{Q}_0}$ of the logical qubits $\mb{\bar{Q}}$ of a generic HGP code $\mc{Q}$ to the corresponding logical qubits $\mb{\bar{Q}^{\prime}_0}$ (with the same coordinates) of another identical code $\mc{Q}^{\prime}$. Let $\mr{cw}(\mb{\bar{Q}_0})$ ($\mr{rw}(\mb{\bar{Q}_0})$) denote the number of columns (rows) of the logical qubit grid that $\mb{\bar{Q}_0}$ are supported on. We present such a teleportation gadget $\mr{Tel}(\mb{\bar{Q}_0} \rightarrow \mb{\bar{Q}_0^{\prime}})$ in Alg.~\ref{alg:selective_tel} in $O(\min\{\mr{cw}(\mb{\bar{Q}_0}), \mr{rw}(\mb{\bar{Q}_0})\})$ logical cycles. The gadget performs the teleportation in either a column-by-column fashion or a row-by-row fashion using generalized versions of the GPPMs gadget. 

\begin{algorithm}[h!]
\caption{Selective teleportation between two identical HGP codes}\label{alg:selective_tel}
\Input{Two identical HGP codes $\mc{Q}$ and $\mc{Q}^{\prime}$ with logical qubits $\mb{\bar{Q}}$ and $\mb{\bar{Q}^{\prime}}$, respectively; A subset of logical qubits $\mb{\bar{Q}_0} \subseteq \mb{\bar{Q}}$ of $\mc{Q}$. }
\Output{A gadget $\mr{Tel}(\mb{\bar{Q}_0} \rightarrow \mb{\bar{Q}_0^{\prime}})$ that teleports $\mb{\bar{Q}_0}$ to $\mb{\bar{Q}_0^{\prime}}$ in $O(\min\{\mr{cw}(\mb{\bar{Q}_0}), \mr{rw}(\mb{\bar{Q}_0})\})$ logical cycles.}

\If{$\mr{cw}(\mb{\bar{Q}_0}) \geq \mr{rw}(\mb{\bar{Q}_0})$}
{
    \tcp{Perform column-by-column teleportations}
    \For{Each column $j$ of logical qubits $\mb{\bar{Q}_0}|_j$ of $\mb{\bar{Q}_0}$}
    {
    \tcp{Teleport $\mb{\bar{Q}_0}|_j$ to $\mb{\bar{Q}_0}^{\prime}|_j$ by implementing the teleportation circuit in Fig.~\ref{fig:Clifford_circuits}(a) in parallel using generalized versions of the GPPMs gadget.}
    \tcp{First, measure the $ZZ$s in Fig.~\ref{fig:Clifford_circuits}(a) in parallel using selective ancillae and homomorphic CNOTs}
    Prepare an ancilla code $\mc{Q}^{\prime \prime}$ that only contains the $j$-th column of logical qubits (by using the puncturing techniques described in Sec.~\ref{sec:homomorphic_measurements_HGP_technical}) in the $Z$ basis; Then ``mask" the logical qubits $\mb{\bar{Q}^{\prime \prime}}|_j \backslash \mb{\bar{Q}}^{\prime \prime}|_j$ by resetting them to $\overline{\ket{+}}$ using another mask code (see Fig.~\ref{fig:HGP_Grid_PPMs}). \\
    Apply the $\mc{Q}$-controlled homomorphic CNOT between $\mc{Q}$ and $\mc{Q}^{\prime \prime}$ and then $\mc{Q}^{\prime}$-controlled homomorphic CNOT between $\mc{Q}^{\prime}$ and $\mc{Q}^{\prime \prime}$. \\
    Transversally measure $\mc{Q}^{\prime \prime}$ in the $Z$ basis. \\
    \tcp{Then, measure the $X$s in Fig.~\ref{fig:Clifford_circuits}(a)}
    Perform $X$ measurements on $\mb{\bar{Q}_0}|_j$ using a GPPM gadget.
    }
}
\Else
{
Perform row-by-row teleportations, similar to the column-by-column teleportation described above.
}
\end{algorithm}

\subsection{Logical cyclic shift \label{sec:cyclic_shift}}
Here, we present a gadget for performing a cyclic logical shift on a $[[n, k, d]]$ HGP code $\mc{Q}$ with OGSC base classical codes by combining the logical translation gadget in Sec.~\ref{sec:translation_gadget} and the GPPMs gadget. Let $\{\bar{Q}_{i,j}\}_{i \in [M], j \in [N]}$, where $MN = k$, be the 2D grid of the logical qubits. We label them in a zigzag pattern, i.e. $\bar{Q}_l := \bar{Q}_{\lceil l/N\rceil, l \mod N}$. Here, to keep the notation simple, we set $a \mod a = a$ for any $a \in \mbb{Z}$. Then, we can perform a cyclic shift, $\bar{Q}_{l} \rightarrow \bar{Q}_{(l + 1)\mod k}$, in the following two steps:
\begin{enumerate}
    \item Perform a horizontal logical translation $\bar{T}_{0,1}$. This realizes most of the cyclic shifts, except for the rightmost column of logical qubits $\{\bar{Q}_{N i}\}_{i \in [M]}$, which are permuted to the leftmost column. It thus remains to do a vertical cyclic shift of this leftmost column.
    \item Teleport the first column of logical qubits of $\mc{Q}$ to another identical code $\mc{Q}^{\prime}$ in $O(1)$ logical cycles using the selective teleportation gadget in Alg.~\ref{alg:selective_tel},  perform a vertical logical translation $\bar{T}_{1,0}$ on $\mc{Q}^{\prime}$, and then teleport the first column of logical qubits back to $\mc{Q}$.
\end{enumerate}

\subsection{Parallel logical $\ket{i}$ state preparation \label{sec:i_prep}}
Here, we present a gadget for preparing $\overline{\ket{i}}^{\otimes k}$ for a $[[n, k, d]]$ HGP code. We first prepare two identical HGP codes $\mc{Q}_A$ and $\mc{Q}_B$, both initialized in $\overline{\ket{0}}^{\otimes k}$ states. Then, we apply the following sequence of operations iteratively for $\mc{Q}_A$:
\begin{equation}
    M_{X}^D \rightarrow \textrm{CZ-S} \rightarrow \bar{T}_{i,i},
    \label{eq:Y_sequence}
\end{equation}
for $i \in [\sqrt{k}/2]$, and the following sequence of operations iteratively for $\mc{Q}_B$:
\begin{equation}
    M_{X}^D \rightarrow \textrm{CZ-S} \rightarrow \bar{T}_{-i,-i},
    \label{eq:Y_sequence_2}
\end{equation}
for $i \in [\sqrt{k}/2 - 1]$. $M_X^D$ denotes a subroutine for measuring all diagonal logical qubits in the $X$ basis non-destructively, using Alg.~\ref{alg:X_measurement_diagonal} in $O(\log k)$ logical cycles. As shown in Fig.~\ref{fig:Y_Prep}, each sequence in Eq.~\eqref{eq:Y_sequence} or Eq.~\eqref{eq:Y_sequence_2} generates $\sqrt{k}$ $\overline{\ket{i}}$ states on the diagonal qubits utilizing the CZ-S gate and then distributes them to non-diagonal qubits utilizing the logical translation gadget (see Sec.~\ref{sec:translation_gadget}). The two sequences of operations then fill $\mc{Q}_A$ and $\mc{Q}_B$ with two complementary half blocks of $\overline{\ket{i}}$ states, respectively. 
Finally, we merge the $\overline{\ket{i}}$ states in $\mc{Q}_B$ into $\mc{Q}_A$ by performing the following transversal operations: Initially, each pair of qubits in $A$ and $B$ are either stabilized by $\langle \bar{Y}_A, \bar{Z}_B\rangle$ or $\langle \bar{Z}_A, \bar{Y}_B\rangle$; Performing transversal $\bar{X}_A \bar{X}_B$ measurements project all the pairs into the same entangled state stabilized by $\langle \bar{X}_A\bar{X}_B, \bar{Y}_A \bar{Z}_B, \bar{Z}_A \bar{Y}_B\rangle$ (up to some Pauli corrections); Final transversal $Z$ measurements on $B$ project each pair into a product state stabilized by $\langle Y_A, Z_B \rangle$, which indicates that the $Y$ states in $B$ are all merged into $A$. 

Since each sequence in Eq.~\eqref{eq:Y_sequence} or Eq.~\eqref{eq:Y_sequence_2} takes $O(\log k)$ logical cycles and there are $O(\sqrt{k})$ sequences in total, the entire gadget takes $O(\sqrt{k}_0 \log k)$ logical cycles.

\begin{figure}[h!]
    \centering
    \includegraphics[width=0.5\textwidth]{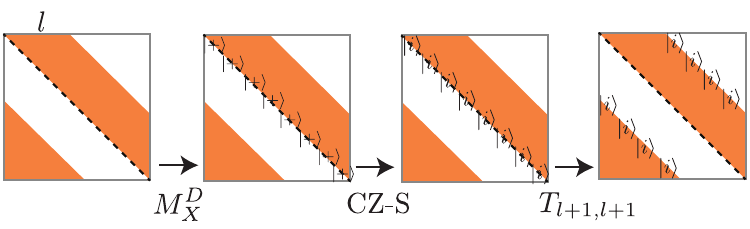}
    \caption{\textbf{Illustration of one sequence of operations of the iterative gadget for preparing $\overline{\ket{i}}^{\otimes k}$.} The orange region and the white region represent logical qubits in $\overline{0}$ and $\overline{i}$ states, respectively. }
    \label{fig:Y_Prep}
\end{figure}

\begin{algorithm}[h!]
\caption{The $M_X^D$ subroutine in Eq.~\eqref{eq:Y_sequence} for measuring the diagonal qubits of a HGP code $\mc{Q}$ in $X$ basis non-destructively.}\label{alg:X_measurement_diagonal}

Prepare an identical ancilla code $\mc{Q}^{\prime}$ in $\overline{\ket{+}}^{\otimes k}$.\\

For simplicity, assume that $\sqrt{k} = 2^m$ for some $m \in \mbb{Z}$. \\

\For{$t \in [m]$}
{
\tcp{Reset the non-diagonal logical qubits of $\mc{Q}^{\prime}$ to $\overline{\ket{0}}$ by recursively measuring subgrids that do not overlap with the diagonal line $i = j$. See Fig.~\ref{fig:digonal_measurements} for an illustration of this step.}
Apply the $\mr{GPPM}(Z, \mc{E}_t, \mc{E}^{\prime}_t)$ gadget to $\mc{Q}^{\prime}$, where 
$\mc{E}_t := \bigcup_{j = 0,1,\cdots, 2^{t - 1} - 1}\{\{\frac{\sqrt{k}}{2^t}(2j) + 1\},\cdots, \{\frac{\sqrt{k}}{2^t}(2j + 1)\}\}$ and 
$\mc{E}^{\prime}_t := \bigcup_{j = 0,1,\cdots, 2^{t - 1} - 1}\{\{\frac{\sqrt{k}}{2^t}(2j + 1) + 1\}, \cdots, \{\frac{\sqrt{k}}{2^t}(2j + 2)\}\}$.
\\
Apply the $\mr{GPPM}(Z, \mc{E}^{\prime}_t, \mc{E}_t)$ gadget to $\mc{Q}^{\prime}$. \\
}

Perform transversal CNOTs from $\mc{Q}^{\prime}$ to $\mc{Q}$.\\

Transversally measure $\mc{Q}^{\prime}$ in the $X$ basis.\\

\label{alg:log_k_X_prep}
\end{algorithm}
\begin{figure}
    \centering
    \includegraphics[width=0.5\textwidth]{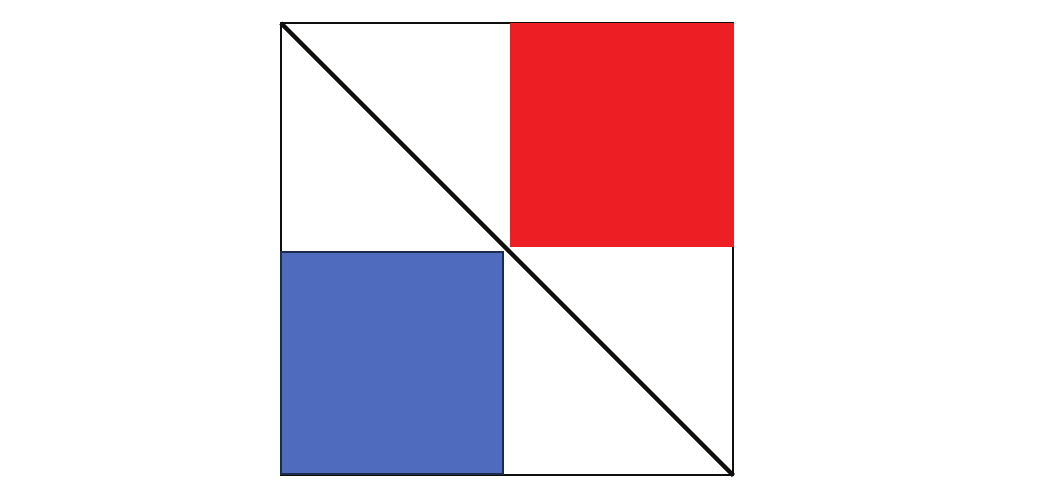}
    \caption{\textbf{Illustration for measuring all the non-diagonal qubits of an HGP code in log depth for step 3 of Alg.~\ref{alg:X_measurement_diagonal} at $t = 1$.} Performing the single-qubit measurements can be viewed as ``filling" an empty patch except for the diagonal. The first time step fills the blue and the red squares with two GPPMs gadget. This leaves two sub-squares to be filled. This procedure can then be applied recursively to fill the entire square in log depth. Note that later steps might fill the regions that have already been filled by the previous steps, but importantly, the diagonal line will not be filled.}
    \label{fig:digonal_measurements}
\end{figure}

\subsection{Parallel logical Hadamard gates \label{sec:transversal_H}}
Here, we provide a gadget for implementing logical Hadamard gates transversely on logical qubits of a $[[n, k, d]]$ symmetric HGP code. We first apply the fold-transversal H-SWAP gate, which applies the desired transversal Hadamards up to extra swaps along the diagonal (see Table~\ref{tab:HGP_gadgets}), then cancel the extra logical swaps by utilizing the GPPMs gadget in a similar fashion as that for transversely preparing the $\overline{\ket{i}}$ states. 

\begin{figure}
    \centering
    \includegraphics[width=0.5\textwidth]{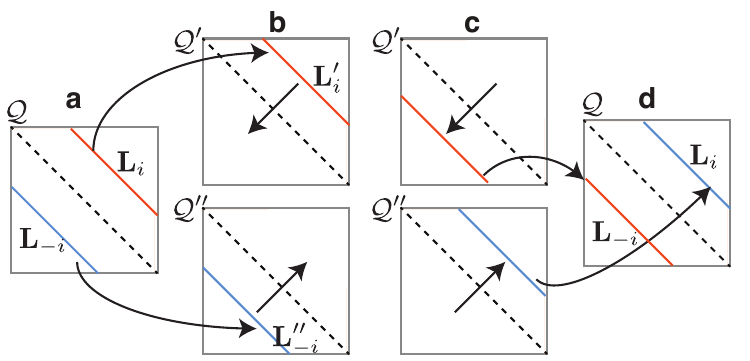}
    \caption{\textbf{Illustration for one sequence of operations of the gadget for swapping logical qubits of a symmetric HGP code along the diagonal.} (a) to (b): teleportation of the two mirrored lines of logical qubits from $\mc{Q}$ to $\mc{Q}^{\prime}$ and $\mc{Q}^{\prime \prime}$. (b) to (c): apply the logical translation gadget to $\mc{Q}^{\prime}$ and $\mc{Q}^{\prime \prime}$. (c) to (d): teleportation back into $\mc{Q}. $}
    \label{fig:SWAP}
\end{figure}

To implement the logical swaps on a code $\mc{Q}$, we prepare two identical ancilla codes $\mc{Q}^{\prime}$ and $\mc{Q}^{\prime \prime}$. Let $\mb{L}_i$ denote a line of logical qubits $\{\bar{Q}_{j, j + i}\}_{j \in [\sqrt{k}]}$ (only including qubits with valid coordinates $\in [\sqrt{k}]\times [\sqrt{k}]$) with an offset $i$ to the diagonal. The swap gadget amounts to swapping the ``twin" lines $\mb{L}_i \leftrightarrow \mb{L}_{-i}$ for $i \in [\sqrt{k} - 1]$. As illustrated in Fig.~\ref{fig:SWAP}, each such swap can be implemented by the following sequence of operations:
\begin{equation}
\begin{aligned}
    \mr{Tel}(\mb{L}_i \rightarrow \mb{L}_i^{\prime})&;\mr{Tel}(\mb{L}_{-i} \rightarrow \mb{L}_{-i}^{\prime \prime}) \rightarrow \bar{T}^{\prime}_{-2i, -2i}; \bar{T}^{\prime \prime}_{2i, 2i}\\
    & \rightarrow \mr{Tel}(\mb{L}^{\prime}_{-i} \rightarrow \mb{L}_{-i}); \mr{Tel}(\mb{L}^{\prime \prime}_{i} \rightarrow \mb{L}_{i}),
\end{aligned}
\label{eq:line_swap}
\end{equation}
where each of the teleportation between two identified lines of qubits across two codes, e.g. $\mr{Tel}(\mb{L}_i \rightarrow \mb{L}^{\prime})$ from $\mb{L}_i$ to $\mb{L}_i^{\prime}$, can be implemented by the teleportation circuit in Fig.~\ref{fig:Clifford_circuits}(a), where the $ZZ$ measurements on $\mb{L}_i$ and $\mb{L}_i^{\prime}$ are implemented using another identical ancilla code $\mc{Q}^{\prime \prime \prime}$, whose logical qubits $\mb{L}_i^{\prime \prime}$ are initialized in $\overline{\ket{0}}$ while the rest are initialized in $\overline{\ket{+}}$. Such a selective initialization can be implemented by measuring the diagonal qubits in $Z$ basis in $O(\log k)$ cycles using Alg.~\ref{alg:log_k_X_prep} (with $Z$ and $X$ flipped and up to an extra $X$ measurements using a single $X$-GPPMs gadget) and then distributing the $\overline{\ket{0}}$ states to $\mb{L}^{\prime \prime \prime}_i$ by performing a logical translation $\bar{T}^{\prime \prime \prime}_{i,i}$. The $X$ measurements in Fig.~\ref{fig:Clifford_circuits}(a) on $\mb{L}_i$ can be implemented similarly using ancilla code $\mc{Q}^{\prime \prime \prime}$ whose logical qubits $\mb{L}_i^{\prime \prime}$ are initialized in $\overline{\ket{+}}$ while the rest are initialized in $\overline{\ket{0}}$.

Overall, the transversal Hadamard gates, implemented by a H-SWAP gate and $O(\sqrt{k})$ sequences of line swaps in Eq.~\eqref{eq:line_swap}, take $O(\sqrt{k} \log k)$ logical cycles. 

\section{Proof of Theorem~\ref{theorem:parallel_Clifford_gates} \label{sec:proof_parallel_Clifford_gates}}
\begin{figure*}
    \centering
    \includegraphics[width=1\textwidth]{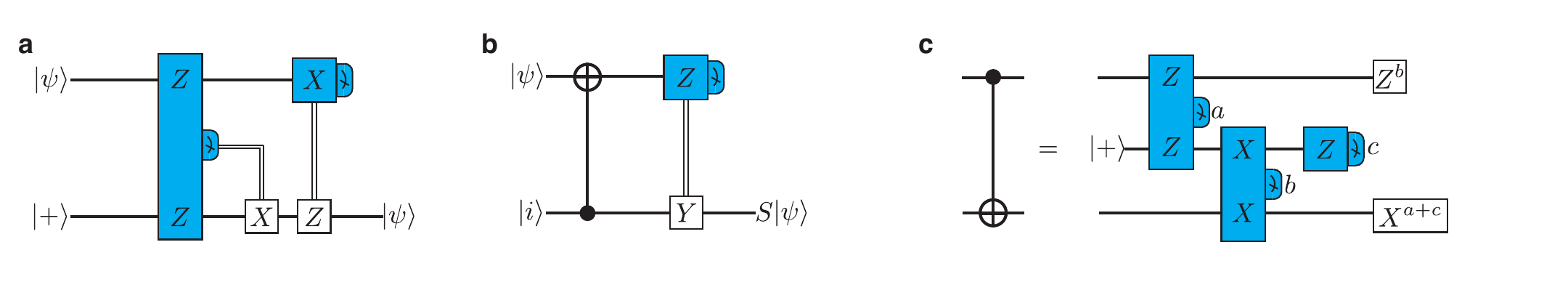}
    \caption{Circuits for implementing (a) teleportation (b) teleported $S$ gate (c) measurement-based CNOT. }
    \label{fig:Clifford_circuits}
\end{figure*}

In this section, we provide the proof of Theorem~\ref{theorem:parallel_Clifford_gates} regarding the parallel implementation of arbitrary Clifford gates (consisting of Hadamards, $S$ gates, and CNOTs) using the gadgets in Table~\ref{tab:HGP_gadgets}. We first prove that a layer of $O(k)$ Clifford gates can be implemented in parallel in $O(k^{3/4})$ logical cycles using the gadgets in Table~\ref{tab:HGP_gadgets}, including the translation gadget. Then, it will become apparent that the translation gadget is essential only for the parallelism, but not necessary for generating the full Clifford group. We note that this task of implementing Clifford gates in parallel is a compilation problem with a restricted gate set and we only provide an upper bound on the circuit depth using a constructive compilation.

We consider a layer of random Clifford gates containing $O(k)$ Hadamard gates, $O(k)$ $S$ gates, and $O(k)$ random CNOTs, supported on the logical qubits $\mb{\bar{Q}}|_{H}$, $\mb{\bar{Q}}|_{S}$, and $\mb{\bar{Q}}|_{\mr{CNOT}}$ of a HGP code $\mc{Q}$, respectively. We implement these three types of gates separately.

First, we note that we can teleport any subset $\mb{\bar{Q}}_0$ of the logical qubits of $\mc{Q}$ transversely to the corresponding logical qubits of another identical code $\mc{Q}^{\prime}$, and vice versa, in $O(\sqrt{k})$ logical cycles using the selective teleportation gadget in Alg.~\ref{alg:selective_tel}. 

To implement the Hadamard gates, we teleport $\mb{\bar{Q}}|_{H}$ transversely to $\mb{\bar{Q}}^{\prime}|_{H}$ of another code $\mc{Q}^{\prime}$. Then, we apply the transversal Hadamard gates (without extra swaps) on all the logical qubits of $\mc{Q}^{\prime}$ in $O(\sqrt{k}\log k)$ logical cycles using the subroutine described in Sec.~\ref{sec:transversal_H}. Finally, we teleport $\mb{\bar{Q}}^{\prime}|_{H}$ back to $\mb{\bar{Q}}|_{H}$.

We implement the $S$ gates using teleported gates. As shown in Fig.~\ref{fig:Clifford_circuits}(b), we prepare another identical code $\mc{Q}^{\prime}$, where $\mb{\bar{Q}}^{\prime}|_{S}$ are initialized in $\ket{i}$ states while the rest are initialized in $\ket{+}$ states. Then transversal CNOTs between $\mb{\bar{Q}}$ and $\mb{\bar{Q}}^{\prime}$ followed by transversal measurements of $\mb{\bar{Q}}$ teleport the logical qubits from $\mc{Q}$ to $\mc{Q}^{\prime}$ with the desired $S$ gates applied. The selective initialization of $\mc{Q}^{\prime}$ can be implemented by first preparing all the logical qubits in $\ket{i}$ using the subroutine in Sec.~\ref{sec:i_prep} in $O(\sqrt{k}\log k)$ logical cycles, followed by resetting the qubits $\mb{\bar{Q}}^{\prime}\backslash \mb{\bar{Q}}^{\prime}|_{S}$ to $\ket{+}$ using GPPMs in a column-by-column fashion in $O(\sqrt{k})$ logical cycles.

\begin{figure}
    \centering
    \includegraphics[width=0.5\textwidth]{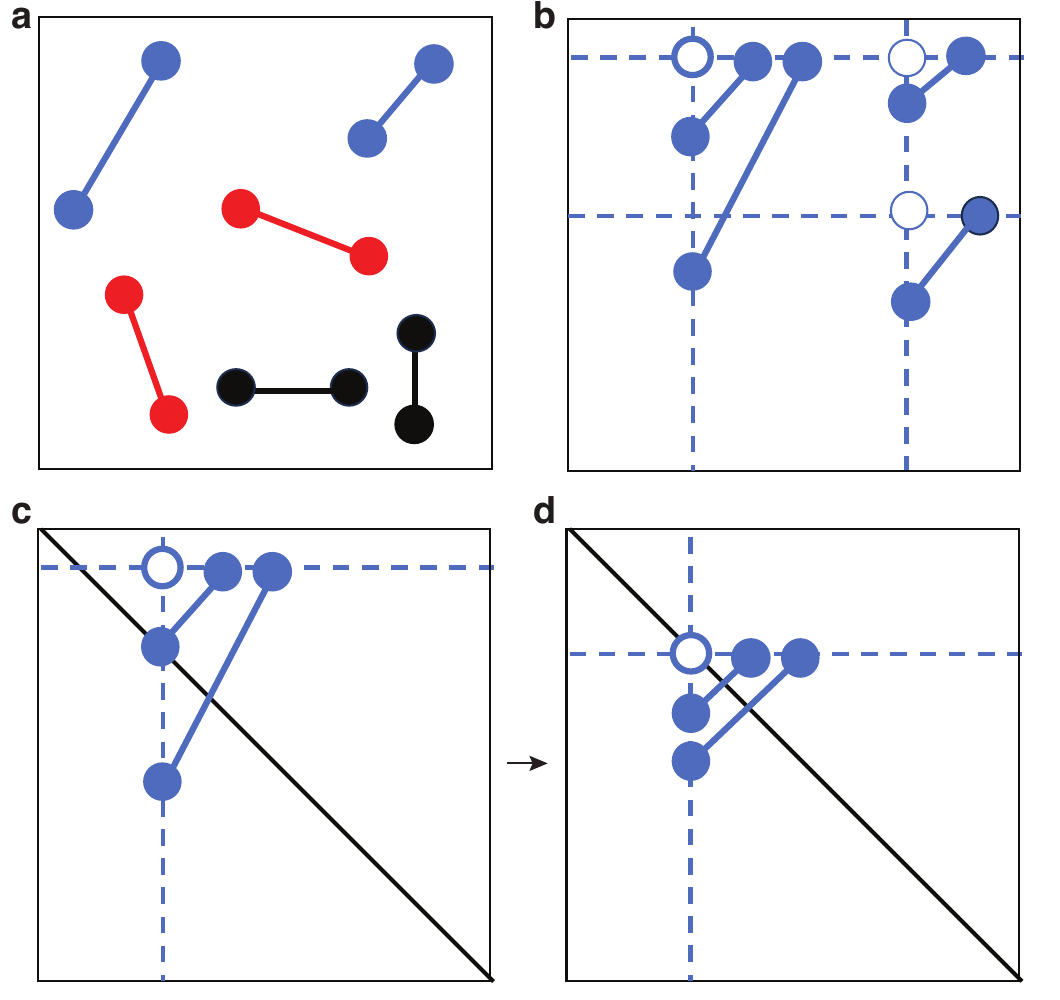}
    \caption{\textbf{Illustration for implementing a layer of random CNOTs using the gadgets in Table~\ref{tab:HGP_gadgets}.} (a) Classification of all the CNOTs into three types: the ``Aligned"-CNOTs (black), the ``TLBR"-CNOTs (red), and the ``TRBL"-CNOTs (blue). (b) The ``TRBL"-CNOTs form different clusters according to their shared ancilla (empty circles) and the clusters are partitioned into sparse clusters (the right two) and dense clusters (the left one). (c) to (d): Shifting and ``symmetrizing" the CNOTs of a dense cluster such that each CNOT is mirrored along the diagonal line $i = j$.}
    \label{fig:random_CNOTS}
\end{figure}
To implement the CNOTs, we first classify them into three types: the ``Aligned"-CNOTs, acting on logical qubits within the same row or the same column (see the black CNOTs in Fig.~\ref{fig:random_CNOTS}(a)); the ``TLBR"-CNOTs, acting on pairs of qubits oriented from the top left to the bottom right (see the red CNOTs in Fig.~\ref{fig:random_CNOTS}(a)); and the ``TRBL"-CNOTs, acting on pairs of qubits oriented from the top right to the bottom left (see the blue CNOTs in Fig.~\ref{fig:random_CNOTS}(a)). The ``Aligned"-CNOTs are the easiest to implement. We implement the vertically aligned CNOTs in a column-by-column fashion. For each qubit pair in a column, we introduce a distinct ancilla on the same column and implement the measurement-based CNOT consisting of two Bell measurements (followed by measuring the ancilla). The $Z$ (and similarly, $X$-) Bell measurements required for all the vertical CNOTs in a column can be implemented using a single GPPMs gadget as long as they do not share ancillae. We can assume that we have enough ancillae since otherwise we can simply teleport part of the vertical CNOTs (at most half) to another empty patch and implement them separately. As a result, all the vertically-aligned CNOTs can be implemented in $O(\sqrt{k})$ logical cycles and similarly, all the horizontally aligned CNOTs can be implemented in $O(\sqrt{k})$ logical cycles in a row-by-row fashion.

Finally, we finish the proof by showing that the ``TRBL"-CNOTs can be implemented in $O(k^{3/4})$ logical cycles (same for the ``TLBR"-CNOTs). As shown in Fig.~\ref{fig:random_CNOTS}(b), for any ``TRBL"-CNOT acting on a pair of qubits $\bar{Q}_{i,j}$ and $\bar{Q}_{i^{\prime}, j^{\prime}}$, where $i > i^{\prime}$ and $j < j^{\prime}$, we implement it by introducing an ancilla $\bar{Q}_{i^{\prime},j}$ (see the empty circles in Fig.~\ref{fig:random_CNOTS}(b)) and implementing the circuit in Fig.~\ref{fig:Clifford_circuits}(c), where the $Z$- and $X$-Bell measurements are applied on horizontal and vertical pairs, respectively. We group the CNOTs into different clusters according to the ancilla they share. We further partition these clusters into two types: dense clusters with more than $m$ CNOTs and sparse clusters otherwise, where $m$ is some constant integer that we will specify later (see Fig.~\ref{fig:random_CNOTS}(b) for an illustration with $m = 1$). We can implement the sparse clusters in a column-by-column and row-by-row fashion, similar to that for implementing the ``Aligned"-CNOTs. Specifically, we implement them in $\leq m$ sequences and for each sequence, we pick one CNOT from each cluster and implement all these picked CNOTs by first performing all the required vertical $X$-Bell measurements in a column-by-column fashion in $O(\sqrt{k})$ logical cycles and then performing all the required $Z$-Bell measurements in a row-by-row fashion in $O(\sqrt{k})$ logical cycles. In total, implementing these sparse clusters thus takes $T_S = O(m \sqrt{k})$ logical cycles. 

Lastly, we implement the dense clusters one-by-one and for each dense cluster with $n$ qubit pairs $\{\{\bar{Q}_{r_i, \beta}, \bar{Q}_{\alpha, c_i}\}\}_{i \in [n]}$ sharing a common ancilla $\bar{Q}_{\alpha, \beta}$, we implement it in parallel in $O(1)$ logical cycles. The strategy is to first shift and ``symmetrize" these pairs such that each pair is mirrored along the diagonal line $i = j$ (see Fig.~\ref{fig:random_CNOTS}(a) to (b) for an illustration), and then implement these symmetrized CNOTs in parallel using the CZ-S gate (up to some Hadamards), which applies pair-wise CZs folded along the diagonal. 
Specifically, we first teleport the cluster to another empty code $\mc{Q}^{\prime}$ in $O(1)$ logical cycles by first teleporting the row and then the column. Without loss of generality, assume that $\alpha < \beta$. Then, we shift the entire cluster such that the ancilla $\bar{Q}^{\prime}_{\alpha, \beta}$ is on the diagonal $(\beta, \beta)$ by applying the translational gate $\bar{T}_{\beta - \alpha, 0}$. Now, each pair $(\bar{Q}^{\prime}_{r_i, \beta}, \bar{Q}^{\prime}_{\alpha, c_i})$ gets shifted to $(\bar{Q}^{\prime}_{\beta - \alpha + r_i, \beta}, \bar{Q}^{\prime}_{\beta, c_i})$. Then, to symmetrize these pairs, we swap the column qubits $\bar{Q}^{\prime}_{\beta - \alpha + r_i, \beta} \leftrightarrow \bar{Q}^{\prime}_{c_i, \beta}$ for $i \in [n]$. Each of the swaps can be done by introducing an ancilla within the same column and performing pairs of Bell measurements (similar to implementing a CNOT). We can assume that there are enough empty ancillae and the target qubit locations do not overlap with the original qubit locations so that the $n$ swaps can be done in $O(1)$ logical cycles using the parallel GPPMs (otherwise, we can teleport at most half of the qubit pairs to another code and implement them separately and in parallel). Finally, we apply the fold-transversal CZ-S gate to implement the symmetrized CNOTs (up to some Hadamards that can be addressed separately, as described earlier) and revert the above process to teleport the qubits back to their original position in $\mc{Q}$. Since each of the dense clusters can be implemented in $O(1)$ logical cycles and there are at most $k/m$ such dense clusters, implementing all the dense clusters takes $T_D = O(k/m)$ logical cycles.

To sum up, implementing all the ``TRBL"-CNOTs thus takes $T = T_S + T_D = O(m\sqrt{k}) + O(k/m)$ logical cycles. By choosing $m = \Theta(k^{1/4})$, we have $T = O(k^{3/4})$, which completes the proof of Theorem~\ref{theorem:parallel_Clifford_gates}. 

Finally, we comment that the full Clifford group can be generated without using the translation gadget and assuming that the base codes are quasi-cyclic. The construction for any selective $H$, $S$, or CNOT uses essentially the same ingredients as described above, although different Clifford gates might have to be executed sequentially in the absence of the translation gadget. 

\section{Parallel PPMs for 3D homological product codes \label{appendix:3D_PPMs}}
In this section, we provide more details on constructing the parallel PPMs gadget in Def.~\ref{def:CPPMs} on any subcube of the logical qubits of a 3D homological product code.

Similar to that for the HGP code, the construction here works for a canonical basis of logical qubits that we describe in the following.
As shown in Eq.~\ref{eq:3D_code_complex}, 
\begin{equation}
    \begin{tikzcd}[column sep=small]
    	{Q_0} & & {C^1_0\otimes C^2_0 \otimes C^3_0} & \\
    	Q_1 & C^1_0\otimes C^2_0\otimes C^3_1 & C^1_0\otimes C^2_1\otimes C^3_0 & C^1_1\otimes C^2_0\otimes C^3_0\\
     Q_2 & C^1_0\otimes C^2_1\otimes C^3_1 & C^1_1\otimes C^2_0\otimes C^3_1 & C^1_1\otimes C^2_1\otimes C^3_0\\
     Q_3 & & C^1_1\otimes C^2_1\otimes C^3_1
    	\arrow["{M_X}", from=2-1, to=1-1]
            \arrow["{H_X}", from=3-1, to=2-1]
            \arrow["{H_Z^T}", from=4-1, to=3-1]
            \arrow["", from=2-2, to=1-3]
    	\arrow["", from=2-3, to=1-3]
     \arrow["", from=2-4, to=1-3]
     \arrow["", from=3-2, to=2-2]
     \arrow["", from=3-2, to=2-3]
     \arrow["", from=3-3, to=2-2]
     \arrow["", from=3-3, to=2-4]
     \arrow["", from=3-4, to=2-3]
     \arrow["", from=3-4, to=2-4]
     \arrow["", from=4-3, to=3-2]
     \arrow["", from=4-3, to=3-3]
     \arrow["", from=4-3, to=3-4]
    \end{tikzcd},
    \label{eq:3D_code_complex}
\end{equation}
a 3D homological product code is constructed by taking the total complex of the tensor product of three length-$1$ complexes $\{C^{i}_1 \xrightarrow{\partial^i_1} C^i_0\}_{i = 1,2,3}$. 
We associate the first three vector spaces $Q_3, Q_2$ and $Q_1$ of the total complex with $Z$ checks, qubits, $X$ checks, respectively. The extra vector space $Q_0$ is associated with $X$ meta checks that check the $X$ checks: $M_X H_X = 0$.
We assign the base complexes with three classical codes with check matrices $\{H_i\}_{i=1,2,3}$ as follows:
\begin{equation}
    \partial^1_1 = H_1, \quad \partial^2_1 = H_2, \quad \partial^3_1 = H_3^T.
\end{equation}
Then the check matrices $H_X$ and $H_Z$ and the $X$ meta check matrix $M_X$ are given by
\begin{equation}
\begin{aligned}
    H_Z & = \left(H_1^T\otimes I\otimes I, I\otimes H_2^T\otimes I, I\otimes I \otimes H_3 \right), \\
    H_X & = \left( \begin{array}{ccc}
        I\otimes H_2 \otimes I & H_1\otimes I\otimes I & 0 \\
        I\otimes I\otimes H_3^T & 0 & H_1\otimes I \otimes I \\
        0 & I\otimes I \otimes H_3^T & I\otimes H_2 \otimes I \\
    \end{array}\right), \\
    M_X & = \left( I\otimes I \otimes H_3^T, I\otimes H_2 \otimes I, H_1\otimes I\otimes I\right).
\end{aligned}
\label{eq:check_mat_3D_codes}
\end{equation}


Based on Eq.~\eqref{eq:check_mat_3D_codes}, we can derive a canonical basis of logical operators using the K\"{u}nneth formula~\cite{bravyi2013homological, hatcher2005algebraic}:
\begin{equation}
\begin{aligned}
    \mb{\bar{X}} & = \{\left(\begin{array}{c}
         0\\
         0 \\
         f\otimes g\otimes h\\
    \end{array} \right) \mid f \in \mr{rs}(H_1)^{\bullet}, g \in \mr{rs}(H_2)^{\bullet}, h \in \ker{H_3}\}, \\
    \mb{\bar{Z}} & = \{\left(\begin{array}{c}
         0\\
         0 \\
         f^{\prime}\otimes g^{\prime}\otimes h^{\prime}\\
    \end{array} \right) \mid f^{\prime} \in \ker{H_1}, g^{\prime} \in \ker{H_2}, h^{\prime} \in \mr{rs}(H_3)^{\bullet}\}.
\end{aligned} 
\label{eq:3D_logicals}
\end{equation}

Again, without loss of generality, we assume that each check matrix $H_{\alpha}$ can be row-reduced to their canonical form, from which we can derive the canonical form of $\ker{H_{\alpha}}$ and $\mr{rs}(H_{\alpha})^{\bullet}$ in Eq.~\eqref{eq:ker_im}. Then, we can find a canonical basis for the logical operators in Eq.~\eqref{eq:3D_logicals} forming conjugate pairs $\{(\bar{X}_{i,j,k}, \bar{Z}_{i,j,k})\}_{i \in [k_1], j \in [k_2], k \in [k_3]}$, where:
\begin{equation}
    \begin{aligned}
        \bar{X}_{i,j,k} & =  \left(\begin{array}{c}
         0\\
         0 \\
         e^{n_1}_i \otimes e^{n_2}_j\otimes b^3_k\\
    \end{array} \right), \\
    \bar{Z}_{i,j,k} & =  \left(\begin{array}{c}
         0\\
         0 \\
         b^{1}_i \otimes b^{2}_j\otimes e^{n_3}_k\\
    \end{array} \right).
    \end{aligned}
\end{equation}

The logical operators are all supported on $\mb{B_1}\times \mb{B_2} \times \mb{B_3} \simeq [n_1]\times [n_2]\times [n_3]$ and the logical qubits $\mb{\bar{Q}} = \{\bar{Q}_{i,j,k}\}_{(i,j,k) \in [k_1]\times [k_2]\times [k_3]}$ can be arranged on a $[k_1]\times [k_2] \times [k_3]$ cube, where the logical operator pairs of $\bar{Q}_{i,j,k}$ intersects on the physical qubit $Q_{i,j,k}$.

With this canonical logical basis, we can implement the Cube PPMs gadget on a 3D code $\mc{Q}$ in Def.~\ref{def:CPPMs} using essentially the same two-step protocol as that for the GPPMs for HGP codes (see Alg.~\ref{alg:HGP_grid_PPMs} and Fig.~\ref{fig:HGP_Grid_PPMs}). So we omit the details here and only sketch the protocol: (1) Construct an ancilla $\mc{Q}^{\prime}$ by performing puncturing and augmenting on $H_1$ and $H_2$ according to $\mc{E}_x$ and $\mc{E}_y$, and then construct a mask code $\mc{Q}^{\prime \prime}$ by performing puncturing and augmenting on $H_3$ according to $\mc{E}_z$. (2) Prepare the ancilla code $\mc{Q}^{\prime}$ in the logical $Z$ basis and reset some logical qubits to $\overline{\ket{+}}$ and some to GHZ states using $\mc{Q}^{\prime \prime}$. (3) Perform the desired CPPMs on $\mc{Q}$ using $\mc{Q}^{\prime}$. 

\section{Single-shot state preparation for $3D/4D$ homological product codes \label{sec:single_shot_state_prep}}
\begin{definition}[Reduced weight]
Given a binary check matrix $H \in \mathbb{F}_2^{m\times n}$ and an error $e \in \mathbb{F}_2^n$, we define the reduced weight of $e$ w.r.t. $H$ as $|e|_{H} := \min \{|e^*|, H e^* = H e\}$.
\end{definition}

\begin{definition}[Soundness]
Let $t$ be an integer, and $f: \mathbb{Z} \rightarrow \mathbb{R}$ be some monotonically increasing function with $f(0) = 0$. Given a binary check matrix $H \in \mathbb{F}_2^{m\times n}$, we say it is $(t,f)$-sound if for any $e \in \mathbb{F}_2^n$ such that $|H e| < t$, we have
\begin{equation}
    |e|_H \leq f(|H e|).
    \label{eq:soundness_eq}
\end{equation}
\end{definition}

\begin{definition}[Single-shot state preparation] 
For a state preparation protocol that prepares a $|0\rangle_L$ ($|+\rangle_L$) state by measuring one round of $X$ ($Z$) checks associated with $H_X$ ($H_Z$) and applies corresponding correction, we say that such a protocol is $(q, f)$-single shot if for any syndrome error $s_e$ with $|s_e| < q$ that occurs during the check measurement, the corrected output differs from $\ket{0}_L$ ($\ket{+}_L$) by an error $E$ that satisfies:
\begin{equation}
    |E|_{H_X} (|E|_{H_Z}) \leq f(2|s_e|).
\end{equation}
\end{definition}

\begin{lemma}[Soundness is sufficient for single-shot state preparation]
For a CSS code with a $(t,f)$-sound $X$ ($Z$) check matrix $H_X$ ($H_Z$) of single-shot distance $d_{SS}$, there exists a $(q, f)$ single shot protocol for preparing $\ket{0}_L$ ($\ket{1}_L$), where $q = \frac{1}{2} \min\{t, d_{SS}\}$.
\label{lemma:soundness_imply_singleshot}
\end{lemma}
\begin{proof}
    We consider the case for preparing $\ket{0}_L$ by measuring one round of $X$ checks, and the other case is mirrored. Let $s$ be a (random) measured syndrome pattern in the absence of measurement errors, and $s$ corresponds to a $Z$ error $E_0$ such that $H_X E_0 = s$. Let $s_e$ be a syndrome error that adds to $s$. Let $M_X$ be the metacheck matrix for $H_X$ that satisfies $M_X H^T_X = 0 \mod 2$. We apply the following two-stage correction protocol:
    \begin{itemize}
        \item Find a minimum-weight syndrome correction $s_r$ such that the corrected syndrome $s' := s + s_e + s_r$ passes the meta checkes, i.e., $s' \in \ker{M_X}$.
        \item If $s'$ is a valid syndrome, i.e. $s' \in \mathrm{Im}(H_X)$, find a $Z$ Pauli correction $E_r$ that matches the corrected syndrome, i.e. $H_X E_r = s + s_e + s_r$; Otherwise, declare a logical failure. 
    \end{itemize}

    Now, we prove that the residual error $E = E_r E_0$ of the above protocol satisfies $|E|^r \leq f(2|s_e|)$ if $|s_e| < \frac{1}{2}\min\{t, d_{SS}\}$. First, we prove by contradiction that the corrected syndrome $s'$ is a valid syndrome. Assume $s'$ is not a valid syndrome. Since $s$ is a valid syndrome, $s_e + s_r$ must not be a valid syndrome. However, $|s_e + s_r| \leq 2 |s_e| < d_{SS}$ (by the minimum-weight assumption of $s_r$). This would imply that there exists a $s_0 \in \mathrm{Ker}(M_X)\backslash \mathrm{Im}(H_X)$ with $|s_0| < d_{SS}$, which leads to a contradiction. Therefore, the corrected syndrome $s'$ will be a valid syndrome. Next, we show that the residual error $E$ has a small reduced weight. Since $|H_Z E| = |H_Z (E_r E_0)| = |s_e + s_r| \leq 2 |s_e| < t$, we have, by the soundness property in Eq.~\eqref{eq:soundness_eq},
    \begin{equation}
        |E|^r \leq f(|s_e + s_r|) \leq f(2|s_e|).
    \end{equation}
\end{proof}

\begin{proposition}[4D homological product codes support single-shot state preparation]
A 4D homological product code out of four identical $[n_1, k_1, d_1]$ classical base codes support $(q, f)$-single shot state preparation in both the $X$ and the $Z$ basis, where $q = \frac{1}{2}d_1$ and $f(x) = x^3/4$.
\end{proposition}
\begin{proof}
    According to Ref.~\cite{campbell2019theory}, a 4D homological product code inherently has $(t,f)$-soundness for both their $X$ and $Z$ checks, where $t = d_1$ and $f(x) = x^3/4$. Moreover, it has a single-shot distance $d_{SS} = \infty$. Then, according to Lemma~\ref{lemma:soundness_imply_singleshot}, it supports $(q,f)$-single shot state preparation in both $X$ and $Z$ basis, where $q = \frac{1}{2}d_1$.
\end{proof} 
\end{appendix}

%

\end{document}